\documentclass{article} 
\usepackage{amsthm,amsfonts,amsmath,amscd,amssymb,graphicx,tikz-cd,mathtools,mathrsfs,url,comment} 
\usepackage{authblk}
\usepackage{hyperref}

\newtheorem{theorem}{Theorem}[section]

\newtheorem{lemma}[theorem]{Lemma}
\newtheorem{proposition}[theorem]{Proposition}

\theoremstyle{definition}

\newtheorem*{unnumbered-definition}{Definition}
\numberwithin{equation}{section}
\newtheorem{example}[theorem]{Example}

\newtheorem{remark}[theorem]{Remark}

\def\R{{\mathbb R}}

\def\C{{\mathbb C}}

\def\gog{{\mathfrak g}} \def\tot{{\mathfrak t}}

\def\CH{{\mathcal H}}        
\def\SU{{\rm SU}}\def\U{{\rm U}}\def\SO{{\rm SO}}\def\O{{\rm O}}\def\GL{{\rm GL}}

\newcommand\xqed[1]{%
  \leavevmode\unskip\penalty9999 \hbox{}\nobreak\hfill
  \quad\hbox{#1}}
\newcommand\remdone{\xqed{$\triangle$}}

\begin{document}

\author{Colin McSwiggen \normalsize \\ {\it Division of Applied Mathematics, Brown University} \\ {\it 182 George St., Providence, RI 02906, USA} \\ \url{colin_mcswiggen@alumni.brown.edu}}

\title{The Harish-Chandra integral: \\ An introduction with examples}
\date{}

\maketitle

\abstract{This expository paper introduces the theory of Harish-Chandra integrals, a family of special functions that express the integral of an exponential function over the adjoint orbits of a compact Lie group.  Originally studied in the context of harmonic analysis on Lie algebras, Harish-Chandra integrals now have diverse applications in many areas of mathematics and physics.  We review a number of these applications, present several different proofs of Harish-Chandra's celebrated exact formula for the integrals, and give detailed derivations of the specific integral formulae for all compact classical groups.  These notes are intended for mathematicians and physicists who are familiar with the basics of Lie groups and Lie algebras but who may not be specialists in representation theory or harmonic analysis.}

\setcounter{tocdepth}{1}
\tableofcontents

\pagebreak

\section{Introduction}

\subsection{An integral over a group}

This paper introduces the theory of {\it Harish-Chandra integrals}.  These are functions of the form
\begin{equation} \label{eqn:first-HC-int}
\mathcal{H}(x,y) := \int_G e^{\langle \mathrm{Ad}_g x, y \rangle} dg,
\end{equation}
where $G$ is a compact Lie group, $x$ and $y$ lie in a Cartan subalgebra of $\gog = \mathrm{Lie}(G)$ or its complexification $\gog_\C$, $dg$ is the normalized Haar measure, and $\langle \cdot, \cdot \rangle$ is an Ad-invariant inner product.

Harish-Chandra first studied such integrals in the 1950's with the goal of developing a theory of Fourier analysis on semisimple Lie algebras, but in the intervening decades they have far outgrown their original applications.  They are now ubiquitous special functions that appear in various guises throughout representation theory, geometric analysis, random matrix theory, and physics.  For this reason, they also provide unexpected links between seemingly disparate subjects: since we can study these same functions from many different perspectives, they can help us to translate techniques and ideas between different areas of mathematics, leading to applications in fields as diverse as enumerative geometry, stochastic analysis on Lie algebras, and high-dimensional statistics.  We explain a number of these connections below in Section \ref{sec:HC-everywhere}.

A remarkable fact about the integral (\ref{eqn:first-HC-int}) is that it admits an exact expression as an exponential polynomial.  Harish-Chandra proved the following formula, which holds when $G$ is connected and semisimple \cite{HC}:
\begin{equation} \label{eqn:hc} \Delta_\gog(x) \Delta_\gog(y) \int_G e^{\langle \mathrm{Ad}_g x, y \rangle} dg = \frac{ [ \! [ \Delta_\gog, \Delta_\gog ] \!] }{|W|} \sum_{w \in W} \epsilon(w) e^{\langle w(x),y \rangle}.\end{equation}
Here $x$ and $y$ lie in a Cartan subalgebra of $\gog_\C$, $W$ is the Weyl group, $\epsilon(w)$ is the sign of $w \in W$, $\Delta_\gog(x) := \prod_{\alpha \in \Phi^+} \langle \alpha, x \rangle$ is the product of the positive roots $\Phi^+$ (which we identify with elements of $\gog$ via the inner product), and $[ \! [ \Delta_\gog, \Delta_\gog ] \!]$ is a constant computed below in (\ref{eqn:normconst}).  There are many ways to interpret this formula; in Section \ref{ch:proofs}, we illustrate several of these interpretations with various different proofs of (\ref{eqn:hc}).

The importance of such integrals for mathematical physics was first noted by Itzykson and Zuber \cite{IZ}, who independently discovered the formula for the case of an integral over the unitary group. The unitary integral is now known as the Harish-Chandra--Itzykson--Zuber (HCIZ) integral and has become an important and widely studied identity in quantum field theory, random matrix theory, and algebraic combinatorics.  It is usually written
\begin{equation} \label{eqn:hciz} \int_{\U(N)} e^{\mathrm{tr} (AUBU^\dagger)} dU = \left( \prod_{p=1}^{N-1}p! \right) \frac{\det \big[ e^{a_i b_j} \big]_{i,j = 1}^N}{\Delta(A) \Delta(B)}, \end{equation}
where $\U(N)$ is the group of $N$-by-$N$ unitary matrices, $A$ and $B$ are fixed $N$-by-$N$ diagonal matrices with eigenvalues $a_1 < \hdots < a_N$ and $b_1 < \hdots < b_N$ respectively, and $$\Delta(A) := \prod_{i < j} (a_j - a_i)$$ is the Vandermonde determinant.  Much has already been written about the HCIZ integral in particular, so this article focuses mainly on the more general Harish-Chandra integral (\ref{eqn:first-HC-int}).  For a detailed discussion of the HCIZ integral specifically, see the blog post by Terry Tao \cite{TaoHCIZblog}.

Despite the existence of exact formulae such as (\ref{eqn:hc}) and (\ref{eqn:hciz}), our understanding of these integrals is far from complete.  In fact, for many questions of interest, the exact formulae are no help.  For example, an important research program in random matrix theory has been to characterize the behavior of the HCIZ integral in various regimes as $N \to \infty$; see e.g. \cite{AM, GZ, AG, CGM}.  However, as $N$ grows large the determinant in the formula (\ref{eqn:hciz}) produces factorially many terms of opposite signs, making it a poor starting point for this type of asymptotic analysis.  As of this writing, the large-$N$ behavior of more general integrals of the form (\ref{eqn:first-HC-int}) has still not been rigorously studied.  There remain many other interesting unanswered questions about Harish-Chandra integrals, with a wide scope for new ideas and techniques.

\subsection{Organization of the paper}
In the remainder of this introduction, we briefly recall key definitions from Lie theory and set some notational conventions that will be used in the rest of the article.

In Section \ref{sec:HC-everywhere}, we review a number of ways that Harish-Chandra integrals appear throughout various fields of mathematics and physics, from statistics to quantum gravity.

Section \ref{ch:proofs} collects six different proofs of the formula (\ref{eqn:hc}), each illustrating a distinct perspective on the Harish-Chandra integral: Harish-Chandra's original proof via invariant differential operators, a proof by studying the heat equation on the Lie algebra $\gog$, a proof using a localization technique in symplectic geometry, two representation-theoretic proofs, and a proof using ideas from harmonic analysis.  We also present two further proofs of the HCIZ formula (\ref{eqn:hciz}): one proof via a character expansion and another by induction on the rank of the unitary group.  Sections \ref{sec:hc_heat_eqn_proof}, \ref{subsec:symplectic_proof} and \ref{subsec:KCF-proof} are adapted from material that was previously published by the author in \cite{McS-heateqn}.

Section \ref{ch:specific-integrals} presents detailed derivations of the specific realizations of the integral formula (\ref{eqn:hc}) for all compact classical groups, and discusses how (\ref{eqn:hc}) can be used to compute integrals over arbitrary compact Lie groups that may be neither semisimple nor connected.

This paper is not meant to be a comprehensive reference on Harish-Chandra integrals, and several important topics are omitted.  Notably, we do not include a detailed discussion of large-$N$ asymptotics, for which we refer the reader to \cite{AM, GZ, AGsurvey, CGM}. We also do not discuss correlation functions \cite{Morozov, Sha, PEDZ} or integrals over non-compact groups \cite{Rossmann, PW, FS}, and we only briefly mention Frenkel's generalization of the Harish-Chandra formula to affine Lie algebras \cite{IF}.

\subsection{Definitions from Lie theory} \label{sec:lie-defs}
Here we quickly recall the definitions of several key concepts in Lie theory that will appear continually in what follows.  Detailed introductions to these topics can be found in any of the excellent books by Hall \cite{HallLieGroups}, Procesi \cite{ProcesiLieGroups} or Bump \cite{BumpLieGroups}.

A {\it (real) Lie group} is a finite-dimensional smooth manifold $G$ such that multiplication and inversion in $G$ are smooth maps.  In other words, the map $(g, h) \mapsto g^{-1} h$ is a smooth function from $G \times G$ to $G$.  The {\it Lie algebra} of $G$, written $\gog$ or $\mathrm{Lie}(G)$, is the tangent space to $G$ at the identity element: $\gog := T_{\mathrm{id_G}} G$.  The {\it Lie exponential map} $\exp : \gog \to G$ sends each $x \in \gog$ to the element $e^x := \gamma_x(1) \in G$, where $\gamma_x : \R \to G$ is the unique one-parameter subgroup of $G$ whose tangent vector at the identity is $x$.

The action of $G$ on itself by conjugation fixes the identity element, so that the linearization of this action gives a representation of $G$ on $\gog$, called the {\it adjoint representation}.  Concretely, for each $g \in G$ we obtain an operator $\mathrm{Ad}_g$ on $\gog$ defined by
$$\mathrm{Ad}_g x := \frac{d}{dt} \bigg |_{t=0} g e^{tx} g^{-1} , \qquad x \in \gog.$$
If we then also linearize the map $g \mapsto \mathrm{Ad}_g$, we obtain a representation of $\gog$ on itself, which is also called the adjoint representation.  That is, for each $x \in \gog$ we have an operator $\mathrm{ad}_x$ on $\gog$ defined by
$$\mathrm{ad}_x y := \frac{d}{dt} \bigg |_{t=0} \mathrm{Ad}_{\exp(tx)} y, \qquad y \in \gog.$$
The {\it Lie bracket} $[x, y] := \mathrm{ad}_x y$ is an antisymmetric operation,
$$[x, y] = -[y,x],$$
and also satisfies the {\it Jacobi identity}:
$$[x, [y,z]] + [y,[z,x]] + [z,[x,y]] = 0, \qquad x,y,z \in \gog.$$
In all of the specific examples considered in this article, elements of $G$ and $\gog$ are identified with matrices such that the adjoint representation of $G$ is given by matrix conjugation, $\mathrm{Ad}_g x = g x g^{-1}$, and the Lie bracket is just the matrix commutator, $[x,y] = xy - yx$.

A {\it subalgebra} of $\gog$ is a subspace that is closed under the Lie bracket.  An {\it ideal} of $\gog$ is a subalgebra $\mathfrak{i}$ that is ``contagious'' under the bracket, that is, if $y \in \mathfrak{i}$ then $[x,y] \in \mathfrak{i}$ for all $x \in \gog$.  The algebra $\gog$ is {\it abelian} if $[x, y] = 0$ for all $x,y \in \gog$.  If $\gog$ is nonabelian and has no nonzero proper ideals, it is said to be {\it simple}.  A Lie algebra is {\it semisimple} if it is isomorphic to a direct sum of simple Lie algebras, and is said to be {\it compact} if it is the Lie algebra of a compact Lie group.

The {\it Killing form} on $\gog$ is the bilinear form
$$B(x,y) = \mathrm{tr}(\mathrm{ad}_x \circ \mathrm{ad}_y), \qquad x, y \in \gog,$$
where the trace is taken in the ring of linear operators on $\gog$.  The Killing form is nondegenerate if and only if $\gog$ is semisimple, and is negative semidefinite if $\gog$ is compact.  It is invariant under the adjoint representation, meaning that it satisfies
$$B(\mathrm{Ad}_g x, \mathrm{Ad}_g y) = B(x,y), \qquad g \in G,$$
and also
$$B( \mathrm{ad}_x y, z ) = B( x, \mathrm{ad}_y z ).$$
If $\gog$ is simple, then the Killing form is the unique bilinear form with the above invariance properties, up to a scalar multiple.  If $\gog$ is the Lie algebra of one of the compact classical groups in its defining representation, then the Killing form is a negative multiple of the {\it Hilbert-Schmidt inner product} $\mathrm{tr}(x^\dagger y)$.  Here $x$ and $y$ are regarded as $N$-by-$N$ matrices, and the dagger $\dagger$ indicates the conjugate transpose.  For example, if $\gog = \mathfrak{su}(N)$, we can regard $x, y \in \gog$ as $N$-by-$N$ traceless skew-Hermitian matrices.  Then $B(x,y) = -2N \mathrm{tr}(x^\dagger y).$  Similarly, if $\gog = \mathfrak{so}(N)$ for $N \ge 3$, we can regard $x, y \in \gog$ as $N$-by-$N$ real skew-symmetric matrices, and $B(x,y) = (2-N) \mathrm{tr}(x^T y).$

A {\it complex Lie group} is a Lie group that is also a complex analytic manifold, such that the inversion and multiplication operations are holomorphic maps.  If $G$ is a compact, connected real Lie group, its {\it complexification} $G_\C$ is the unique (up to isomorphism) connected complex Lie group such that $G$ is a maximal compact subgroup of $G_\C$, and such that the Lie algebra of $G_\C$ is $\gog_\C \cong \gog \otimes_\R \C$.  The inclusion $G \hookrightarrow G_\C$ gives an embedding $\gog \hookrightarrow \gog_\C$.

A compact Lie group $G$ admits a unique probability measure, called the {\it Haar probability measure} or {\it normalized Haar measure}, that is invariant under the actions of $G$ on itself by both left and right multiplication.  Additionally, $G$ contains a {\it maximal torus}, that is, an abelian subgroup $T \subset G$ that is diffeomorphic to a torus and is maximal among such subgroups.  Maximal tori are generally non-unique, but they all have the same dimension, and this dimension is called the {\it rank} of $G$ (or of $\gog$).  The Lie algebra $\tot := \mathrm{Lie}(T) \subset \gog$ is called a {\it Cartan subalgebra} of $\gog$; its complexification $\tot_\C \cong \tot \otimes \C$ is also called a Cartan subalgebra of $\gog_\C$.  The Cartan subalgebra $\tot$ (resp.~$\tot_\C$) is a maximal abelian subalgebra of $\gog$ (resp.~$\gog_\C$), and the operators $\mathrm{ad}_h$ for all $h \in \tot_\C$ are simultaneously diagonalizable on $\gog_\C$.  This fact allows us to define the roots of $\gog$, which are linear functionals that encode the eigenvalues of the operators $\mathrm{ad}_h,$ $h \in \tot_\C$.

Concretely, we say that $\alpha \in \tot_\C^*$ is a {\it root of $\gog$ (with respect to $\tot$)} if it is nonzero and satisfies $[h, x] = i \alpha(h) x$ for all $h \in \tot_\C$ and some nonzero $x \in \gog_\C$.  This definition of the roots differs by a factor of $i$ from the definition that is usually used in the study of complex Lie algebras, but it is convenient when studying compact Lie algebras because it makes the roots into real-valued, rather than imaginary-valued, functionals on $\tot \subset \tot_\C$.  We can then use an invariant inner product $\langle \cdot, \cdot \rangle$ on $\gog$ to identify the roots with elements of $\tot$.  Accordingly, we will often write $\langle \alpha, h \rangle$ rather than $\alpha(h)$.

We write $\Phi \subset \tot$ for the collection of all roots of $\gog$, called the {\it root system} of $\gog$ (or of $\gog_\C$).  The root system $\Phi$ is a combinatorial object that encodes a great deal of the structure of $\gog$; in fact, compact and complex simple Lie algebras are completely classified by their root systems.  If $\alpha$ is a root then $-\alpha$ is also a root, so we may choose a subset $\Phi^+ \subset \Phi$, called the {\it positive roots}, such that for each $\alpha \in \Phi$, exactly one of $\alpha$ or $-\alpha$ is in $\Phi^+$.  A positive root $\alpha$ is called a {\it simple root} if it cannot be written as a sum of two positive roots.  For each $\alpha \in \Phi$, the {\it root space}
$$ \gog_\alpha := \{ x \in \gog_\C \ | \ [h, x] = i \alpha(h) x, \ \ \forall \, h \in \tot_\C \} $$
is a one-dimensional complex subspace of $\gog_\C$, and we have the {\it root space decomposition}
$$ \gog_\C = \tot_\C \oplus \bigoplus_{\alpha \in \Phi} \gog_\alpha. $$

For $\alpha \in \Phi^+$, let $s_\alpha$ denote the linear operator on $\tot_\C$ corresponding to reflection through the hyperplane
$$\{ x \in \tot_\C \ | \ \alpha(x) = 0 \}.$$
The {\it Weyl group} $W$ of the root system $\Phi$ (or of the algebra $\gog$ with respect to the Cartan subalgebra $\tot$) is the group generated by the operators $s_\alpha$, $\alpha \in \Phi^+$.  Clearly $W$ acts on both $\tot_\C$ and $\tot$.  The {\it dominant} (or {\it fundamental}, or {\it positive}) {\it Weyl chamber} is the set
$$ \{ x \in \tot \ | \ \langle \alpha, x \rangle > 0, \ \ \forall \, \alpha \in \Phi^+ \}.$$
Its closure is a fundamental domain for the action of $W$ on $\tot$.  For $G$ compact and connected and $T \subset G$ a maximal torus, we have $W \cong N_G(T)/T$, where $N_G(T)$ is the normalizer of $T$ in $G$.  In fact, in the study of connected compact groups, the Weyl group is often defined this way, without reference to the root system $\Phi$.  In this setting, we will also refer to $W$ as the Weyl group of $G$ with respect to $T$.

\subsection{Notation} \label{sec:notation}
Our main notational conventions are as follows.  Throughout the paper, $G$ represents a compact real Lie group of rank $r$.  When we make further assumptions on $G$, such as connectedness or semisimplicity, these will always be made explicit.  We write $\gog$ for the Lie algebra of $G$, $\tot \subset \gog$ for a Cartan subalgebra, and $\gog_\C$ and $\tot_\C$ for their respective complexifications.  We write $W$ for the Weyl group of $\gog$ with respect to $\tot$, which we may identify with a subgroup of $G$ such that $W$ acts on $\tot$ by the adjoint representation.  For $x \in \gog$, $e^x$ or $\exp(x) \in G$ indicates the image of $x$ under the Lie exponential map.  We will often identify $\gog$ with its dual $\gog^*$ via an Ad-invariant inner product $\langle \cdot, \cdot \rangle$, by identifying $x \in \gog$ with the linear functional $\langle x, \cdot \rangle$.  The Harish-Chandra integral will be written $$\CH(x, y) := \int_G e^{\langle \mathrm{Ad}_g x, y \rangle} dg, \qquad x, y \in \tot_\C.$$

We write $\Phi$ for the roots of $\gog$ and $\Phi^+$ for the positive roots.  Since we work more often with the real Lie algebra $\gog$ than with the complexificiation $\gog_\C$, it is convenient to regard the roots as real-valued linear functionals on $\tot$, so that they differ by a factor of $i$ from the ``complex roots'' that are commonly used in the study of complex semisimple Lie algebras.  Specifically, for our purposes, the roots are linear functionals $\alpha \in \tot_\C^*$ satisfying $[h, x] = i \alpha(h) x$ for all $h \in \tot_\C$ and some nonzero $x \in \gog_\C$.  Since each $\alpha \in \Phi$ is real valued on $\tot \subset \tot_\C$, we may regard it as an element of $\tot^*$, and then identify it with an element of $\tot$ via the inner product.  The discriminant of $\gog$ is the function $\Delta_\gog(x) := \prod_{\alpha \in \Phi^+} \langle \alpha, x \rangle$, $x \in \tot$.

If $A$ is a matrix with complex or quaternionic entries, we write $A^\dagger$ for its conjugate transpose and $\bar A$ for its untransposed conjugate.

To keep each section relatively self-contained, key pieces of notation will be reintroduced the first time they appear in a new section.

\section{Harish-Chandra integrals throughout\\ mathematics and physics} \label{sec:HC-everywhere}

Integrals of the form (\ref{eqn:first-HC-int}) arise all over contemporary mathematics: in harmonic analysis, they are a type of generalized Bessel function \cite{AnkerDunklNotes}; in representation theory, they are intimately related to group characters \cite{AK}; and in symplectic geometry, they are Laplace transforms of Duistermaat--Heckman measures of coadjoint orbits \cite{DH}.

The HCIZ integral (\ref{eqn:hciz}) is particularly significant for many reasons.  It initially drew the interest of physicists because it appears in expressions for the partition functions of multi-matrix models in quantum field theory and string theory \cite{IZ, DGZ}.  In random matrix theory, it arises in the joint spectral densities of a number of matrix ensembles, including off-center Wigner matrices and Wishart matrices \cite[ch.~3]{AGsurvey}. Since Wishart matrices model sample covariance matrices, such integrals are used by statisticians in noise models for covariance estimators for large data sets and are also used in signal processing to estimate the channel capacity of multiple-antenna transmission systems; see e.g. \cite{NordenvaadSvensson, GTP}. In integrable systems the HCIZ integral is related to $\tau$-functions for the 2D Toda lattice hierarchy \cite{PZJ}, while in combinatorics and enumerative geometry it is a generating function for the monotone double Hurwitz numbers \cite{GGN}.

In this section, we elaborate on some of these connections to various areas of mathematics and physics.

\subsection{Harmonic analysis}
Harish-Chandra originally studied the function $\CH(x,y)$ for the purpose of developing a theory of Fourier analysis on semisimple Lie algebras \cite{HC}.  In that context, the integral (\ref{eqn:first-HC-int}) plays the role of an {\it elementary spherical function}, which can be thought of as a multivariable generalization of a Bessel function.  In particular, it appears in the integration kernel of an analogue of the radial Fourier transform on $\R^n$.  If $f : \gog \to \C$ is an Ad-invariant Schwartz function, then it is determined by its restriction $\bar f$ to some Cartan subalgebra $\tot \subset \gog$, and we have the {\it radial Fourier transform}
\begin{equation} \label{eqn:dunkl-transf}
\mathscr{R}[\bar f](\xi) := \int_{\tot}  \bar f(x)  \CH(-i \xi, x) \Delta_\gog(x)^2 dx, \qquad \xi \in \tot,
\end{equation}
with the inversion formula
\begin{equation} \label{eqn:dunkl-inversion}
\bar f(x) = c^{-2} \int_{\tot}  \mathscr{R}[\bar f](\xi)  \CH(i \xi, x) \Delta_\gog(\xi)^2 d\xi,
\end{equation}
where
\begin{equation} \label{eqn:dunkl-norm}
c := \int_{\tot} e^{-\frac{|x|^2}{2}} \Delta_\gog(x)^2 dx.
\end{equation}
In fact the radial Fourier transform (\ref{eqn:dunkl-transf}) is a special case of a much more general type of integral transform, the {\it Dunkl transform}, whose kernel incorporates the generalized Bessel functions $\mathcal{B}_{k,\lambda}$ discussed below in Section \ref{sec:HC-and-IS}.  The Dunkl transform includes as special cases the radial Fourier transforms on all Riemannian symmetric spaces of Euclidean type.  For further details, see the notes by Anker \cite{AnkerDunklNotes}.

Another way to think about elementary spherical functions on $\gog$ is as joint eigenfunctions of all $W$-invariant constant-coefficient differential operators on $\tot$.  Lemma \ref{lem:hc15_subfromSH} below, which is an intermediate step in Harish-Chandra's proof of the formula (\ref{eqn:hc}), implies that the integral $\CH(x,y)$ is such a joint eigenfunction.  This is one possible starting point for relating Harish-Chandra integrals to integrable systems, as we discuss below in Section \ref{sec:HC-and-IS}.

For an introduction to the classical Harish-Chandra theory in harmonic analysis and its applications to representation theory, see the book by Varadarajan \cite{Varadarajan}.

\subsection{Representation theory} \label{subsec:repthy-chars}
The characters of finite-dimensional irreducible representations of $G$ can be expressed in terms of Harish-Chandra integrals. Here we assume for simplicity that $G$ is connected and simply connected.  Let $\rho := \frac{1}{2}\sum_{\alpha \in \Phi^+} \alpha$ be half the sum of the positive roots, and define the function $\widehat{\Delta}_\gog$ on the maximal torus $T := \exp(\tot)$ by $$\widehat{\Delta}_\gog(e^x) := \prod_{\alpha \in \Phi^+} (e^{i \langle \alpha, x \rangle/2} - e^{- i\langle \alpha,x \rangle/2}), \qquad x \in \tot.$$  Using the inner product to identify $\tot \cong \tot^*$, we can identify weights of $\gog$ with elements of $\tot$.  For a dominant integral weight $\lambda,$ the Kirillov character formula \cite{Kirillov2, AK} expresses the irreducible character $\chi_\lambda$ as:
\begin{equation} \label{eqn:KCF_intro}
\chi_\lambda(e^x) = \frac{\Delta_\gog(ix)}{\widehat{\Delta}_\gog(e^x)} \frac{\Delta_\gog(\lambda + \rho)}{\Delta_\gog(\rho)} \CH(\lambda + \rho, ix ),
\end{equation}
for $x \in \tot$ with $\Delta_\gog(x) \not = 0$.  We discuss the Kirillov character formula further in Section \ref{subsec:KCF-proof}.

\subsection{Random matrix theory} \label{subsec:RMT}
The HCIZ integral and other analogous orbital integrals are ubiquitous in random matrix theory, because they appear in the joint spectral densities for ensembles of multiple coupled random matrices.  Similarly, they arise when writing the heat kernel on a space of matrices in terms of the eigenvalues.

For example, consider the following {\it two-matrix model}.  Let $A$ and $B$ be two random $N$-by-$N$ Hermitian matrices with joint density
\begin{equation} \label{eqn:2mtx-density}
p(A,B) := \frac{1}{\mathcal{Z}_{N}} \exp \big[ -\mathrm{tr}\, V(A) - \mathrm{tr}\, V(B) + \beta\, \mathrm{tr}(AB) \big],
\end{equation}
where $V(x) := x^2/2 + \kappa x^4/4$, and $\beta, \kappa$ are constants.  The normalization is provided by the {\it partition function} $$\mathcal{Z}_{N} := \iint_{\mathrm{Her}(N)^2} p(A,B) \, dA\, dB,$$ where $\mathrm{Her}(N)$ is the space of $N$-by-$N$ Hermitian matrices equipped with the norm $|A|^2 = \mathrm{tr}(A^2)$.  This type of interacting matrix model was originally studied by Itzykson and Zuber \cite{IZ} as a prototype for non-perturbative approaches to gauge theories, following the program of 't Hooft \cite{Gt}.

Note that the measure $p(A,B) \, dA\, dB$ depends only on the eigenvalues $(a_1, \hdots, a_N)$ of $A$ and $(b_1, \hdots, b_N)$ of $B$.  We can integrate out the angular degrees of freedom in (\ref{eqn:2mtx-density}) to obtain the density $\tilde p(a_1, \hdots, a_N, b_1, \hdots, b_N)$ of this measure with respect to Lebesgue measure on $\R^N \times \R^N$.  The resulting expression gives the joint spectral density in terms of the HCIZ integral:
\begin{multline} \label{eqn:2mtx-density-evs}
\tilde p(a_1, \hdots, a_N, b_1, \hdots, b_N) \\ = \frac{1}{\mathcal{Z}_{N}} \frac{(2\pi)^{N(N-1)}}{\Big( \prod_{j=1}^{N} j! \Big)^2} \Delta(A)^2 \Delta(B)^2 \, e^{ -\mathrm{tr}\, V(A) - \mathrm{tr}\, V(B) } \int_{\U(N)} e^{\beta \, \mathrm{tr}(AUBU^\dagger)} dU,
\end{multline}
where here we may take $A = \mathrm{diag}(a_1, \hdots, a_N)$, $B = \mathrm{diag}(b_1, \hdots, b_N)$. Similarly, we can write the partition function as an integral over the eigenvalues:
\begin{equation} \label{eqn:2mtx-partition}
\mathcal{Z}_{N} = \frac{(2\pi)^{N(N-1)}}{\Big( \prod_{j=1}^{N} j! \Big)^2} \int_{\R^{2N}} \Delta(A)^2 \Delta(B)^2\, e^{-\mathrm{tr}\, V(A) - \mathrm{tr}\, V(B)} \int_{\U(N)} e^{\beta \, \mathrm{tr}(AUBU^\dagger)} dU \prod_{j=1}^N da_j \, db_j.
\end{equation}

There are many other random matrix models whose spectral densities involve the HCIZ integral or its variants.  Among these models, the {\it Wishart matrices} are important for applications to statistics, as they provide a model for sample covariance estimators.  See Section \ref{subsec:statsetc} below, as well as \cite[ch. 7]{PaSh} and \cite[\textsection3.2]{AGsurvey}, for more information on Wishart matrices and statistical applications.  Harish-Chandra integrals can also be used to compute the joint spectral densities for sums of two random matrices with prescribed eigenvalues; see \cite{CMZ, CMZ2}.

\subsection{Symplectic geometry}
The integral $\CH(x,y)$ can also be interpreted in symplectic geometry as the Laplace transform of the Duistermaat--Heckman measure for the action of a maximal torus on a coadjoint orbit.  For definitions and details of these constructions, see Section \ref{subsec:symplectic_proof}.  Here we use the inner product to identify $\gog \cong \gog^*$, so that the adjoint orbit $\mathcal{O}_x$ of $x \in \tot$ is identified with the coadjoint orbit of the linear functional $\langle x, \cdot \rangle$, equipped with the Kostant--Kirillov--Souriau symplectic structure.  Then the orthogonal projection $\phi: \gog \to \tot$ is a moment map for the action of the maximal torus $\mathrm{exp}(\tot)$ on $\mathcal{O}_x$, and we have
\begin{equation} \label{eqn:symplectic-interp}
\CH(x,y) = \frac{1}{\mathrm{Vol}_\mu(\mathcal{O}_{x})} \int_{\mathcal{O}_x} e^{\langle \phi(\beta), y \rangle} d\mu(\beta),
\end{equation}
where $\mu$ is the Liouville measure on the coadjoint orbit and $\mathrm{Vol}_\mu(\mathcal{O}_{x})$ is its Liouville volume.

\subsection{Integrable systems} \label{sec:HC-and-IS}
For $\alpha \in \Phi^+$, let $s_\alpha$ denote the reflection through the hyperplane $$H_\alpha := \{ x \in \tot \ | \ \langle \alpha, x \rangle = 0 \}.$$  A {\it multiplicity parameter} $k$ is a complex-valued function on the roots of $\gog$ that is constant on Weyl orbits, i.e. $k_\alpha = k_{w \alpha}$, $w \in W$.  Given such a multiplicity parameter, the {\it Dunkl Laplacian} $\mathcal{L}_k$ is a differential-difference operator acting on twice-differentiable functions $f : \tot \to \C$ by
\begin{equation} \label{eqn:dunkl-laplacian-def}
\mathcal{L}_k f(x) = \sum_{j = 1}^r \partial_j^2 f(x) + \sum_{\alpha \in \Phi^+} \frac{2 k_\alpha}{\langle \alpha, x \rangle} \partial_\alpha f(x) - \sum_{\alpha \in \Phi^+} \frac{k_\alpha |\alpha|^2}{\langle \alpha, x \rangle^2} \big[ f(x) - f(s_\alpha x) \big],
\end{equation}
where $\partial_\alpha f(x) := \frac{d}{dt} f(x + t \alpha) \big |_{t = 0}$, and $\partial_j$, $j = 1, \hdots, r$ are derivatives along the coordinate directions with respect to some orthonormal basis of $\tot$.  We then define an operator $L_k$ on $\tot$ by
$$L_k f = \delta_k^{-1} \mathcal{L}_k^W (\delta_k f),$$
where $$\mathcal{L}_k^W := \sum_{j = 1}^r \partial_j^2 + \sum_{\alpha \in \Phi^+} \frac{2 k_\alpha}{\langle \alpha, x \rangle} \partial_\alpha$$ is the $W$-invariant part of $\mathcal{L}_k$, and $$\delta_k(x) := \prod_{\alpha \in \Phi^+} \langle \alpha, x \rangle^{k_\alpha}.$$

We can consider the operator $-L_k$ as a Schr\"odinger operator for a quantum mechanical system of $r$ interacting particles, the {\it quantum Calogero--Moser system} associated to the root system of $\gog$.  In this interpretation, the values $k_\alpha$ of the multiplicity parameter determine the coupling constants for the interactions between the particles.  This system is completely integrable, in the sense that the algebra of $W$-invariant differential operators commuting with $L_k$ is isomorphic to a polynomial algebra of rank $r$.  For generic multiplicity parameters $k$ and any $\lambda \in \tot$, one can define the {\it generalized Bessel function} $\mathcal{B}_{k, \lambda}$, which is a distinguished $W$-invariant eigenfunction of $L_k$ with eigenvalue $|\lambda|^2$, normalized so that $\mathcal{B}_{k, \lambda}(0) = 1$, and vanishing on the hyperplane $H_\alpha$ for each $\alpha \in \Phi^+$ with $k_\alpha \not = 0$.  The generalized Bessel functions play the role of energy eigenfunctions for the quantum Calogero--Moser system.

Harish-Chandra integrals are generalized Bessel functions for the multiplicity parameter with all $k_\alpha = 1$.  In this case, $\mathcal{L}^W_k$ is the radial part of the Laplacian on $\gog$ (defined below in Section \ref{sec:radparts}), and $\mathcal{B}_{k, \lambda}(x) = \CH(\lambda, x)$.  For details on the above constructions, we refer the reader to the review articles \cite{EtingofCM, AnkerDunklNotes, OpdamHGFreview}.

There are other known connections between Harish-Chandra integrals and integrable systems.  Notably, Zinn-Justin has related the large-$N$ asymptotics of the HCIZ integral to the dispersionless 2D Toda lattice hierarchy \cite{PZJ}.  In Section \ref{subsec:cm_relationship} we discuss a different relationship between Harish-Chandra integrals and Calogero--Moser systems, distinct from the considerations described above, which arises in a semiclassical regime via a WKB-type ansatz.

\subsection{Combinatorics and enumerative geometry}
The HCIZ integral is a combinatorial generating function for the {\it genus g monotone double Hurwitz numbers} $\vec H_g(\alpha, \beta)$ \cite{GGN}, which we now define.  These numbers are labeled by a nonnegative integer $g$ and two partitions $\alpha, \beta$ of another nonnegative integer $d$.  They count certain walks on the Cayley graph of the symmetric group $S_d$ as generated by transpositions.

A walk on this Cayley graph is identified with a sequence $$\big( \, (s_1 \ t_1)\, ,\, \hdots\, , (s_m \ t_m) \, \big)$$ of transpositions $(s_i \ t_i)$, where $s_i, t_i \in \{1, \hdots, d \}$ and $s_i < t_i$.  We say that the walk is {\it monotone} if the numbers $t_i$ form a weakly increasing sequence.  The number $\vec H_g(\alpha, \beta)$ is defined as the number of monotone walks on the Cayley graph of $S_d$ that meet the following criteria:
\begin{enumerate}
\item The walk consists of $2g - 2 + \ell(\alpha) + \ell(\beta)$ steps, where $\ell(\cdot)$ indicates the length of a partition.
\item The walk begins in the conjugacy class labeled by the partition $\alpha$ and ends in the conjugacy class labeled by the partition $\beta$.
\item The walk's endpoints and steps together generate a transitive subgroup of $S_d$.
\end{enumerate}

These numbers also have a geometric interpretation.  If we remove the requirement that the walks be monotone, we instead get the ordinary double Hurwitz numbers $H_g(\alpha, \beta)$.  By a classical result of Hurwitz \cite{Hurwitz}, $H_g(\alpha, \beta)/d!$ gives a weighted count of certain genus $g$ branched covers of the Riemann sphere, with branching data determined by the partitions $\alpha$ and $\beta$.  Accordingly, we can interpret $\vec H_g(\alpha, \beta)/d!$ as counting a subset of these branched covers satisfying the combinatorial constraint imposed by monotonicity of the corresponding walk on the Cayley graph.

Goulden, Guay-Paquet and Novak proved the following expansion, leading to an interpretation of the HCIZ integral as a generating function for the numbers $\vec H_g(\alpha, \beta)$ \cite{GGN}: 
\begin{multline} \label{eqn:hurwitz-expansion}
\frac{1}{N^2} \log \int_{\U(N)} e^{zN \, \mathrm{tr}(AUBU^\dagger)} dU = \\ \sum_{d=1}^N \frac{z^d}{d!} \sum_{g \ge 0} \left(\frac{1}{N^2}\right)^g \sum_{\alpha, \beta \vdash d} \left(\frac{-1}{N}\right)^{\ell(\alpha) + \ell(\beta)} \prod_{i = 1}^{\ell(\alpha)} \mathrm{tr}(A^{\alpha_i}) \prod_{i = 1}^{\ell(\beta)} \mathrm{tr}(B^{\beta_i}) \vec H_g(\alpha, \beta) \\ + O(z^{N+1}),
\end{multline}
where $z$ is a complex parameter.  Subsequently, Novak used the expansion (\ref{eqn:hurwitz-expansion}) to give a new proof of a theorem on the asymptotic distribution of vertical tiles in a uniform random lozenge tiling \cite{JonLozenge}.  It is an interesting open question whether there are analogous Hurwitz-theoretic expansions for other Harish-Chandra integrals.

There are several further combinatorial applications of Harish-Chandra integrals.  For example, they arise in characterizing inequalities for a family of generalized majorization orders \cite{MN-majorization}.  Coquereaux, Zuber and the author have also used Harish-Chandra integrals to study the volumes of Berenstein--Zelevinsky polytopes and the tensor product multiplicities of semisimple Lie algebras \cite{CZ1, CMZ, McS-deconvolution}.

\subsection{Theoretical physics}

Orbital integrals such as the HCIZ integral were originally studied by physicists in the context of so-called {\it multi-matrix models} \cite{IZ}, which are random matrix ensembles that can be thought of as gauge field theories in a simplified setting where spacetime consists of finitely many points.  Each random matrix in the ensemble then represents the value of the gauge field at one point in spacetime.  As we saw from the example of the two-matrix model in Section \ref{subsec:RMT}, orbital integrals appear in the partition functions of such models due to the coupling between the matrices.

In the {\it large-$N$ limit} (as the size of the matrices goes to infinity), 't Hooft observed that gauge theories can be described in terms of Feynman diagrams that are represented by planar graphs \cite{Gt}. Following this idea, Br\'ezin, Itzykson, Parisi and Zuber used a model of a single random matrix to study the combinatorics of planar maps \cite{BIPZ},\footnote{A {\it map} is a graph embedded in a surface of minimal genus.} initiating a stream of research relating matrix integrals to map enumeration; see \cite{Zvonkin} for a review.  Itzykson and Zuber then considered the case of two random matrices \cite{IZ}, and it is in this context that they derived the HCIZ formula (\ref{eqn:hciz}).  The work of Matytsin \cite{AM} on the large-$N$ limit of the HCIZ integral was also motivated by gauge theory, specifically by a certain model of the strong nuclear force introduced by Kazakov and Migdal \cite{KazMig}.

A map can be regarded as a discretization of the underlying surface.  This point of view provides a link between random matrix theory and discrete random geometry, leading to many applications in two-dimensional quantum gravity; see \cite{DGZ} for a review.  In particular, multi-matrix models are related to the combinatorics of {\it colored} maps, and they correspond to theories of a random function on a random surface (``statistical mechanics coupled to 2D gravity,'' in physics parlance).  From this perspective, the two-matrix model (\ref{eqn:2mtx-density}) describes the Ising model on a certain family of random graphs.  By studying the large-$N$ limit of this model, Boulatov and Kazakov were able to derive a number of exact results for the Ising model on random {\it planar} graphs \cite{Kaz, BouKaz}.

Aside from gauge theories and quantum gravity, the HCIZ integral has further applications in other areas of physics where random matrices play a role, perhaps most notably in quantum chaos and disordered systems \cite{KM, FK, BBJ}.

\subsection{Statistics and signal processing} \label{subsec:statsetc}
{\it Wishart matrices} are a model of a statistical estimator for the covariance matrix of a random vector. Their introduction by Wishart \cite{Wishart} is widely regarded as the historical genesis of random matrix theory.  Multiple definitions exist in the literature, but here we consider matrices of the form
\begin{equation} \label{eqn:wishart-def}
Y_{n,m} = X_{n,m} X_{n,m}^\dagger,
\end{equation}
where $X_{n,m}$ is an $n$-by-$m$ matrix whose columns are i.i.d.~$n$-variate centered real or complex Gaussian vectors.  The matrix $Y_{n,m}$ then has the distribution of a sample covariance matrix constructed from the columns of $X_{n,m}$.

The joint spectral density of a Wishart matrix can be expressed in terms of the HCIZ integral in the complex case, or in terms of an analogous integral over the orthogonal group in the real case; see \cite[\textsection3.2]{AGsurvey} for details.  This has led to applications of the HCIZ formula in techniques for covariance estimation; see e.g. \cite{NordenvaadSvensson, MN, BABP}.

The HCIZ formula has also been widely applied in signal processing for analyzing the performance of multiple-input multiple-output (MIMO) antenna arrays \cite{GTP, ACNZ, DKOPY, ZHK}.  In fact, many information-theoretic problems in the study of multiple-antenna transmission systems turn out to be closely related to covariance estimation: the channel capacity of a MIMO system is typically studied via the moment generating function for the mutual information between the transmitter and the receiver, which in turn is derived from the joint spectral density for a complex Wishart matrix.

\section{Proofs of Harish-Chandra's formula}
\label{ch:proofs}

This section collects many different proofs of Harish-Chandra's formula (\ref{eqn:hc}) and the HCIZ formula (\ref{eqn:hciz}).  To the author's knowledge, these proofs represent all currently known derivations of these two formulae.  The purpose of presenting such a wide variety of arguments is not just to gather them all in one place, but also to illustrate the diverse interpretations of Harish-Chandra integrals in many different areas of mathematics.

In total, we give six proofs of (\ref{eqn:hc}): Harish-Chandra's original proof, based on a simultaneous diagonalizability argument for invariant differential operators (Section \ref{sec:proof}); a proof by relating heat flow on a semisimple Lie algebra to heat flow on a Cartan subalgebra (Section \ref{sec:hc_heat_eqn_proof}); a proof via the Duistermaat--Heckman theorem in symplectic geometry (Section \ref{subsec:symplectic_proof}); two representation-theoretic proofs, one using the Kirillov character formula (Section \ref{subsec:KCF-proof}) and another using a character expansion for the heat kernel on $G$ (Section \ref{subsec:char-exp-proof}); and, finally, a harmonic analysis proof using Rossman's formula for the Fourier transform on a semisimple Lie algebra (Section \ref{subsec:harmonic-proof}).  We also give two further proofs of the HCIZ formula (\ref{eqn:hciz}), which are somewhat specific to the unitary case (Section \ref{subsec:further-hciz-proofs}).

In Sections \ref{sec:statement} and \ref{sec:background}, we introduce notation, formally state the integral formula as a theorem, and review some facts about invariant differential operators that we will need for the proofs.  We also situate the integral formula in the context of the paper \cite{HC} in which it was originally published.  Sections \ref{sec:hc_heat_eqn_proof}, \ref{subsec:symplectic_proof} and \ref{subsec:KCF-proof} are adapted from parts of the paper \cite{McS-heateqn}.

\subsection{Preliminaries and statement of the theorem}
\label{sec:statement}

Before explaining Harish-Chandra's formula, we first have to fix a substantial amount of notation.  Let $G$ be a compact, connected, semisimple, real Lie group of rank $r$ with Lie algebra $\mathfrak{g}$ and normalized Haar measure $dg$.  Let $\tot$ be a Cartan subalgebra of $\mathfrak{g}$, and $\mathfrak{g}_\C := \mathfrak{g} \otimes \mathbb{C}$, $\tot_\C := \tot \otimes \mathbb{C}$ be the complexifications of $\mathfrak{g}$ and $\tot$.

Let $W$ be the Weyl group of $\mathfrak{g}$ with respect to $\tot$.  Then $W$ acts on $\tot$ as a group of linear transformations generated by reflections, and for each $w \in W$ we denote by $\epsilon(w)$ the sign of $w$, that is $\epsilon(w) = (-1)^{|w|}$ where $|w|$ is the number of reflections required to generate $w$.

Let $\langle \cdot, \cdot \rangle: \mathfrak{g} \times \mathfrak{g} \to \R$ be an $\mathrm{Ad}$-invariant inner product on $\mathfrak{g}$, which we extend linearly to a complex-valued form on $\mathfrak{g}_\C$.  In all of the specific examples that we consider, elements of $\gog$ will be identified with square matrices, and we will take $\langle x, y \rangle$ to be the {\it Hilbert--Schmidt inner product} $\mathrm{tr}(x^\dagger y)$, also known as the trace form.  The inner product $\langle \cdot, \cdot \rangle$ induces an isomorphism $\mathfrak{g} \to \mathfrak{g}^*$ by $x \mapsto \langle x, \cdot \rangle$.  Let $n = \dim \mathfrak{g}$.  Given a basis $\{e_1, \hdots, e_n \}$ of $\mathfrak{g}$, we can define coordinate functions in the usual way by setting $x_i = \langle e_i, x \rangle$, so that we may write $x = (x_1, \hdots, x_n)$.  We can then identify $\mathfrak{g}$ with the real subspace of all $x \in \mathfrak{g}_\C$ that have strictly real coordinates.  We use the standard multi-index notation for monomials and for partial derivatives with respect to the coordinates in the chosen basis:
$$x^\beta := x_1^{\beta_1} \cdots x_n^{\beta_n}, \qquad \qquad \frac{\partial^{|\beta|}}{\partial x^\beta} := \frac{\partial^{|\beta|}}{\partial x_1^{\beta_1} \cdots \partial x_n^{\beta_n}},$$ where $\beta = (\beta_1, \hdots, \beta_n)$ is a multi-index with $n$ components.

If $V$ is a vector space over $\C$ or $\R$, we write $\Pi(V)$ for the algebra of polynomial functions on $V$ with complex coefficients.  We identify polynomials on $\mathfrak{g}_\C$ or $\tot_\C$ with their restrictions to $\mathfrak{g}$ or $\tot$ respectively, so that $$\Pi(\gog) \cong \Pi(\gog_\C) \cong \C[x_1, \hdots, x_n],$$ where $x_1, \hdots, x_n$ are the coordinate functions defined previously. Given a polynomial $p(x) = \sum_\beta c_\beta x^\beta \in \Pi(\mathfrak{g})$ where $\beta$ ranges over multi-indices, denote by $p(\partial)$ the differential operator $$p(\partial) = \sum_\beta c_\beta \frac{\partial^{|\beta|}}{\partial x^\beta}.$$ We can extend $\langle \cdot, \cdot \rangle$ to a scalar product $[ \! [ \cdot, \cdot ] \!]$ on $\Pi(\mathfrak{g})$ by defining 
\begin{equation} \label{eqn:poly-bracket-def}
[ \! [ p, q ] \!] = p(\partial)q(x) \big |_{x=0}.
\end{equation}
If $\{e_1, \hdots, e_n\}$ is an orthonormal basis of $\mathfrak{g}$ with respect to $\langle \cdot, \cdot \rangle$, then an orthonormal basis for $\Pi(\mathfrak{g})$ with respect to $[ \! [ \cdot , \cdot ] \!]$ is given by monomials of the form $(\prod_i e_i^{\beta_i})/\sqrt{\beta!}$, where $\beta$ is a multi-index and the multi-index factorial has the usual meaning $\beta ! := \beta_1 ! \hdots \beta_n!$.  Lemma \ref{lem:normconst} below shows how to compute $[ \! [ p , q ] \!]$ in terms of the coefficients of $p$ and $q$.  One can then easily verify that $[ \! [ \cdot , \cdot ] \!]$ is symmetric and nondegenerate, and that $[ \! [ x , y ] \!] = \langle x , y \rangle$ for $x, y \in \mathfrak{g}$.

Finally, let $\Phi^+$ be a choice of positive roots of $\mathfrak{g}$.  The {\it discriminant} of $\mathfrak{g}$ is the homogeneous polynomial $\Delta_\gog: \tot \to \mathbb{C}$ given by taking the product of the positive roots: \begin{equation} \label{eqn:pidef} \Delta_\gog(x) := \prod_{\alpha \in \Phi^+} \langle \alpha, x \rangle.\end{equation}  The discriminant plays an important role in geometric analysis on $\mathfrak{g}$. One of its essential properties is that it is {\it skew} with respect to the action of $W$: for $w \in W$, $\Delta_\gog(w(x)) = \epsilon(w) \Delta_\gog(x)$.  This follows from the fact that if $\alpha$ is a simple root, the reflection through the plane $\{ \alpha = 0 \}$ sends $\alpha \mapsto -\alpha$ and permutes the other positive roots.

With the above definitions, Harish-Chandra's formula states:

\begin{theorem} \label{thm:hc} For all $x, y \in \tot_\C$,
\begin{equation} \tag{\ref{eqn:hc}}
\Delta_\gog(x) \Delta_\gog(y) \int_G e^{\langle \mathrm{Ad}_g x, y \rangle} dg = \frac{ [ \! [ \Delta_\gog, \Delta_\gog ] \!] }{|W|} \sum_{w \in W} \epsilon(w) e^{\langle w(x),y \rangle}.
\end{equation}
\end{theorem}

On the left-hand side of this equation, we have an integral over the Lie group $G$.  On the right-hand side, we have a finite sum over the Weyl group $W$.  Since $G$ is compact, the action of $W$ on $\tot_\C$ is represented by a finite subgroup of $G$ acting by the adjoint representation, so that (\ref{eqn:hc}) has the interpretation that the integral on the left is equal to a normalized, alternating sum of the integrand's values at finitely many points.  In other words, the integral is {\it localized} at the points of $G$ that represent elements of $W$.

The following lemma shows how to evaluate $[ \! [ p, q ] \!]$ for two polynomials $p,q$.

\begin{lemma} \label{lem:normconst}
Let $$p(x) = \sum_{\beta} p_\beta x^\beta \quad \textrm{and} \quad q(x) = \sum_{\beta} q_\beta x^\beta$$ be two polynomial functions on $\mathfrak{g}$, where $(x_1, \hdots, x_n)$ are coordinates in an orthogonal basis and $\beta$ runs over multi-indices.  Then $$p(\partial) q(x) \big |_{x=0} = \sum_{\beta} p_\beta q_\beta \beta!$$
\end{lemma}

Note that since only finitely many $p_\beta$ and $q_\beta$ are nonzero, the sum is finite.

\begin{proof}[Proof]
Expanding out terms, we have $$p(\partial) q(x) = \sum_{\alpha, \beta} p_\alpha q_\beta \frac{\partial^{|\alpha|} x^\beta}{\partial x^\alpha}.$$  If $|\alpha | \ge |\beta|$ and $\alpha \not = \beta$, then ${\partial^{|\alpha|} x^\beta}/{\partial x^\alpha} = 0$. If $|\alpha | < |\beta|$, then ${\partial^{|\alpha|} x^\beta}/{\partial x^\alpha}$ has positive degree and is killed by evaluating at $x = 0$, so that we are left only with terms where $\alpha = \beta$, and the sum becomes $$p(\partial) q(x) \big |_{x=0} = \sum_{\beta} p_\beta q_\beta \frac{\partial^{|\beta|} x^\beta}{\partial x^\beta} = \sum_{\beta} p_\beta q_\beta \beta!$$ as desired.
\end{proof}

This lemma allows us to compute $[\![ \Delta_\gog, \Delta_\gog ] \!]$ from the coefficients of $\Delta_\gog$.  Writing $$\Delta_\gog(x) = \sum_{|\beta| = |\Phi^+|} \pi_\beta x^\beta$$ for some constants $\pi_\beta$, where $|\Phi^+|$ is the number of positive roots, we have \begin{equation} \label{eqn:normconst} [\![ \Delta_\gog, \Delta_\gog ] \!] = \left( \sum_{|\beta| = |\Phi^+|} \pi_\beta \frac{\partial^{|\beta|}}{\partial x^\beta} \right) \left( \sum_{|\beta| = |\Phi^+|} \pi_\beta x^\beta \right) = \sum_{|\beta| = |\Phi^+|} \pi_\beta^2 \beta ! \end{equation}

\subsection{Background and context of the theorem}
\label{sec:background}

To motivate the formula (\ref{eqn:hc}), it's helpful to understand what Harish-Chandra was trying to do more broadly in the paper \cite{HC} in which it first appeared.  That article is primarily concerned with describing the algebras of differential operators on {\it noncompact} Lie algebras.  In fact, Harish-Chandra states at the outset that ``When [the algebra] is compact, this theory is not difficult ... although some of the results obtained here ... seem to be new.''  Nevertheless, the compact case clearly illustrates the central theme of the paper, which addresses the relationship between differential operators on a Lie algebra and on a Cartan subalgebra.  In particular, for functions $f \in C^\infty(\mathfrak{g})$ that are invariant under the adjoint action of $G$ on $\mathfrak{g}$, Harish-Chandra's explicit goal is to express $\overline{p(\partial) f}$ in terms of $\bar p(\partial)$ and $\bar f$, where the bar indicates restriction to $\tot$.  What he finds is that \begin{equation} \label{eqn:deriv_restr} \Delta_\gog \ \overline{p(\partial) f} = \bar p(\partial) (\Delta_\gog \bar f) \end{equation} where $\Delta_\gog$ is the discriminant defined in (\ref{eqn:pidef}).  This result anticipates the modern theory of radial parts of differential operators, discussed below in Section \ref{sec:radparts}, and appears in hindsight as a foretaste of ideas in spherical harmonic analysis that led to many later developments in the theory of special functions, such as Dunkl theory \cite{AnkerDunklNotes}.

Many of the results in \cite{HC} can be viewed as building on the Chevalley restriction theorem (Theorem \ref{thm:chev} below), which says that the $G$-invariant polynomials on $\mathfrak{g}$ are isomorphic to the $W$-invariant polynomials on $\tot$, and that an isomorphism is given in the $\mathfrak{g}$-to-$\tot$ direction simply by restriction of functions.  In effect, Harish-Chandra expands the scope of this theorem, showing that the isomorphism for invariant functions extends to a homomorphism from a particular algebra of $G$-invariant differential operators on $\mathfrak{g}$ to the algebra of all $W$-invariant differential operators on $\tot$.  He then elaborates a number of consequences including the relation (\ref{eqn:deriv_restr}) and the integral formula (\ref{eqn:hc}).

In the rest of this section we will make these ideas concrete, starting with a formal statement of the Chevalley restriction theorem.  The group $G$ acts on functions $f: \mathfrak{g} \to \mathbb{C}$ by sending $f(x) \mapsto f(\mathrm{Ad}_{g^{-1}} x)$ for each $g \in G$.  Likewise, the Weyl group $W$ acts on functions $v: \tot \to \mathbb{C}$ by $v(x) \mapsto v(w^{-1}x)$ for $w \in W$.  Let $I(\mathfrak{g}) \subset \Pi(\mathfrak{g})$ be the space of polynomials on $\mathfrak{g}$ that are invariant under the adjoint action of $G$, and let $I(\tot) \subset \Pi(\tot)$ be the space of polynomials on $\tot$ that are invariant under the action of $W$.  Then we have:

\begin{theorem}[Chevalley restriction theorem] \label{thm:chev}
For all $p \in I(\mathfrak{g})$, the restriction $\bar p$ belongs to $I(\tot)$.  Moreover, $p \mapsto \bar p$ is an isomorphism of $I(\mathfrak{g})$ onto $I(\tot)$.
\end{theorem}

\begin{example} \label{ex:chevalley}
In the case $G = \U(N)$, we have $\mathfrak{g} = \mathfrak{u}(N)$, the space of $N$-by-$N$ skew-Hermitian matrices.  Multiplying by $i$, we can identify $\mathfrak{u}(N)$ with the space of $N$-by-$N$ {\it Hermitian} matrices.  In this setting, the Chevalley restriction theorem tells us that if $p$ is a polynomial in the entries of a Hermitian matrix $M$, then we have $p(M) = p(UMU^\dagger)$ for all $U \in \U(N)$ if and only if $p$ can be written as a symmetric polynomial in the eigenvalues of $M$.
\end{example}

In order to explain how Harish-Chandra extends Theorem \ref{thm:chev} to differential operators, we must introduce some further ideas from geometric analysis.

\subsubsection{Generalities on differential operators}
First we must clarify what exactly we mean by a ``differential operator.''  For a finite-dimensional vector space $V$ over $\mathbb{C}$ or $\mathbb{R}$, let $$\partial \Pi(V) := \{ p(\partial) \textrm{ } | \textrm{ } p \in \Pi(V) \}$$ be the algebra of constant-coefficient differential operators on $V$.  We also regard $\Pi(V)$ itself as an algebra of multiplication operators on $V$, so that $p \in \Pi(V)$ acts by $f \mapsto pf$. The algebra $\mathcal{D}(V)$ generated by both $\Pi(V)$ and $\partial \Pi(V)$ is the algebra of {\it polynomial differential operators} on $V$, and these are the operators that we will primarily study.  By applying the product rule, it is always possible to write any element of $\mathcal{D}(V)$ in the form $\sum_i p_i \cdot q_i(\partial)$, where $p_i$ and $q_i$ are polynomials.\footnotemark  \footnotetext{Here the dot ``$\cdot$'' indicates multiplication but also emphasizes that $p_i$ acts as a multiplication operator {\it following} differentiation by $q_i(\partial)$, as opposed to differentiation by $(p_i q_i)(\partial) \in \partial \Pi(V)$.  We will sometimes also use the notation $p(\partial) \circ q$ to indicate the composition of operators, i.e. $(p(\partial) \circ q)f = p(\partial)(qf)$.}

More generally, if $U \subseteq V$ is a nonempty open set, then we can define a {\it differential operator on $U$} to be any operator acting on $C^\infty(U)$ that has the form $\sum_{i=1}^n a_i \cdot q_i(\partial)$ where $a_i \in C^\infty(U)$ and $q_i \in \Pi(V)$.  For example, in this case the coefficients $a_i$ could blow up at the boundary of $U$ or could be rational functions with no poles in the interior of $U$.  When we use the term ``differential operator'' with no further qualification, we will mean an operator of this form. Such operators form an algebra $\mathfrak{D}(U)$, and there is a natural inclusion of $\mathcal{D}(V)$ as a subalgebra of $\mathfrak{D}(U)$ given by $p_i \cdot q_i(\partial) \mapsto p_i |_U \cdot q_i(\partial)$.

Let $W \subset V$ be a subspace and let $(w_1, \hdots, w_k)$ be linear coordinates on $W$. Then we can identify $\mathcal{D}(W)$ with the subalgebra of $\mathcal{D}(V)$ generated by 1, $\{ w_i \}_{i=1}^k$, and $\{ w_i (\partial) \}_{i=1}^k$. Thus each $D \in \mathcal{D}(W)$ can be thought of as a differential operator either on $W$ or on $V$.  We will ignore this distinction with impunity, since for any $f \in C^\infty(V)$ we have $(Df) |_W = D(f|_W)$.  Observe also that if $V$ is a vector space over $\R$, then there is a natural correspondence between polynomial differential operators on $V$ and on the complexification $V \otimes \C$, obtained by identifying 
\begin{equation} \label{eqn:diffop-identify-cx-real}
\frac{\partial}{\partial x_j} \ \ \longleftrightarrow \ \ \frac{\partial}{\partial z_j} = \frac{\partial}{\partial x_j} - i \frac{\partial}{\partial y_j}
\end{equation}
where $x_j$, $y_j$ are real coordinates on $V$ and $z_j = x_j + i y_j$ is a complex coordinate on $V \otimes \C$.

\subsubsection{Group actions on differential operators}
We now return to the case where the underlying vector space is one of the Lie algebras $\mathfrak{g}$ or $\tot$. There is a natural way in which $G$ acts on $\mathcal{D}(\mathfrak{g})$ and $W$ acts on $\mathcal{D}(\tot)$, extending the respective actions on $\Pi(\mathfrak{g})$ and $\Pi(\tot)$.  We define a $G$-action on $\partial\Pi(\mathfrak{g})$ by stipulating that differential operators transform in the opposite way to functions: under the map $x \mapsto \mathrm{Ad}_{g^{-1}} x$, we have $$\frac{\partial}{\partial x} \ \mapsto \ \frac{\partial}{\partial (\mathrm{Ad}_{g^{-1}} x)} = (\mathrm{Ad}_g x) (\partial),$$ so that $p(\partial) \mapsto (p \circ \mathrm{Ad}_g )(\partial)$, where the circle ``$\, \circ \,$'' indicates composition of functions rather than operators. This choice guarantees that the actions on $\Pi(\mathfrak{g})$ and on $\partial \Pi(\mathfrak{g})$ are compatible. Since every element of $\mathcal{D}(\mathfrak{g})$ can be written as $\sum_i p_i \cdot q_i(\partial)$, the action of $G$ on $\mathcal{D}(\mathfrak{g})$ is fully determined by the actions on $\Pi(\mathfrak{g})$ and $\partial\Pi(\mathfrak{g})$: an element $g \in G$ sends $$p \cdot q(\partial) \mapsto (p \circ \mathrm{Ad}_{g^{-1}}) \cdot (q \circ \mathrm{Ad}_{g})(\partial)\, ,$$ and this action extends to all of $\mathcal{D}(\mathfrak{g})$ by linearity.  The action of $W$ on $\mathcal{D}(\tot)$ is defined analogously.  Note that the algebra homomorphism $\mathcal{D}(\mathfrak{g}) \to \mathcal{D}(\mathfrak{g})$ thus induced by any $g \in G$ is inverted by $g^{-1}$, so that for each $g \in G$ or $w \in W$ we obtain an automorphism of $\mathcal{D}(\mathfrak{g})$ or $\mathcal{D}(\tot)$ respectively.

In what follows, we will refer to several algebras of {\it invariant} differential operators.  Namely, let $\mathcal{I}'$ denote those elements of $\mathcal{D}(\mathfrak{g})$ that are invariant under the action of $G$, and let $\mathcal{I}(\mathfrak{g})$ be the subalgebra of $\mathcal{I}'$ generated by both $I(\mathfrak{g})$ and $\partial I(\mathfrak{g}) $, where $\partial I (\mathfrak{g}) := \{ p(\partial) \textrm{ } |\textrm{ } p \in I(\mathfrak{g})\}$.  Let $\mathcal{I}(\tot)$ denote those elements of $\mathcal{D}(\tot)$ that are invariant under the action of $W$.  Note that by the correspondence (\ref{eqn:diffop-identify-cx-real}), we can equivalently consider these as algebras of invariant differential operators on the complexifications $\gog_\C$ and $\tot_\C$.

We record for later use the following important property of the discriminant $\Delta_\gog$ \cite[Corollary to Lemma 10]{HC}.
\begin{proposition} \label{prop:pi_harmonic}
$D\Delta_\gog = 0$ for all $D \in \mathcal{I}(\tot)$ that annihilate the constants.
\end{proposition}
\begin{proof}
Since $D$ is $W$-invariant and $\Delta_\gog$ is skew, $D\Delta_\gog$ must be skew.  For each reflection $w_\alpha \in W$ we thus have $w_\alpha(D\Delta_\gog) = -D\Delta_\gog$, which implies that $D\Delta_\gog$ vanishes on the hyperplane $\{\alpha = 0 \}$.  Therefore the linear functional $\alpha$, regarded as a polynomial of degree 1, divides $D\Delta_\gog$.  Now, it is a basic fact of Lie theory that the root system of a semisimple Lie algebra is {\it reduced}, meaning that $\Phi$ contains no scalar multiples of $\alpha$ except for $\alpha$ itself and $-\alpha$.  This implies that no positive root divides any other positive root, since a polynomial of degree 1 divides another polynomial of degree 1 if and only if the two polynomials are scalar multiples of each other.  In other words, as polynomials, the positive roots are relatively prime.  Therefore the fact that each $\alpha$ divides $D\Delta_\gog$ implies that the entire discriminant $\Delta_\gog$ divides $D\Delta_\gog$.  But $D\Delta_\gog$ is strictly lower degree than $\Delta_\gog$, so we must have $D\Delta_\gog = 0$.
\end{proof}
We say that $\Delta_\gog$ is {\it $W$-harmonic}.  In fact, the space of all $W$-harmonic polynomials on $\tot$ is exactly the linear span of the partial derivatives of $\Delta_\gog$ \cite[ch.~3, Theorem 3.6]{SH}.  In particular, since the Laplacian on $\tot$ is $W$-invariant, $\Delta_\gog$ is harmonic in the traditional sense.

\subsubsection{Radial part of a differential operator}
\label{sec:radparts}
There is one final idea that we need to introduce before discussing Harish-Chandra's homomorphism of invariant differential operators.  This is the notion of the {\it radial part} of a differential operator, which was fully developed by Helgason in the 1960s and 1970s\footnotemark \footnotetext{See Helgason's book \cite[ch. 2, \textsection3]{SH} for a detailed reference on radial parts and other geometric operations on differential operators in a general setting.} but which already plays an important role in \cite{HC}.  We say that a submanifold $M \subset \mathfrak{g}$ is {\it transverse} to the adjoint orbits in $\mathfrak{g}$ if for each $p \in M$ we have a decomposition of tangent spaces \begin{equation} \label{eqn:transverse}T_p \mathfrak{g} = T_p \mathcal{O}_p \oplus T_p M,\end{equation} where $\mathcal{O}_p = \{ \mathrm{Ad}_g p \ | \ g \in G \}$ is the adjoint orbit of $p$.  We state without proof the following theorem, which is a special case of \cite[ch.~2, Theorem 3.6]{SH}.

\begin{theorem} \label{thm:radial-part} Let $M \subset \mathfrak{g}$ be a submanifold of $\mathfrak{g}$ that is transverse to the adjoint orbits. Let $D$ be a differential operator on $\mathfrak{g}$.  Then there exists a unique differential operator $\gamma(D)$ on $M$ such that, for each function $f \in C^\infty(\mathfrak{g})$ that is locally $\mathrm{Ad}$-invariant in the sense that $f(\mathrm{Ad}_g x) = f(x)$ for all $g$ in some neighborhood of $\mathrm{id}_G$, $$\overline{(Df)} = \gamma(D) \bar f,$$ the bar indicating restriction to $M$. \end{theorem}

The differential operator $\gamma(D)$ is called the {\it radial part of $D$ with transversal manifold $M$}.  The sets $$\tot' := \{ x \in \tot \ | \ \Delta_\gog(x) \not = 0 \}, \qquad \tot_\C' := \{ x \in \tot_\C \ | \ \Delta_\gog(x) \not = 0 \}$$ are called the {\it regular elements} of $\tot$ and $\tot_\C$ respectively.  As a submanifold of $\tot$, $\tot'$ is both dense in $\tot$ and transverse to the adjoint orbits.  To see that transversality holds, consider the root space decomposition of $\mathfrak{g}_\C$, $$\mathfrak{g}_\C = \tot_\C \oplus \bigoplus_\alpha \mathfrak{g}_\alpha,$$ where $\alpha$ runs over the roots of $\mathfrak{g}_\C$ with respect to $\tot_\C$.  Under the usual identification $T_x \mathfrak{g} \cong \mathfrak{g}$, at $x \in \tot'$ we have $$T_x \tot' \cong \tot, \quad T_x \mathcal{O}_x \cong [ \mathfrak{g}, x] = \mathfrak{g} \cap \bigoplus_\alpha \mathfrak{g}_\alpha,$$ which gives the transversality property.  Thus for each $D \in \mathfrak{D}(\mathfrak{g})$ we have a well-defined operator $\gamma(D) \in \mathfrak{D}(\tot')$, the radial part of $D$ with transversal manifold $\tot'$.  Moreover, we have the following fact \cite[Lemma 7]{HC}:

\begin{lemma} \label{lem:rp_hom}
The map $D \mapsto \gamma(D)$ is a homomorphism from $\mathcal{I}'$ to $\mathfrak{D}(\tot')$.
\end{lemma}

\begin{proof}
We follow the proof in \cite{HC}.  We first note that $\gamma$ is clearly linear from its definition.  Let $x \in \tot'$ and let $T \subset G$ be the maximal torus in $G$ with Lie algebra $\tot$.  Let $g \mapsto gT$ denote the quotient map $G \to G/T$.  Let $U \subset \tot'$, $V \subset G$ be open connected neighborhoods of $x$ and $\mathrm{id}_G$ respectively, and let $VT$ be the image of $V$ under the quotient map.  Define the function $\phi: VT \times U \to \mathfrak{g}$ by $\phi(gT, x) = \mathrm{Ad}_g x$.  Then $\phi$ is a submersion, and since $\dim(VT \times U) = \dim \mathfrak{g}$, $N := \phi(VT \times U)$ is an open submanifold of $\mathfrak{g}$.  If $V$ and $U$ are taken to be sufficiently small, then $\phi$ is bijective and defines an analytic isomorphism of $VT \times U$ onto $N$.  For $\psi \in C^\infty(U)$, define $f_\psi \in C^\infty(N)$ by $f_\psi(\phi(gT, x)) = \psi(x)$.  Then $f_\psi$ is locally $\mathrm{Ad}$-invariant.  Let $D_1, D_2 \in \mathcal{I}'$.  We have $$\overline{D_1 D_2 f_\psi} = \gamma(D_1 D_2) \bar f_\psi = \gamma(D_1 D_2) \psi.$$ On the other hand, since $D_2$ is $\mathrm{Ad}$-invariant, $D_2 f_\psi$ must be locally $\mathrm{Ad}$-invariant, so $$\overline{D_1 D_2 f_\psi} = \gamma(D_1) \overline{(D_2 f_\psi)},$$ and $\overline{(D_2 f_\psi)} = \gamma(D_2) \bar f_\psi = \gamma(D_2) \psi,$ so that $$\gamma(D_1 D_2) \psi =  \gamma(D_1) \gamma(D_2) \psi.$$ Since $\psi \in C^\infty(U)$ was arbitrary, $\gamma$ is a homomorphism.
\end{proof}

\subsubsection{The $\delta$ homomorphism}
We can now state the first major theorem proved in \cite{HC}, which Harish-Chandra calls ``the central result of this paper'' and which lies in the background of most of its other results.  The theorem gives a meaningful sense to the idea of ``restriction to $\tot$'' for invariant differential operators on $\mathfrak{g}$, rather than merely invariant polynomials.  Further, it relates the restrictions of these operators to their radial parts. Indeed, Harish-Chandra finds that Chevalley's isomorphism $I(\mathfrak{g}) \to I(\tot)$ extends uniquely to a homomorphism $\mathcal{I}(\mathfrak{g}) \to \mathcal{I}(\tot)$:

\begin{theorem} \label{thm:central} There exists a unique homomorphism $\delta: \mathcal{I}(\mathfrak{g}) \to \mathcal{I}(\tot)$ such that $\delta(p) = \bar p$ and $\delta(p(\partial)) = \bar p (\partial)$ for all $p \in I(\mathfrak{g})$. Moreover, on $\tot'$, we have \begin{equation} \label{eqn:radconj} \gamma(D) = \Delta_\gog^{-1} \delta(D) \circ \Delta_\gog \end{equation} for all $D \in \mathcal{I}(\mathfrak{g})$, where the circle indicates composition of operators. \end{theorem}

The bulk of Section 3 of \cite{HC} is devoted to a concrete construction of the homomorphism $\delta$ and to showing the relation (\ref{eqn:radconj}) between $\delta$ and the radial part map.  The details of Harish-Chandra's construction are beyond the scope of this paper, since we don't actually need the full power of the $\delta$ homomorphism to prove the integral formula.  Instead, it will suffice to understand the relationship between $\gamma(p(\partial))$ and $\bar p(\partial)$ for $p \in I(\mathfrak{g})$, as described in the following theorem \cite[Lemma 8]{HC}.

\begin{theorem} \label{thm:sh5.33}
For $p \in I(\mathfrak{g})$, $\gamma(p(\partial)) = \Delta_\gog^{-1} \bar p(\partial) \circ \Delta_\gog$.
\end{theorem}

We will sketch the proof, following the argument of \cite[ch. 2, Theorem 5.33]{SH}.  Let $\omega(x) := \langle x, x \rangle$, the {\it quadratic Casimir polynomial} on $\mathfrak{g}$. Then the Laplacian on $\mathfrak{g}$ is $\omega(\partial)$.  One first shows by a direct calculation\footnotemark \footnotetext{See \cite[ch. 3, Proposition 3.14]{SH} for details.} using the root space decomposition that $\gamma(\omega(\partial)) = \Delta_\gog^{-1} \bar \omega(\partial) \circ \Delta_\gog.$ The main idea of the proof is then to extend this result from $\omega$ to all $p \in I(\mathfrak{g})$ by using a trick of taking commutators with the Laplacian, which we now show in detail.

For differential operators $D_1$ and $D_2$, we write $$\{D_1, D_2 \} := D_1 \circ D_2 - D_2 \circ D_1$$ for their commutator.  Consider the derivations $$\mu: D \mapsto \frac{1}{2} \{\omega(\partial), D \}$$ on $\mathfrak{D}(\mathfrak{g})$ and $$\bar \mu: d \mapsto \frac{1}{2} \{\gamma(\omega(\partial)), d \}$$ on $\mathfrak{D}(\tot')$.  Then $\mu$ has the following property \cite[ch. 2, Lemma 5.34]{SH}.

\begin{lemma} \label{lem:mu_powers}
If $p$ is a homogeneous polynomial on $\mathfrak{g}$ of degree $m$, then $$\mu^m(p) = m! \, p(\partial).$$
\end{lemma}

\begin{proof}
The proof proceeds by induction on $m$.  The base case $m=1$ follows from direct calculation.  Let $p = q_1 \hdots q_m$, where each $q_i$ is an arbitrary linear function, and make the inductive hypothesis that $$\mu^{m-1}(q_1 \hdots q_{m-1}) = (m-1)! \, (q_1 \hdots q_{m-1})(\partial).$$  Observing that $\mu^2(q_i) = \mu(q_i(\partial)) = 0$, by the Leibniz rule for derivations we have $$\mu^m(q_1 \hdots q_{m-1} q_m) = \mu^m(q_1 \hdots q_{m-1}) \circ q_m + m \mu^{m-1}(q_1 \hdots q_{m-1}) \circ \mu(q_m).$$  Applying the inductive hypothesis on the right-hand side, we find that the first term vanishes and the second term is exactly $m! \, p(\partial)$.
\end{proof}

We will also need the following commutator identity, which holds in any associative algebra.

\begin{proposition} \label{prop:comm_id_rp}
Let $A$ be an associative algebra. For $a \in A$, define the derivation $d_a: A \to A$ by $d_a(b) = \frac{1}{2}(ab - ba).$  If $c \in A$ is invertible and commutes with $b$, then $$d_{c^{-1}ac}^k(b) = c^{-1}d_a^k (b)c.$$
\end{proposition}
\begin{proof}
This follows from the observation that $$d_a^k(b) = 2^{-k} \sum_{j=0}^k {k \choose j}(-1)^j a^{k-j}ba^j,$$ which is shown by an easy induction on $k$.
\end{proof}

Now we can complete the proof of Theorem \ref{thm:sh5.33}.

\begin{proof}[Proof of Theorem \ref{thm:sh5.33}]
It suffices to assume that $p$ is homogeneous of degree $m$, because $\gamma$ is linear.  Since $\gamma$ is a homomorphism on $\mathcal{I}(\mathfrak{g})$, for $D \in \mathcal{I}(\mathfrak{g})$ we have $$\gamma(\mu(D)) = \bar \mu(\gamma(D)),$$ and thus by Lemma \ref{lem:mu_powers} we have $$m!\!\ \gamma(p(\partial)) = \gamma(\mu^m(p)) = \bar \mu^m(\gamma(p)) = \bar \mu^m(\bar p).$$  Finally we apply Proposition \ref{prop:comm_id_rp} with $A = \mathfrak{D}(\tot'),$ $a = \bar \omega(\partial)$, $b = \bar p$, and $c = \Delta_\gog$, observing that $c^{-1} a c = \Delta_\gog^{-1} \bar \omega(\partial) \circ \Delta_\gog = \gamma(\omega(\partial)).$  This gives $$m!\!\ \gamma(p(\partial)) = \bar \mu^m(\bar p) = (d_{c^{-1} a c})^m(\bar p) = \Delta_\gog^{-1} d_a^m(\bar p) \circ \Delta_\gog = m!\! \ \Delta_\gog^{-1} \bar p(\partial) \circ \Delta_\gog,$$ which completes the proof.
\end{proof}

\subsection{Harish-Chandra's original proof}
\label{sec:proof}

We turn now to the first proof of the integral formula (\ref{eqn:hc}).  While I have reorganized the presentation and simplified some steps with particular help from the invaluable reference by Helgason \cite{SH}, the arguments in this section are quite close to Harish-Chandra's originals, with some additional explanation and a few modifications to notation and terminology to bring them more in line with contemporary usage.

To recap definitions and assumptions, here $G$ is a compact, connected, semisimple real Lie group of rank $r$ with Lie algebra $\mathfrak{g}$ and normalized Haar measure $dg$, $\tot \subset \mathfrak{g}$ is a Cartan subalgebra, $W$ is the corresponding Weyl group, $\mathfrak{g}_\C$ and $\tot_\C$ are the complexifications of $\mathfrak{g}$ and $\tot$ respectively, and $\langle \cdot, \cdot \rangle$ is an $\mathrm{Ad}$-invariant inner product on $\mathfrak{g}$.

The proof outline goes as follows.  Before proving Theorem \ref{thm:hc} itself, we need a few preliminary lemmas.  The most important of these identifies a space of analytic functions on which all operators in $\partial I(\tot)$ are simultaneously diagonalizable.  After these lemmas, the core argument of the proof proceeds in three steps:
\begin{enumerate}
\item Define a function $\phi_f(x) := \Delta_\gog(x) \int_G e^{\langle \mathrm{Ad}_g x, y \rangle} dg$ on $\tot_\C$.  With the assumptions that $x \in \tot$ and $\Delta_\gog(y) \not = 0$, show that $\phi_f$ is a joint eigenfunction for all operators $q(\partial) \in \partial I(\tot)$, and use the simultaneous diagonalization lemma to write down an explicit formula for $\phi_f$ containing a number of unknown constants indexed by the Weyl group $W$.
\item Use the same explicit formula to write down expressions for $\phi_f(w(x))$, $w \in W$, and average these over $W$ to eliminate all but a single unknown constant.
\item Determine the value of the remaining constant by computing $\Delta_\gog(\partial) \phi_f(x) |_{x=0}$ in two different ways, and setting the resulting expressions equal to each other.  This gives the integral formula for $x \in \tot$ and $\Delta_\gog(y) \not = 0$, and we then use the analyticity of $\phi_f$ to extend the result to all $x, y \in \tot_\C$.
\end{enumerate}

The first lemma that we need relates the algebraic structures of $\Pi(\tot)$ and $I(\tot)$.

\begin{lemma} \label{lem:sym_free_alg_over_i}  There are homogeneous elements $v_1, \hdots, v_{|W|} \in \Pi(\tot)$ such that every $v \in \Pi(\tot)$ can be written uniquely in the form $v = \sum_{i=1}^{|W|} u_i v_i$ where each $u_i \in I(\tot)$. \end{lemma}

In other words, $\Pi(\tot)$ is a free module of rank $|W|$ over $I(\tot)$.

\begin{proof}[Proof]
Let $J_+$ be the ideal in $I(\tot)$ generated by elements of positive degree.  Since $I(\tot)$ is a subalgebra of $\Pi(\tot)$, we can consider the ideal $J_+\Pi(\tot)$ generated by $J_+$ in the larger algebra $\Pi(\tot)$.  By a theorem of Chevalley \cite{ChevalleyInvar}, the quotient $\Pi(\tot)/J_+\Pi(\tot)$ is a complex vector space of dimension $|W|$, so we can choose $v_1, \hdots, v_{|W|} \in \Pi(\tot)$ such that $$\Pi(\tot) = \mathrm{span}\{v_1, \hdots, v_{|W|}\} + J_+\Pi(\tot).$$ Then an easy induction shows that in fact for any $m \ge 1$, $$\Pi(\tot) = \mathrm{span}\{v_1, \hdots, v_{|W|}\} \cdot I(\tot) + J_+^m \Pi(\tot).$$  If $p \in I(\tot)$ and $(p)_d$ is its degree $d$ homogeneous component, then $(p)_d \in I(\tot)$ as well, so that we may take the elements $v_i$ to be homogeneous.  Moreover, since $\Pi(\tot)$ contains no zero-divisors we can consider its field of fractions $\mathrm{Frac}(\Pi(\tot))$, and this is a field extension of degree $|W|$ over $\mathrm{Frac}(I(\tot))$.

Now let $v \in \Pi(\tot)$ be homogeneous of degree $d$.  We will show that $v$ can be written uniquely in the form stated in the lemma.  Let $d_i := \deg v_i$ and choose $m > d$ and $u_1', \hdots, u_{|W|}' \in I(\tot)$ such that $v - \sum_{i = 1}^{|W|} u'_i v_i \in J_+^m \Pi(\tot)$.  Let $u_i := (u'_i)_{d- d_i}.$ Then, since $$\bigg( v - \sum_{i = 1}^{|W|} u'_i v_i \bigg)_d = 0$$ by construction and $v = (v)_d$ is homogeneous, we must have $$v - \sum_{i = 1}^{|W|} u_i v_i = 0.$$  Thus we can conclude that $$\Pi(\tot) = \mathrm{span}\{v_1, \hdots, v_{|W|}\} \cdot I(\tot),$$ whereby it follows that $\mathrm{Frac}(\Pi(\tot))$ is spanned over $\mathrm{Frac}(I(\tot))$ by the elements $v_i$.  But since $\mathrm{Frac}(\Pi(\tot))$ is a degree $|W|$ extension of $\mathrm{Frac}(I(\tot))$, the elements $v_i$ must therefore be linearly independent, showing that the decomposition $v = \sum_{i=1}^{|W|} u_i v_i$ is unique.
\end{proof}

The purely algebraic statement of Lemma \ref{lem:sym_free_alg_over_i} will be our main tool in proving the next lemma, which is the aforementioned simultaneous diagonalization result for $\partial I (\tot)$. We pose an infinite system of eigenvalue problems for {\it each point} in $\tot'$, identify a family of analytic functions that solve all of them simultaneously, and then show that this family of solutions is exhaustive up to the assumption of analyticity.  This result has an interesting physical interpretation in the language of quantum integrable systems, where differential operators on $\tot$ can be viewed as quantum Hamiltonians; see \cite[chs. 4--5]{EtingofCM} for details on this topic.

\begin{lemma} \label{lem:diffeqs}
Let $U$ be a nonempty connected open subset of $\tot_\C$.  Let $x_0 \in \tot_\C$ such that $\Delta_\gog(x_0) \not = 0$.  Suppose $\phi$ is an analytic function on $U$ satisfying the system of differential equations \begin{equation} \label{eqn:diffeqs} q(\partial) \phi = q(x_0) \phi, \quad \forall \, q \in I(\tot). \end{equation}  Then there exist constants $c_w \in \mathbb{C}$, $w \in W,$ such that for all $x \in U$, $$\phi(x) = \sum_{w \in W} c_w e^{\langle x, w(x_0) \rangle}.$$  Moreover, for any such $\phi$, the constants $c_w$ are uniquely determined.
\end{lemma}

In other words, the functions $e^{\langle x, w(x_0) \rangle}$ for $w \in W$ form a basis of the complex vector space of analytic solutions to the linear system (\ref{eqn:diffeqs}).

\begin{proof}[Proof]
First we note that $\Delta_\gog(x_0) \ne 0$ implies that the points $w(x_0)$, $w \in W$ are all distinct, so that the analytic functions $$\phi_w(x) = e^{\langle x, w(x_0)\rangle}, \quad w \in W$$ are linearly independent on $U$.\footnote{These are well-known facts, but see e.g.~\cite[Lemma 4]{HCe}, \cite[Lemma 41]{HCb} for detailed proofs.}

If we identify a point $y \in \tot_\C$ with the linear functional $\langle y, \cdot \rangle$ on $\tot$, then $y(\partial) \phi_w(x) = \langle y, w(x_0) \rangle \phi_w(x)$. Writing any $q \in \Pi(\tot)$ in terms of such linear functionals, we find that $q(\partial) \phi_w = q(w(x_0))\phi_w$ for $q \in \Pi(\tot)$, and in particular for $q \in I(\tot)$.  In other words, each $\phi_w$ solves (\ref{eqn:diffeqs}).

Let $E$ be the vector space over $\mathbb{C}$ consisting of all analytic solutions to (\ref{eqn:diffeqs}).  We know already that $\dim E \ge |W|$ since the $\phi_w$ are linearly independent.  To show that the $\phi_w$ form a basis for $E$, it is therefore sufficient to show that assuming $\dim E > |W|$ leads to a contradiction.

Choose a point $x_1 \in U$ and let $v_1, \hdots, v_{|W|}$ be as in Lemma \ref{lem:sym_free_alg_over_i}.  If $\dim E > |W|$, we can choose $\psi \not = 0$ in $E$ satisfying the $|W|$ linear conditions $$v_i(\partial) \psi(x) \big |_{x = x_1} = 0, \quad 1 \le i \le |W|.$$  But this is impossible: by Lemma \ref{lem:sym_free_alg_over_i}, for any $v \in \Pi(\tot)$ we can write $v = \sum_{i=1}^{|W|} u_i v_i$ with $u_i \in I(\tot)$, and since $\psi$ solves (\ref{eqn:diffeqs}), we have $$v(\partial) \psi(x) \big |_{x = x_1} = \sum_{i = 1}^{|W|} u_i(x_0) v_i(\partial) \psi(x) \big |_{x = x_1} = 0.$$  In other words all derivatives of $\psi$ vanish at $x_1$, and since $\psi$ is analytic it must therefore be identically zero, which contradicts our assumption.
\end{proof}

Now, for any $f \in C^\infty(\mathfrak{g})$, define a function $\phi_f \in C^\infty(\tot)$ by 
\begin{equation}  \label{eqn:phif-def}
\phi_f(x) := \Delta_\gog(x) \int_G f(\mathrm{Ad}_g x)\, dg.
\end{equation}
We make the following observation about how $\phi_f$ transforms under the action of invariant differential operators.

\begin{lemma} \label{lem:hc15_subfromSH}
For all $p \in I(\mathfrak{g})$, $$\phi_{p(\partial) f} = \bar p(\partial) \phi_f.$$
\end{lemma}

\begin{proof}
Let $F(x) := \int_G f(\mathrm{Ad}_g x) dg.$  Then by Theorem \ref{thm:sh5.33}, for $x \in \tot'$ we have $(p(\partial)F)(x) = (\Delta_\gog^{-1} \bar p(\partial)(\Delta_\gog \bar F))(x)$, so that $$\phi_{p(\partial) f} (x) = \Delta_\gog(x) p(\partial) F(x) = \bar p(\partial) \phi_f(x).$$  Since $\phi_{p(\partial) f}$ and $\bar p(\partial) \phi_f$ are both continuous and $\tot'$ is dense in $\tot$, this equality must in fact hold for all $x \in \tot$.
\end{proof}

We now can give Harish-Chandra's original proof of the integral formula (\ref{eqn:hc}), following the outline of steps at the beginning of this subsection.

\begin{proof}[Proof (Theorem \ref{thm:hc})]
\ \\
\noindent {\bf Step 1: Identify an appropriate joint eigenfunction of $\partial I (\tot)$.} \\
Choose $y \in \tot'_\C$ and define $f: \mathfrak{g} \to \mathbb{C}$ by $f(x) = e^{\langle x, y \rangle}$, so that for $x \in \tot$ we have \begin{equation} \label{eqn:phi_f_def} \phi_f(x) = \Delta_\gog(x) \int_G e^{\langle \mathrm{Ad}_g x, y \rangle} dg.\end{equation}  For $y_0 \in \mathfrak{g}$ we have $\partial_{y_0} f = \langle y_0, y \rangle f$, so that for $q \in \Pi(\mathfrak{g})$ we have $q(\partial)f = q(y) f$.  But by Lemma \ref{lem:hc15_subfromSH}, for $p \in I(\mathfrak{g})$ we have $\phi_{p(\partial) f} = \bar p (\partial) \phi_f$, so that $$\bar p (\partial) \phi_f = \phi_{p(\partial) f} = \phi_{p(y) f} = p(y) \phi_f.$$  Since this holds for all $p \in I(\mathfrak{g})$, we may apply the isomorphism of the Chevalley restriction theorem to conclude that $$q(\partial) \phi_f = q(y) \phi_f, \quad \forall \, q \in I(\tot).$$  In other words the analytic function $\phi_f$ satisfies the system of differential equations (\ref{eqn:diffeqs}).  Therefore, by Lemma \ref{lem:diffeqs}, there is a unique choice of constants $c_w$, $w \in W$ such that we can write $$\phi_f (x) = \sum_{w \in W} c_w e^{\langle w(x), y \rangle}, \qquad x \in \tot.$$

\noindent {\bf Step 2: Average over $W$ to eliminate all but one unknown constant.} \\
By the skewness of $\Delta_\gog$ and the $\mathrm{Ad}$-invariance of the integral in (\ref{eqn:phi_f_def}), we have $$\phi_f(w(x)) = \epsilon(w) \phi_f(x).$$  Multiplying this identity on both sides by $\epsilon(w)$ and taking the average over $W$ gives \begin{equation} \label{eqn:phi_avg} \phi_f(x) = |W|^{-1} \sum_{w \in W} \epsilon(w) \phi_f(w(x)) = |W|^{-1} c \sum_{w \in W} \epsilon(w) e^{\langle w(x),y \rangle}, \end{equation}
where $c = \sum_{w \in W} \epsilon(w) c_w$. \\

\noindent {\bf Step 3: Determine the remaining constant.} \\
We now determine $c$ by computing $\Delta_\gog(\partial)\phi_f(x)|_{x=0}$ in two different ways.  Define, for each $w \in W$, $$\psi_w(x) := e^{\langle w(x),y \rangle},$$ so that (\ref{eqn:phi_avg}) becomes $$\phi_f = |W|^{-1} c \sum_{w \in W} \epsilon(w) \psi_w.$$ Then for $q \in \Pi(\tot)$, $q(\partial) \psi_w = q(w^{-1}(y)) \psi_w$.  In particular, $$\Delta_\gog(\partial) \psi_w = \Delta_\gog(w^{-1}(y)) \psi_w =  \epsilon(w^{-1}) \Delta_\gog(y) \psi_w = \epsilon(w) \Delta_\gog(y) \psi_w,$$ and therefore \begin{multline} \label{eqn:exp1} \Delta_\gog(\partial) \phi_f (x) \big |_{x=0} = \Delta_\gog(\partial) \left[ |W|^{-1} c \sum_{w \in W} \epsilon(w) \psi_w(x) \right]_{x = 0} \\ = |W|^{-1}c \sum_{w \in W} \epsilon(w)^2 \Delta_\gog(y) \psi_w(0) = |W|^{-1} c \, \Delta_\gog(y) |W| = c \, \Delta_\gog(y).\end{multline}

Now that we have one expression for the value of $\Delta_\gog(\partial) \phi_f (x) |_{x=0}$, we'll calculate it again in a different way and equate the two answers to each other.  We apply the product rule to compute \begin{multline} \label{eqn:exp2} \Delta_\gog(\partial) \phi_f (x) \big |_{x=0} = \Delta_\gog(\partial) \left( \Delta_\gog(x) \int_G f(\mathrm{Ad}_g x) dg \right) \bigg |_{x = 0} \\ = [\![\Delta_\gog,\Delta_\gog]\!] \int_G f(0) dg + \Delta_\gog(0) \cdot \Delta_\gog(\partial) \int_G f(\mathrm{Ad}_g x) dg \bigg |_{x = 0} \\ = [\![\Delta_\gog,\Delta_\gog]\!] f(0) + 0= [\![\Delta_\gog,\Delta_\gog]\!], \end{multline} where we have used the facts that the Haar measure $dg$ is taken to be normalized and that $\Delta_\gog(0) = 0$.

Equating (\ref{eqn:exp1}) and (\ref{eqn:exp2}), we see that $c = [\![\Delta_\gog,\Delta_\gog]\!] / \Delta_\gog(y)$.  Plugging this result into (\ref{eqn:phi_avg}) and multiplying both sides by $\Delta_\gog(y)$ gives the desired formula, $$\Delta_\gog(x) \Delta_\gog(y) \int_G e^{\langle \mathrm{Ad}_g x, y \rangle} dg = \frac{[\![\Delta_\gog,\Delta_\gog]\!]}{|W|} \sum_{w \in W} \epsilon(w) e^{\langle w(x),y \rangle},$$ which we have established {\it with the assumptions} that $x \in \tot$ and $y \in \tot'_\C$.

To extend the result to all $x, y \in \tot_\C$, we observe that the left- and right-hand sides are both holomorphic functions on $\tot_\C \times \tot_\C$, and that these functions agree on $\tot \times \tot'_\C$, so that they must agree on all of $\tot_\C \times \tot_\C$.  This completes the proof.
\end{proof}

\subsection{Heat equation proof} \label{sec:hc_heat_eqn_proof}

In this section, we prove Harish-Chandra's formula by relating the heat flow on $\gog$ to the heat flow on $\tot$.  This proof was published in \cite{McS-heateqn} and generalizes a method employed by Itzykson and Zuber to prove the HCIZ formula in \cite{IZ}.

The heat equation proof provides insight into the asymptotics of the Harish-Chandra integral (\ref{eqn:hc}) as the rank of $G$ increases.  As explained in Section \ref{subsec:cm_relationship} below, the relationship between the Harish-Chandra integral and the heat equation on $\tot$ leads to a heuristic calculation suggesting that the large-rank asymptotics of (\ref{eqn:hc}) can be described in terms of a certain hydrodynamic scaling limit of a classical Calogero--Moser system associated to the root system of $\gog$.  If rigorously proven, this would generalize known results for the $\U(N)$ case, where leading-order asymptotics for the HCIZ integral (\ref{eqn:hciz}) as $N \to \infty$ were formally computed by \cite{AM} and rigorously justified by \cite{GZ} and \cite{AG}.  The result is an expression for the leading-order contribution to the large-$N$ limit in terms of a particular solution to the complex Burgers' equation.  It has been shown that this phenomenon can be understood in terms of a relationship between the HCIZ integral and the Calogero--Moser system associated to the $A_{N-1}$ root system of $\U(N)$ \cite{GM}.  As discussed above in Section \ref{sec:HC-and-IS}, Harish-Chandra integrals arise naturally in the study of the quantum Calogero--Moser system, though it is still unclear whether this fact is directly related to the observations in Section \ref{subsec:cm_relationship}.

While it suffices for the purpose of proving the integral formula to relate heat flow on a Lie algebra to heat flow on a Cartan subalgebra, one could obtain richer information by studying the relationship between Brownian motion on the full algebra and Brownian motion confined to a Weyl chamber.  This latter process is a generalization of Dyson Brownian motion and has been studied by \cite{DG}.  The arguments regarding asymptotics for the $\U(N)$ integral due to \cite{GZ} and \cite{AG} rely on a large-deviations principle for Dyson Brownian motion, suggesting that a similar investigation of processes on general Weyl chambers could prove useful in studying the large-rank asymptotics of integrals over other groups.

Sections \ref{subsec:heat-equations} through \ref{subsec:normalization} below comprise the heat equation proof of (\ref{eqn:hc}).

\subsubsection{Heat equations on $\mathfrak{g}$ and $\tot$}
\label{subsec:heat-equations}

Let $\omega$ be the quadratic Casimir polynomial on $\mathfrak{g}$ defined by $\omega(x) := |x|^2 = \langle x, x \rangle$.  The Laplacian on $\mathfrak{g}$ is the differential operator $\omega(\partial)$.  A function $\phi: \mathfrak{g} \times (0, \infty) \to \mathbb{C}$ satisfies the heat equation on $\mathfrak{g}$ if
\begin{equation} \label{eqn:heat} \left(\partial_t - \frac{1}{2} \omega(\partial_x) \right) \phi(x,t) = 0, \qquad x \in \mathfrak{g}, \ t \in (0, \infty). \end{equation}
Similarly, $\psi: \tot \times (0, \infty) \to \mathbb{C}$ satisfies the heat equation on $\tot$ if 
\begin{equation} \label{eqn:heat-csa} \left(\partial_t - \frac{1}{2} \bar \omega(\partial_x) \right) \psi(x,t) = 0, \qquad x \in \tot, \ t \in (0, \infty), \end{equation}
where the bar indicates restriction to $\tot$.

We want to establish a relationship between the solutions of (\ref{eqn:heat}) and of (\ref{eqn:heat-csa}).  In general, if $\phi$ satisfies (\ref{eqn:heat}), it is not the case that $\bar \phi$ satisfies (\ref{eqn:heat-csa}).  However, we will show:
\begin{lemma} \label{lem:soln-reln} If $\phi \in C^{2}_1(\mathfrak{g} \times (0, \infty))$ solves (\ref{eqn:heat}) and is invariant under the adjoint action of $G$ in the sense that $$\phi(\mathrm{Ad}_g x, t) = \phi(x, t), \quad \forall \, g \in G,$$ then $\Delta_\gog(x) \bar \phi(x, t)$ solves (\ref{eqn:heat-csa}). \end{lemma}

The proof of Lemma \ref{lem:soln-reln} makes use of the following expression for the radial part of the Laplacian, which is a special case of Theorem \ref{thm:sh5.33}.  

\begin{lemma} \label{lem:radial-laplacian} The radial part of the Laplacian $\omega(\partial)$ with transversal manifold $\tot'$ is given by $$\gamma(\omega(\partial)) = \Delta_\gog^{-1} \bar \omega(\partial) \circ \Delta_\gog.$$ \end{lemma}

\begin{proof}
Recall that the regular elements of $\tot$ are the subset $$\tot' := \{ x \in \tot \ | \ \Delta_\gog(x) \ne 0 \}.$$ By Theorem \ref{thm:sh5.33} above, the radial part of the Laplacian $\omega(\partial)$ with transversal manifold $\tot'$ is given by
\begin{equation} \label{eqn:radial-laplacian}
\gamma(\omega(\partial)) = \Delta_\gog^{-1} \bar \omega(\partial) \circ \Delta_\gog.
\end{equation}
Restricting (\ref{eqn:heat}) to $\tot'$ and applying (\ref{eqn:radial-laplacian}), we obtain the {\it radial heat equation}
\begin{equation} \label{eqn:radial-heat} \left( \partial_t - \frac{1}{2} \Delta_\gog^{-1}(x) \bar \omega(\partial_{x}) \circ \Delta_\gog(x) \right) \bar \phi(x, t) = 0, \qquad x \in \tot', \end{equation}
where the bar indicates restriction to $\tot'$.  Multiplying both sides by $\Delta_\gog$ we have \begin{equation} \label{eqn:rest-heat-soln} \left( \partial_t - \frac{1}{2} \bar \omega(\partial_{x}) \right) \Delta_\gog(x) \bar \phi(x, t) = 0, \qquad x \in \tot'. \end{equation} Thus $\Delta_\gog(x) \bar \phi(x, t)$ solves the heat equation on $\tot'$ and therefore on all of $\tot$ by continuity since $\tot'$ is dense.
\end{proof}

\subsubsection{The $G$-averaged heat kernel}

The {\it heat kernel} on $\mathfrak{g}$ is the function $K: \mathfrak{g}^2 \times (0, \infty) \to \mathbb{C}$ given by 
\begin{equation} \label{eqn:g-heat-kernel}
K(x, y; t) := \left(\frac{1}{2 \pi t} \right )^{\dim \mathfrak{g} /2} e^{-\frac{1}{2t} |x - y|^2}.
\end{equation}
It solves the heat equation (\ref{eqn:heat}) where the spatial derivatives act in either the $x$ or the $y$ variables, with the boundary condition \begin{equation} \label{eqn:bc1} \lim_{t \to 0} K(x, y; t) = \delta(x - y).\end{equation} The limit is understood in the distributional sense of \begin{equation*} \label{eqn:dist-sense} \lim_{t \to 0} \int_{\mathfrak{g}} K(x, y; t) \varphi(y) dy = \varphi(x) \end{equation*} for $\varphi \in C_c^\infty(\mathfrak{g})$, where $dy$ is the Lebesgue measure induced by the inner product.  The choice of this integration measure is significant, as it guarantees that for all $x$ and $t$ we have
$$\int_{\mathfrak{g}} K(x, y; t)\, dy = 1.$$

Following \cite{IZ}, we define the {\it $G$-averaged heat kernel} as
\begin{align} \nonumber \tilde K(x, y; t) &:= \int_G K(\mathrm{Ad}_g x, y ; t)\, dg \\
\label{eqn:k-avg-def} &= \left(\frac{1}{2 \pi t} \right )^{\dim \mathfrak{g} /2} e^{-\frac{1}{2t}(|x|^2 + |y|^2)} I(x, y; t),
\end{align}
where  \begin{equation} \label{eqn:I-integral-def} I(x, y; t) := \int_G e^{\frac{1}{t} \langle \mathrm{Ad}_g x, y \rangle} dg.\end{equation}
Observe that $I(x, y; 1) = \CH(x,y)$ for $x, y \in \tot$, so that the Harish-Chandra integral appears naturally in this context.

The $G$-averaged heat kernel $\tilde K$ is constant on adjoint orbits of both $x$ and $y$, and by linearity it satisfies the heat equation (\ref{eqn:heat}) on $\mathfrak{g}$ as well, so that by Lemma \ref{lem:soln-reln}, $\Delta_\gog(y) \tilde K(x, y; t)$ satisfies the heat equation (\ref{eqn:heat-csa}) on $\tot$ with the spatial derivatives acting in the $y$ variables.  Therefore the function \begin{equation} \label{eqn:v-def} V(x, y; t) := (2\pi)^{(\dim \mathfrak{g} - r)/2} \Delta_\gog(x) \Delta_\gog(y) \tilde K(x, y; t) \end{equation} also satisfies (\ref{eqn:heat-csa}) and is skew with respect to the action of $W$ on either of $x$ or $y$ individually.  In the next step we identify the boundary conditions that $V$ satisfies as $t$ approaches 0.  This will allow us to write an exact expression for $V$ using the fundamental solution to the heat equation on $\tot$, yielding (\ref{eqn:hc}).

\subsubsection{Boundary conditions for $V$}

We now further assume that both $x, y \in \tot'$.  In order to compute the distributional limit of $V$ as $t$ approaches 0, we use Laplace's method to determine the asymptotics of $I(x, y; t)$ to leading order in $t$.  We will show:

\begin{lemma} \label{lem:bcs} If $x, y \in \tot'$, then \begin{equation} \label{eqn:bcs} \lim_{t \to 0} V(x, y; t) = C \sum_{w \in W} \epsilon(w) \delta(w(x) - y) \end{equation} for some constant $C \in \R$, where the distributional sense of the limit is given by integration against test functions in $C_c^\infty(\tot)$ with respect to the Lebesgue measure induced by the restriction of the inner product to $\tot$. \end{lemma}

\begin{proof}
We first rewrite \begin{equation} \label{eqn:rewrite-v} V(x, y;t) = \frac{t^{-\dim \mathfrak{g} /2}}{(2\pi)^{r/2}} \Delta_\gog(x) \Delta_\gog(y) e^{-\frac{1}{2t}(|x|^2 + |y|^2)} I(x, y; t).\end{equation}  Next we rewrite $I(x, y; t)$ as follows.  Let $T$ be the maximal torus in $G$ with Lie algebra $\tot$.  Let $dh$ be the normalized Haar measure on $T$ and let $d(gT)$ be the unique left-invariant probability measure on $G/T$.  Then by a standard Fubini-type theorem for Lie groups \cite[ch.~1, Theorem 1.9]{SH} we have: \begin{multline} \label{eqn:rewrite-i} I(x, y; t) = \int_G e^{\frac{1}{t} \langle \mathrm{Ad}_g x, y \rangle} dg = \int_{G/T} \int_T e^{\frac{1}{t} \langle \mathrm{Ad}_{gh} x, y \rangle} dh \ d(gT) \\ = \int_{G/T} e^{\frac{1}{t} \langle \mathrm{Ad}_{g} x, y \rangle} d(gT).\end{multline}

We will apply Laplace's method to the last integral in (\ref{eqn:rewrite-i}).  To do so, we need the following lemmas computing the critical points of the function $gT \mapsto \langle \mathrm{Ad}_g x, y \rangle$ along with its Hessian matrix at each critical point.

\begin{lemma} \label{lem:crit-pts}
The critical points of the function $gT \mapsto \langle \mathrm{Ad}_g x, y \rangle$ on $G/T$ are the points $gT$ such that $\mathrm{Ad}_g x = w(x)$ for some $w \in W$.
\end{lemma}
\begin{proof}
We first note that the map $gT \mapsto \mathrm{Ad}_g x$ is a diffeomorphism of $G/T$ onto the adjoint orbit $\mathcal{O}_{x} \subset \gog$.  Thus we may equivalently find the critical points of the function $x_0 \mapsto \langle x_0, y \rangle$ for $x_0 \in \mathcal{O}_{x}$.  The tangent space at a point $x_0 \in \mathcal{O}_{x}$ is $T_{x_0} \mathcal{O}_{x} \cong [x_0, \mathfrak{g}],$ and the partial derivative of the linear functional $\langle y, \cdot \rangle$ in the direction of a tangent vector $y_0$ is equal to $\langle y, y_0 \rangle.$  Thus the condition for $x_0 \in \mathcal{O}_{x}$ to be a critical point is $$\langle y, y_0 \rangle = 0, \quad \forall \, y_0 \in [x_0, \mathfrak{g}],$$ or equivalently $$\langle [x_0, y_0], y \rangle = 0, \quad \forall \, y_0 \in \mathfrak{g}.$$  Using the antisymmetry of the bracket and the invariance of the inner product, this is equivalent to $$\langle y_0, [x_0, y] \rangle = 0, \quad \forall \, y_0 \in \mathfrak{g}.$$ By the nondegeneracy of the inner product, this will hold if and only if $[x_0, y] = 0$, which implies $x_0 \in \tot$.  Writing $x_0 = \mathrm{Ad}_g x$, we have $x_0 \in \tot$ exactly when $\mathrm{Ad}_g x = w(x)$ for some $w \in W$.
\end{proof}

\begin{lemma} \label{lem:hessian-det}
Let $H$ be the Hessian matrix of the function $gT \mapsto \langle \mathrm{Ad}_g x, y \rangle$ at a critical point $g_0 T$, in exponential coordinates on $G/T$ given by $e^{\xi}T \mapsto \xi + \tot \in \gog / \tot$.  Let $w$ be the element of $W$ such that $\mathrm{Ad}_{g_0} x = w(x)$ for all $x \in \gog$.  Then
\begin{equation} \label{eqn:hessian-det}
\sqrt{\det(-H)} = \epsilon(w) \Delta_\gog(x) \Delta_\gog(y).
\end{equation}
\end{lemma}
In order to give the correct global phase when applying Laplace's method, the branch of the square root in (\ref{eqn:hessian-det}) is chosen by writing $\sqrt{\det(- H)} = \prod_{j=1}^n \sqrt{-\mu_j}$ where $\mu_j$ are the eigenvalues of $H$, and taking $|\mathrm{arg} \sqrt{-\mu_j}| < \pi/4$.
\begin{proof}
Since $$\mathrm{Ad}_{\exp(\xi)} = \sum_{j = 0}^\infty \frac{1}{j!} \mathrm{ad}_{\xi}^j,$$ expanding to second order in $\xi$ around the critical point $g_0 T$ gives \begin{equation} \label{eqn:second-order} \langle \mathrm{Ad}_{\exp(\xi)} w(x), y \rangle = \langle w(x), y \rangle + \frac{1}{2}\langle \mathrm{ad}_{\xi}^2 w(x), y \rangle + O(|\xi|^3).\end{equation}
To obtain the Hessian we must compute explicitly the second-order term in (\ref{eqn:second-order}).  For a root $\alpha$, let $\mathfrak{g}_\alpha \subset \mathfrak{g}_\C$ denote the corresponding root space.  For each $\alpha \in \Phi^+$, choose $x_\alpha \in \mathfrak{g}_\alpha$, $x_{-\alpha} \in \mathfrak{g}_{-\alpha}$ normalized so that $\langle x_\alpha, x_{-\alpha} \rangle = 1$.  Assuming without loss of generality that $\xi$ lies in the orthogonal complement of $\tot$, we may write \begin{equation} \label{eqn:x_root_decomp} \xi = \sum_{\alpha \in \Phi^+} c_\alpha x_\alpha + c_{-\alpha} x_{-\alpha} \end{equation} for some coefficients $c_\alpha,$ $c_{-\alpha} \in \C$, and we find by a straightforward calculation\footnote{Recall from Section \ref{sec:notation} that we take the roots to be real valued on $\tot$, so that $[h, x_\alpha] = i \alpha(h) x_\alpha$ for $h \in \tot$.} that $$\frac{1}{2}\langle \mathrm{ad}_{\xi}^2 w(x), y \rangle = - \sum_{\alpha \in \Phi^+} \alpha(w(x))\alpha(y) c_\alpha c_{-\alpha}.$$  Thus $H$ is a block-diagonal matrix composed of 2-by-2 blocks of the form $$\begin{bmatrix} 0 & -\alpha(w(x))\alpha(y) \\ -\alpha(w(x)) \alpha(y) & 0 \end{bmatrix}$$ for each $\alpha \in \Phi^+$.  With the appropriate choice of branch for the square root, we find $$\sqrt{\det(-H)} = \Delta_\gog(w(x)) \Delta_\gog(y) = \epsilon(w) \Delta_\gog(x) \Delta_\gog(y)$$ as desired.
\end{proof}

Returning now to the proof of Lemma \ref{lem:bcs} and applying Laplace's method to (\ref{eqn:rewrite-i}) together with Lemmas \ref{lem:crit-pts} and \ref{lem:hessian-det}, we obtain \begin{equation} \label{eqn:i-approx} I(x, y; t) = C \frac{\left( 2\pi t \right)^{(\dim \mathfrak{g} -r)/2}}{\Delta_\gog(x) \Delta_\gog(y)} \sum_{w \in W} \epsilon(w) e^{\frac{1}{t}\langle w(x), y \rangle}(1 + O(t)) \end{equation} where the constant $C$ arises from the normalization of the measure $d(gT)$.  Substituting this result into (\ref{eqn:rewrite-v}), we find \begin{equation} \label{eqn:v-approx} V(x, y; t) = C \left( \frac{1}{2 \pi t} \right)^{r/2} \sum_{w \in W} \epsilon(w) e^{-\frac{1}{2t} | w(x) - y |^2}(1 + O(t))\end{equation} as $t \to 0$, which gives the desired limit (\ref{eqn:bcs}).
\end{proof}

Because $V$ solves the heat equation on $\tot$, taking the convolution of the boundary data (\ref{eqn:bcs}) with the fundamental solution gives \begin{equation} \label{eqn:exact-v} V(x, y; t) = C \left( \frac{1}{2 \pi t} \right)^{r/2} \sum_{w \in W} \epsilon(w) e^{-\frac{1}{2t} |w(x) - y|^2}. \end{equation} Thus the higher-order terms on the right-hand side of (\ref{eqn:v-approx}) actually vanish.  In physical terms, the expression (\ref{eqn:exact-v}) is analogous to a Slater determinant, with $W$-skewness playing the role of the antisymmetry property of fermions.

\subsubsection{Normalization} \label{subsec:normalization}

It only remains to rearrange terms and determine the constant $C$.  Evaluating $V$ at $t = 1$, we have:
\begin{align*}
V(x, y; 1) &= (2\pi)^{(\dim \mathfrak{g} - r)/2} \Delta_\gog(x) \Delta_\gog(y) \tilde K(x, y; 1) \\
& = (2\pi)^{(\dim \mathfrak{g} - r)/2} \Delta_\gog(x) \Delta_\gog(y) \left(\frac{1}{2 \pi} \right )^{\dim \mathfrak{g} /2} e^{-\frac{1}{2}(|x|^2 + |y|^2)} I(x, y; 1) \\
& = C \left( \frac{1}{2 \pi} \right)^{r/2} \sum_{w \in W} \epsilon(w) e^{-\frac{1}{2} |w(x) - y|^2} \\
& = C \left( \frac{1}{2 \pi} \right)^{r/2} e^{-\frac{1}{2}(|x|^2 + |y|^2)} \sum_{w \in W} \epsilon(w) e^{\langle w(x),y \rangle}.
\end{align*}
After cancelations, this becomes \begin{equation} \label{eqn:almost-hc} \Delta_\gog(x) \Delta_\gog(y) I(x, y; 1) = C \sum_{w \in W} \epsilon(w) e^{\langle w(x),y \rangle}.\end{equation}
Up to this point we have assumed that $x, y \in \tot'$, but we can immediately remove this assumption: since both sides of (\ref{eqn:almost-hc}) are analytic in $x$ and $y$, this identity holds for all $x, y \in \tot_\C$.

Finally, we determine $C$.  Applying $\Delta_\gog(\partial_{x})$ to both sides of (\ref{eqn:almost-hc}) and evaluating at $x = 0$, we obtain $$\Delta_\gog(y) [ \! [ \Delta_\gog, \Delta_\gog ] \! ] = C |W| \Delta_\gog(y), \quad y \in \tot_\C,$$
with $[ \! [ \Delta_\gog, \Delta_\gog ] \! ]$ defined by (\ref{eqn:poly-bracket-def}), so that $C = [ \! [ \Delta_\gog, \Delta_\gog ] \! ] / |W|$.  This completes the proof of Theorem \ref{thm:hc}.

\subsubsection{Relationship to Calogero--Moser systems} \label{subsec:cm_relationship}
To illustrate the relationship between Harish-Chandra integrals and Calogero--Moser systems, we observe that by Lemma \ref{lem:soln-reln} the function
\begin{equation} \label{eqn:WKB-ansatz}
W(x, y; t) := \frac{1}{r^2} \log \tilde K(\sqrt{r} x, \sqrt{r} y ; t)
\end{equation}
satisfies \begin{equation} \label{eqn:free_energy_pde} 2 \frac{\partial W}{\partial t} = r |\nabla W|^2 + \frac{2}{r} \nabla(\log \Delta_\gog ) \cdot \nabla W - \frac{1}{r} \bar \omega(\partial) W - \frac{1}{r^3} \Delta_\gog^{-1} \bar \omega(\partial) \Delta_\gog, \end{equation} where spatial derivatives act in the $y$ variables and $\Delta_\gog = \Delta_\gog(y)$.  Since $\Delta_\gog$ is harmonic by Proposition \ref{prop:pi_harmonic}, the last term on the right-hand side of (\ref{eqn:free_energy_pde}) vanishes.  In physical terms, the substitution (\ref{eqn:WKB-ansatz}) resembles a WKB ansatz with $r^{-2}$ playing the role of $\hbar$, so the large-rank limit in this setting bears some similarities to a semiclassical approximation.  Some caution is required in applying this analogy however, since $r$ is not a numerical parameter but rather the dimension of the underlying space.

To compute the large-$N$ asymptotics for the integral over $\U(N)$ (where $r = N$), Matytsin \cite{AM} drops the term corresponding to $r^{-1} \bar \omega(\partial) W$ in (\ref{eqn:free_energy_pde}), arguing heuristically that it should be subdominant as $N \to \infty$.  If we neglect this term and make the substitution $W = S - r^{-2} \log \Delta_\gog$, then we arrive at \begin{equation} \label{eqn:CM_preHJ} 2 \frac{\partial S}{\partial t} = r |\nabla S|^2 - r^{-3} |\nabla (\log \Delta_\gog) |^2. \end{equation}  In fact (\ref{eqn:CM_preHJ}) is the Hamilton--Jacobi equation for a rational Calogero--Moser system associated to the root system of the algebra $\mathfrak{g}$, with $r$ particles.  The factors of $r$ and $r^{-3}$ multiplying the two terms in (\ref{eqn:CM_preHJ}) have the effect of scaling both the spacing between particles and the interaction strength by $1/r$.  From this observation we expect that the large-rank asymptotics of Harish-Chandra integrals can be understood in terms of hydrodynamic scaling limits of Calogero--Moser systems.  This is known to be true for integrals over $\U(N)$, as in \cite{AM, GZ} the leading-order asymptotics of the HCIZ integral are derived in terms of a particular solution to the complex Burgers' equation.  As explained in \cite{GM}, the complex Burgers' equation also arises as a hydrodynamic limit of the Calogero--Moser system associated to the $A_{N-1}$ root system.

\subsection{Symplectic geometry proof} \label{subsec:symplectic_proof}

It is also possible to prove Theorem \ref{thm:hc} using a localization technique in symplectic geometry.  While the existence of such a proof was observed already by Duistermaat and Heckman in the 1980's \cite{DH}, here we present the details of the argument with the goal of making this derivation accessible to non-specialists.

The symplectic geometry approach illustrates a completely different perspective on the Harish-Chandra integral.  Rather than writing the integral in terms of the heat kernel on $\mathfrak{g}$, we instead view it as an oscillatory integral over a coadjoint orbit in $\mathfrak{g}^*$.  The proof begins with the Laplace's method approximation (\ref{eqn:i-approx}) for $I(x, y; t)$ obtained in the proof of Lemma \ref{lem:bcs}, though here we work with $I(x, y; -it^{-1})$ rather than $I(x, y; t)$, so that (\ref{eqn:i-approx}) becomes a stationary-phase approximation.  Then, instead of observing that $V$ solves the heat equation on $\tot$, we use the Duistermaat--Heckman theorem to deduce that the approximation is exact, after which it remains only to compute the normalization constant following the argument of Section \ref{subsec:normalization} above.

We first state the Duistermaat--Heckman theorem, which was proved in \cite{DH} and later shown to be an instance of a more general principle of {\it equivariant localization} \cite{AtiyahBottLoc}.  Let $(M, \omega)$ be a compact symplectic manifold of dimension $2n$, and suppose that the $r$-dimensional torus $T$ acts smoothly on $M$.  For each $y \in \tot$ we define a vector field $X_y$ on $M$ by \begin{equation} \label{eqn:torus_vf_def} X_y(x) f = \frac{d}{dt} \bigg |_{t = 0} f(e^{ty} \cdot x), \quad x \in M, \ f \in C^\infty(M).\end{equation} We assume that there is a {\it moment map} for the action of $T$ on $M$, that is a smooth function $\phi: M \to \tot^*$ such that\footnotemark \footnotetext{Some authors define the moment map with the opposite sign in (\ref{eqn:moment_map_def}).} \begin{equation} \label{eqn:moment_map_def} \omega(X_y, \, \cdot \,) = - \langle d\phi(\cdot), y \rangle, \qquad y \in \tot.\end{equation} The {\it Liouville measure} $\mu$ on $M$ is given by the volume form $\omega^{\wedge n}/{n!}$.  Then we have the following:

\begin{theorem}[Duistermaat--Heckman] \label{thm:dh}
The integral \begin{equation} \label{eqn:dh-int} \int_{M} e^{it \langle \phi(x), y \rangle} d\mu(x) \end{equation} is exactly equal to its leading-order approximation by the method of stationary phase as $t \to \infty$.\footnotemark \footnotetext{The Duistermaat--Heckman theorem is sometimes stated differently, as follows: every regular value of $\phi$ has a neighborhood in which $\phi_*\mu$ is equal to Lebesgue measure times a polynomial of degree at most $n-r$.  Theorem \ref{thm:dh} follows from this statement by observing that the integral (\ref{eqn:dh-int}), considered as a function of $t$, is the inverse Fourier transform of the measure $\langle \phi, y \rangle_* \mu$, i.e.~the pushforward of $\mu$ by the map $ \langle \phi(\cdot), y \rangle: \tot^* \to \R$.}
\end{theorem}

To prove the exactness of (\ref{eqn:i-approx}) from Theorem \ref{thm:dh}, we reinterpret $I(x,y;t)$ in the language of symplectic geometry.  Here we temporarily stop identifying $\gog$ and $\gog^*$ in order to emphasize the fact that in symplectic geometry it is natural to study {\it co}adjoint rather than adjoint orbits, but we still use $\langle \cdot, \cdot \rangle$ to indicate the duality pairing.

For $x \in \mathfrak{g}$, let $x^* \in \mathfrak{g}^*$ be its dual under the inner product, that is $x^*(y) = \langle x, y \rangle$ for $x, y \in \mathfrak{g}$.  Let $\beta \in \tot^*$.  We define the {\it coadjoint orbit} $$\mathcal{O}^*_{\beta} := \{ \mathrm{Ad}_g^* \beta \ | \ g \in G \} \subset \mathfrak{g}^*,$$ where $\mathrm{Ad}^*$ is the coadjoint representation of $G$ on $\mathfrak{g}^*$, defined by $$\mathrm{Ad}^*_g x^* := (\mathrm{Ad}_{g^{-1}} x)^*.$$

At a point $x^* \in \mathcal{O}^*_{\beta}$, under the usual identification $T_{x^*}\mathfrak{g}^* \cong \mathfrak{g}^*$, we have $$T_{x^*} \mathcal{O}^*_{\beta} \cong \{ [x, y]^* \ | \ y \in \mathfrak{g}^* \}.$$  The {\it Kostant--Kirillov--Souriau form} is the $G$-invariant 2-form $\omega$ on $\mathcal{O}^*_{\beta}$ defined by \begin{equation} \label{eqn:kks_form} \omega_{x^*}([x, y]^*, [x, z]^*) = \langle x, [y, z] \rangle. \end{equation}  This form can be shown by direct computation to be nondegenerate and closed, and it therefore makes $\mathcal{O}^*_{\beta}$ into a symplectic manifold \cite[ch. 1]{AK}, on which the maximal torus $T \subset G$ acts smoothly via the coadjoint representation.

Let $\phi: \mathcal{O}^*_\beta \to \tot^*$ be the orthogonal projection onto $\tot^*$.  We will show that $\phi$ is a moment map for the action of $T$ on $\mathcal{O}^*_\beta$.  Observe that we have \begin{equation} \label{eqn:phi_killing} \langle \phi(x^*),y \rangle = \langle x, y \rangle, \qquad y \in \tot,\ x^* \in \mathcal{O}^*_\beta.  \end{equation}  Moreover, $d\phi([x,z]^*)$ is also just the orthogonal projection of $[x,z]^*$ onto $\tot^*$, so that $$- \langle d\phi([x,z]^*),y \rangle = - \langle [x,z], y \rangle = \langle x, [y, z] \rangle = \omega_{x^*}([x,y]^*, [x,z]^*).$$ A direct computation from (\ref{eqn:torus_vf_def}) gives $X_y(x^*) = [x,y]^*$, so that $\phi$ satisfies (\ref{eqn:moment_map_def}) and therefore is a moment map as desired.

Next we relate the Liouville measure on $\mathcal{O}^*_\beta$ to the Haar measure on $G$.  Let $2n = \dim \mathcal{O}^*_{\beta}.$  The Liouville measure $\omega^{\wedge n}/n!$ is $G$-invariant due to the invariance of $\omega$.  Pulling the Liouville measure back along the map $\mathrm{Ad}^* \beta: G \to \mathcal{O}^*_{\beta}$, we obtain a finite invariant measure on $G$, which by the uniqueness of the Haar measure must equal a constant times $dg$.  Thus for a Borel set $E \subset G$ we have

\begin{equation} \label{eqn:liouville_haar_measures}
\int_E dg = \frac{1}{\mathrm{Vol}_\mu(\mathcal{O}^*_\beta)} \int_{\{ \mathrm{Ad}^*_g \beta \ | \ g \in E \}} \frac{\omega^{\wedge n}}{n!}.
\end{equation}

Putting together (\ref{eqn:phi_killing}) and (\ref{eqn:liouville_haar_measures}) and writing $d\mu := \omega^{\wedge n}/n!$, we can express $I(x, y; t)$ as an integral over $\mathcal{O}_{x^*}^*$:
$$I(x, y; t) = \int_G e^{\frac{1}{t} \langle \mathrm{Ad}_g x, y \rangle} dg = \frac{1}{\mathrm{Vol}_\mu(\mathcal{O}^*_{x^*})} \int_{\mathcal{O}_{x^*}^*} e^{\frac{1}{t} \langle \phi(\beta), y \rangle} d\mu(\beta).$$  This function is analytic in $t$ for $t \in (0, \infty)$, and so by analytic continuation we can equivalently consider $$I(x, y; -it^{-1}) = \frac{1}{\mathrm{Vol}_\mu(\mathcal{O}^*_{x^*})} \int_{\mathcal{O}_{x^*}^*} e^{it \langle \phi(\beta), y \rangle} d\mu(\beta).$$
We have now written $I$ in the form (\ref{eqn:dh-int}), so that by Theorem \ref{thm:dh} and (\ref{eqn:i-approx}) we have $$I(x, y; -it^{-1}) = C \left( \frac{2\pi}{i t} \right)^{(\dim \mathfrak{g} -r)/2}(\Delta_\gog(x) \Delta_\gog(y))^{-1} \sum_{w \in W} \epsilon(w) e^{it \langle w(x), y \rangle}$$ for some constant $C$. The proof concludes by evaluating at $t = -i$ and computing $C$ as in Section \ref{subsec:normalization}.

\subsection{Representation-theoretic proofs} \label{subsec:repthy_proof}

Next we give two different proofs using ideas from representation theory: one proof via the Kirillov character formula, and a second via the character expansion for the heat kernel on the group $G$.

\subsubsection{Via the Kirillov character formula} \label{subsec:KCF-proof}

As discussed above in Section \ref{subsec:repthy-chars}, the Harish-Chandra integral is closely related to the irreducible characters of $G$.  Here we discuss how Theorem \ref{thm:hc} is essentially equivalent to the Kirillov character formula in the case of a compact, connected, semisimple group.

Let $\lambda \in \tot^*$ be the highest weight of an irreducible representation of $G$ with character $\chi_\lambda$, and let $\rho := \frac{1}{2}\sum_{\alpha \in \Phi^+} \alpha$ be half the sum of the positive roots, also called the {\it Weyl vector} of $\gog$.  Let $\mu$ be the Liouville measure associated to the Kostant--Kirillov--Souriau form on the coadjoint orbit $\mathcal{O}^*_{\lambda+\rho}$, and set $$\widehat{\Delta}_\gog(e^z) := \prod_{\alpha \in \Phi^+} (e^{i \langle \alpha, z \rangle/2} - e^{- i \langle \alpha,z \rangle/2}), \qquad z \in \tot.$$ The {\it Kirillov character formula} for compact groups \cite[ch. 5, Theorem 9]{AK} says:

\begin{theorem} \label{thm:kirillov_char-proofs}
For $\xi \in \tot'$,
\begin{equation} \label{eqn:kirillov_char-proofs} \chi_\lambda(e^\xi) = \frac{\Delta_\gog(i\xi)}{\widehat{\Delta}_\gog(e^\xi)} \int_{\mathcal{O}^*_{\lambda + \rho}} e^{i \langle \beta, \xi \rangle} d\mu(\beta). \end{equation}
\end{theorem}

Although here we assume that $G$ satisfies the assumptions of Theorem \ref{thm:hc}, versions of this formula hold in a variety of situations even for non-compact groups; see \cite{AK} for a detailed discussion.  Notably, Frenkel \cite{IF} generalized (\ref{eqn:kirillov_char-proofs}) to a character formula for affine Lie algebras by proving an analogue of the Harish-Chandra formula (\ref{eqn:hc}) for loop groups.

On the other hand, the {\it Weyl character formula} (see e.g. \cite[Theorem 25.4]{BumpLieGroups}) states:
\begin{theorem} \label{thm:weyl_char}
For $\xi \in \tot'$,
\begin{equation} \label{eqn:weyl_char}
\chi_\lambda(e^\xi) = \frac{1}{\widehat{\Delta}_\gog(e^\xi)} \sum_{w \in W} \epsilon(w) e^{i \langle w(\lambda + \rho), \xi \rangle}.
\end{equation}
\end{theorem}

For this reason, the function $\widehat{\Delta}_\gog$ is often called the {\it Weyl denominator}.  Note that either of (\ref{eqn:kirillov_char-proofs}) or (\ref{eqn:weyl_char}) completely determines $\chi_\lambda$. Since the character is analytic it is determined on the maximal torus $\exp(\tot)$ by its values on $\exp(\tot')$, and since it is a class function it is determined on all of $G$ by its restriction to a maximal torus.

One standard proof of Theorem \ref{thm:kirillov_char-proofs} due to Kirillov \cite{Kirillov2} works by applying Theorem \ref{thm:weyl_char} to the right-hand side of (\ref{eqn:hc}). However, the Kirillov formula can also be proven by other methods; it follows, for example, from the equivariant index theorem for a twisted Dirac operator on the coadjoint orbit \cite{BerlineGetzlerVergne}. Nothing stops us, therefore, from using (\ref{eqn:kirillov_char-proofs}) to prove (\ref{eqn:hc}) instead, as there is no circularity involved. For the sake of completeness, we record this proof below.

Equating the right-hand sides of (\ref{eqn:kirillov_char-proofs}) and (\ref{eqn:weyl_char}), writing $i\xi = y$ and $\lambda + \rho = x^*$, and using the relation (\ref{eqn:liouville_haar_measures}) between the Liouville measure on $\mathcal{O}^*_{\lambda+\rho}$ and the Haar measure on $G$, we obtain: \begin{align*} \Delta_\gog(i\xi) \int_{\mathcal{O}^*_{\lambda + \rho}} e^{i \langle \beta, \xi \rangle} d\mu(\beta) &= \mathrm{Vol}_\mu(\mathcal{O}^*_{x^*}) \Delta_\gog(y) \int_G e^{\langle \mathrm{Ad}_g x, y \rangle} dg \\ &= \sum_{w \in W} \epsilon(w) e^{\langle w(x), y \rangle}.\end{align*}

Applying $\Delta_\gog(\partial_{y})$ to the second and third expressions above and evaluating at $y = 0$, we find
\begin{equation} \label{eqn:coadjoint_volume}
\mathrm{Vol}_\mu(\mathcal{O}^*_{x^*}) = \frac{|W|} {[\! [ \Delta_\gog, \Delta_\gog ] \! ]} \Delta_\gog(x),
\end{equation}
which recovers the integral formula (\ref{eqn:hc}) in the case that $y \in \tot'$ and $x^* = \lambda + \rho$, where $\lambda$ is the highest weight of an irreducible representation of $G$.  Since the right-hand side of (\ref{eqn:hc}) is $W$-invariant in $x$, the result also holds for $x$ such that $w(x)^* = \lambda + \rho$, $w \in W$, with $\lambda$ a highest weight.  Such $x$ form a lattice spanning $\tot$, so by scaling $y$ and shifting the scaling onto $x$ we obtain the result for all $y \in \tot'$ and $x$ in a dense subset of $\tot$. Analytic continuation then gives the result for all $x, y \in \tot_\C$, completing the proof of Theorem \ref{thm:hc}.

\begin{remark} \label{rem:O_rho-vol}
Comparing (\ref{eqn:coadjoint_volume}) with the Kirillov character formula (\ref{eqn:kirillov_char-proofs}) for $\lambda = 0$ (corresponding to the trivial representation of $G$), we find
\begin{equation} \label{eqn:O_rho-vol}
\Delta_\gog(\rho) = \frac{[\! [ \Delta_\gog, \Delta_\gog ] \! ]}{|W|},
\end{equation}
so that (\ref{eqn:coadjoint_volume}) becomes
\begin{equation} \label{eqn:coadjoint_volume2}
\mathrm{Vol}_\mu(\mathcal{O}^*_{x^*}) = \frac{\Delta_\gog(x)}{\Delta_\gog(\rho)}.
\end{equation}
\remdone
\end{remark}

\subsubsection{Via the character expansion for the heat kernel on $G$} \label{subsec:char-exp-proof}

This next proof, due to Altschuler and Itzykson \cite{AltItz}, employs both representation theory {\it and} heat kernel asymptotics.  Here we study the heat kernel $K_G(g_1, g_2 ; t)$ on the group rather than on the algebra.  Instead of deriving the integral formula using the fundamental solution for the heat equation on $\tot$, we use the fact that $K_G$ has a known character expansion.  We can then recover the Harish-Chandra formula by integrating this character expansion over the group and taking an appropriate limit.

It is instructive to compare this strategy to the heat equation proof in Section \ref{sec:hc_heat_eqn_proof} above.  There we related the Harish-Chandra integral $\CH(x,y)$ to the $G$-averaged heat kernel $\tilde K$, defined in (\ref{eqn:k-avg-def}) by integrating the heat kernel on $\gog$ over an adjoint orbit.  We then related $\tilde K$ to a particular solution of the heat equation on $\tot$, which we could write down explicitly, yielding an exact formula for $\CH(x,y)$.  In effect the proof in this subsection also works by obtaining an exact expression for $\tilde K$, but via a different method.  Since $\gog = T_{\mathrm{id}_G}G$, the heat kernel on $\gog$ can be recovered from the local behavior of $K_G$ near the identity.  Thus it is sufficient to obtain an exact expression for the integral of $K_G$ over $G$, which we can accomplish using ideas from character theory.

Now we give the proof.  The heat kernel on $G$ can be written \cite{Fegan}:
\begin{equation} \label{eqn:G-heat-kernel}
K_G(g_1,g_2; t) = \sum_{\lambda \in P^+} (\dim V_\lambda) \chi_\lambda(g_1 g_2^{-1}) e^{-t(|\lambda + \rho|^2 - |\rho|^2)/2},
\end{equation}
where $P^+$ is the set of dominant integral weights of $G$.  We will first integrate $K_G$ over $G$ and derive an asymptotic expression for the resulting function in the limit of small $g_1$, $g_2$ and $t$.  Then we will relate this expression to the Harish-Chandra integral.

Integrating (\ref{eqn:G-heat-kernel}) over $G$, from the orthonormality of matrix coefficients of irreducible unitary representations, we obtain:
\begin{equation} \label{eqn:avg-HK-G}
F(g_1, g_2; t) := \int_G K_G(g_1,g g_2 g^{-1}; t) \, dg = \sum_{\lambda \in P^+} \chi_\lambda(g_1) \chi_\lambda(g_2^{-1}) e^{-t(|\lambda + \rho|^2 - |\rho|^2)/2}. \end{equation}
The function $F$ is clearly conjugation-invariant in both $g_1$ and $g_2$, so without loss of generality we may take $g_1 = e^{x},$ $g_2 = e^{y}$ with $x, y \in \tot$.  Using the Weyl character formula (\ref{eqn:weyl_char}) for $\chi_\lambda$, followed by the Poisson summation formula, we have:
\begin{align}
\nonumber F(g_1, g_2; t) &= \frac{e^{t|\rho|^2/2}}{\widehat{\Delta}_\gog(e^{x}) \widehat{\Delta}_\gog(e^{-y})} \sum_{w \in W} \epsilon(w) \sum_{\lambda \in P} e^{i \langle \lambda, x - w(y) \rangle} e^{-t|\lambda|^2/2} \\
\label{eqn:F-expansion} &= \frac{e^{t|\rho|^2/2}\nu^{1/2} (2\pi/t)^{r/2}}{\widehat{\Delta}_\gog(e^{x}) \widehat{\Delta}_\gog(e^{-y})} \sum_{w \in W} \epsilon(w) \sum_{\beta \in 2\pi Q^\vee} e^{-\frac{1}{2t}|x - w(y) + \beta|^2},
\end{align}
where $P$ is the weight lattice, $Q^\vee$ is the coroot lattice, and $\nu$ is the index of $Q^\vee$ in $P$.\footnote{The formula (\ref{eqn:F-expansion}) was derived by Frenkel as part of his affine generalization of the Harish-Chandra integral and Kirillov orbit method \cite[\textsection4.3]{IF}.  The last line of (\ref{eqn:F-expansion}) is closely related to the numerator of the Weyl--Kac character formula for affine Lie algebras.} Now we take a limit:
\begin{equation} \label{eqn:F-limit}
\lim_{\varepsilon \to 0} \varepsilon^{\dim G} F(e^{\varepsilon x}, e^{\varepsilon y}, \varepsilon^2 t) = \frac{\nu^{1/2} (2\pi/t)^{r/2}}{{\Delta}_\gog(x) {\Delta}_\gog(y)} \sum_{w \in W} \epsilon(w) e^{-\frac{1}{2t}|x - w(y)|^2},
\end{equation}
since $\dim G = 2|\Phi^+| + r$, and only the $\beta = 0$ term of (\ref{eqn:F-expansion}) contributes to leading order.

It remains to relate $F(g_1, g_2; t)$ to the integral $I(x,y;t)$.  We achieve this by relating the heat kernel $K_G$ on $G$ to the heat kernel $K$ on $\gog$, as defined above in (\ref{eqn:g-heat-kernel}).  First observe that for $\varepsilon > 0$, $$K_G(e^{x},e^{y} ; t) = \lim_{\eta \to 0} F(e^{x - y}, e^{\varepsilon \eta \rho} ; t).$$  We thus have:
\begin{equation} \label{eqn:KG-limit}
\lim_{\varepsilon \to 0} \varepsilon^{\dim G} K_G(e^{\varepsilon x},e^{\varepsilon y} ; \varepsilon^2 t) = \lim_{\varepsilon \to 0} \lim_{\eta \to 0} \varepsilon^{\dim G} F(e^{\varepsilon(x - y)}, e^{\varepsilon \eta \rho} ; \varepsilon^2 t).
\end{equation}
Now fix $t > 0$ and $x \ne y \in \tot$, and consider the function
\begin{align*}
f(\varepsilon, \eta) &:= \varepsilon^{\dim G} F(e^{\varepsilon(x - y)}, e^{\varepsilon \eta \rho} ; \varepsilon^2 t)
\\ &= \varepsilon^{\dim G} \frac{e^{\varepsilon^2 t|\rho|^2/2} \,  \nu^{1/2} \, (2\pi/\varepsilon^2 t)^{r/2}}{\widehat{\Delta}_\gog(e^{\varepsilon(x-y)}) \widehat{\Delta}_\gog(e^{-\varepsilon\eta\rho})} \sum_{w \in W} \epsilon(w) \sum_{\beta \in 2\pi Q^\vee} e^{-\frac{1}{2\varepsilon^2 t}|\varepsilon(x -y) - \varepsilon \eta w(\rho) + \beta|^2}.
\end{align*}
The denominator $\widehat{\Delta}_\gog(e^{\varepsilon(x-y)}) \widehat{\Delta}_\gog(e^{-\varepsilon\eta\rho})$ vanishes to order $2 |\Phi^+|$ as $\varepsilon \to 0$ and to order $|\Phi^+|$ as $\eta \to 0$, but these ostensible singularities are canceled respectively by powers of $\varepsilon$ in the numerator and by the vanishing of the sum over the Weyl group, so that in fact $f$ is continuous in the $(\varepsilon, \eta)$ plane. Therefore we may exchange the order of limits in (\ref{eqn:KG-limit}) to obtain:
\begin{align}
\nonumber \lim_{\varepsilon \to 0} \varepsilon^{\dim G} K_G(e^{\varepsilon x},e^{\varepsilon y} ; \varepsilon^2 t) &= \frac{ \nu^{1/2} (2 \pi)^{r/2}}{t^{\dim G /2} \Delta_\gog(\rho)} e^{-\frac{1}{2t} |x - y|^2} \\
 \label{eqn:KG-limit-final} &= \frac{\nu^{1/2} (2\pi)^{(\dim \gog + r)/2}}{\Delta_\gog(\rho)}K(x, y; t).
\end{align}
Next we write:
\begin{align*}
\lim_{\varepsilon \to 0} \varepsilon^{\dim G} F(e^{\varepsilon x}, e^{\varepsilon y}, \varepsilon^2 t) &= \lim_{\varepsilon \to 0} \varepsilon^{\dim G} \int_G K_G(e^{\varepsilon x}, g e^{\varepsilon y} g^{-1}, \varepsilon^2 t) \, dg \\
&=  \int_G \, \lim_{\varepsilon \to 0} \varepsilon^{\dim G} K_G(e^{\varepsilon x}, g e^{\varepsilon y} g^{-1}, \varepsilon^2 t) \, dg,
\end{align*}
where we have used bounded convergence to pull the limit under the integral.  Then (\ref{eqn:KG-limit-final}) gives:
\begin{align*}
\lim_{\varepsilon \to 0} \varepsilon^{\dim G} F(e^{\varepsilon x},e^{\varepsilon y} ; \varepsilon^2 t) &= \frac{\nu^{1/2} (2\pi)^{(\dim \gog + r)/2}}{\Delta_\gog(\rho)} \int_G K(\mathrm{Ad}_g x, y; t) \, dg \\
&= \frac{\nu^{1/2} (2\pi)^{r/2}}{t^{\dim \gog / 2} \Delta_\gog(\rho)} e^{-\frac{1}{2t}( |x|^2 + |y|^2 )} I(x, y; t).
\end{align*}
Comparing to (\ref{eqn:F-limit}), evaluating at $t = 1$, and using the observation of Remark \ref{rem:O_rho-vol} that $\Delta_\gog(\rho) = [\! [ \Delta_\gog, \Delta_\gog ] \! ]/|W|$, we arrive at the Harish-Chandra formula (\ref{eqn:hc}) for $x \ne y \in \tot$. Finally, by analytic continuation we can take $x, y \in \tot_\C$ and remove the assumption $x \ne y$, completing the proof.

\subsection{Harmonic analysis proof}
\label{subsec:harmonic-proof}

The Harish-Chandra formula can also be derived from an identity due to Rossmann for the Fourier transform on a semisimple Lie algebra \cite{Rossmann}.  Here we only consider the case where $\gog$ is compact, though Rossmann's formula holds in greater generality.  Let $M : C^\infty(\gog) \to C^\infty(\tot)$ be the operator that sends $f \in C^\infty(\gog)$ to the function $\phi_f$ defined in (\ref{eqn:phif-def}), that is,
$$ M f(x) := \Delta_\gog(x) \int_G f (\mathrm{Ad}_g x) \, dg, \qquad x \in \tot.$$
Let $\mathcal{F}_\gog$ and $\mathcal{F}_\tot$ be the Fourier transform operators on $\gog$ and $\tot$:
\begin{align*}
\mathcal{F}_\gog \varphi (\xi) &:= \left( \frac{1}{2\pi} \right)^{\dim \gog / 2} \int_\gog e^{-i \langle x, \xi \rangle} \varphi(x) \, dx,  \qquad \xi \in \gog, \\
\mathcal{F}_\tot \psi (\xi) &:= \left( \frac{1}{2\pi} \right)^{r / 2} \int_\tot e^{-i \langle x, \xi \rangle} \psi(x) \, dx,  \qquad \quad \ \, \xi \in \tot,
\end{align*}
where $r$ is the rank of $\gog$, and $\varphi$ and $\psi$ are appropriate functions (say, $L^2$ or Schwartz class) on $\gog$ and $\tot$ respectively.  We will use the following version of Rossmann's formula:

\begin{proposition}[Rossmann] \label{prop:rossmann}
\begin{equation} \label{eqn:rossmann}
M \circ \mathcal{F}_\gog = e^{i \pi (\dim \gog - r)/ 4} \mathcal{F}_\tot \circ M
\end{equation}
as operators on $C_c^\infty(\gog)$.
\end{proposition}

We will give an elegant proof of (\ref{eqn:rossmann}) due to Vergne \cite{Vergne-Rossmann}.  It relies on a classical fact relating the Fourier transform to the time-evolution operator for the quantum harmonic oscillator.

\begin{proof}[Proof of Proposition \ref{prop:rossmann}]
We start by recalling some basic facts about the Hermite functions, which are defined by
$$ \psi_k(x) := (-1)^k (2^k k! \sqrt{\pi})^{-1/2} e^{x^2 / 2} \frac{d^k}{dx^k} e^{-x^2}, \qquad x \in \R,$$
for $k = 0, 1, 2, \hdots$.  The Hermite functions are an orthonormal basis of $L^2(\R)$ and are joint eigenfunctions of both the Fourier transform and the Schr\"odinger operator $-d^2/dx^2 + x^2$.  In particular, the $k^\mathrm{th}$ Hermite function satisfies the differential equation
$$ \left(- \frac{d^2}{dx^2} + x^2 \right) \psi_k(x) = (2k + 1) \psi_k(x) $$
as well as the integral equation
$$ \frac{1}{\sqrt{2\pi}} \int_{-\infty}^\infty e^{-ixy} \psi_k(y) \, dy = (-i)^k \psi_k(x). $$

If $(x_1, \hdots, x_n)$ are coordinates with respect to an orthonormal basis on a Euclidean space $V \cong \R^n$, the multivariate Hermite functions
$$ \psi_{\alpha}(x) := \prod_{j = 1}^n \psi_{\alpha_j} (x_j)$$
are an orthonormal basis of $L^2(V)$, as $\alpha = (\alpha_1, \hdots, \alpha_n)$ runs over multi-indices.  By reducing to the one-dimensional case, it is easily verified that
$$\mathcal{F}_V \psi_\alpha(x) := \left( \frac{1}{2\pi} \right)^{n/2} \int_{V} e^{-i(x,y)} \psi_\alpha(y) \, dy = (-i)^{|\alpha|} \psi_\alpha(x),$$
where $(\cdot, \cdot)$ is the inner product on $V$, and $dy$ is the associated Lebesgue measure.  Defining
$$H_V := \sum_{j=1}^n \left( - \frac{\partial^2}{\partial x_j^2} + x_j^2 \right),$$
one also finds
$$ H_V \psi_\alpha(x) = (2 |\alpha| + n) \psi_\alpha(x).$$
We can extend $H_V$ to a self-adjoint operator on $L^2(V)$ via its action on the basis $\psi_\alpha$.  Then for $t \in \R$, we can define a unitary operator $e^{-i t H_V }$ by setting
$$e^{-i t H_V } \psi_\alpha = e^{-it(2 |\alpha| + 1)} \psi_\alpha.$$
Comparing the eigenvalues of $H_V$ and $\mathcal{F}_V$ on $\psi_\alpha$, we find 
\begin{equation} \label{eqn:harmonic-fourier}
e^{-i \pi H_V / 4} = e^{- i \pi n / 4} \mathcal{F}_V.
\end{equation}

The identity (\ref{eqn:harmonic-fourier}) is well known to physicists and has an interesting physical interpretation.  The operator $H_V$ can be viewed as the Hamiltonian for a quantum harmonic oscillator on $V$.  Then $e^{-it H_V}$ is the time evolution operator for this quantum-mechanical system, which maps the initial wavefunction to the wavefunction at time $t$.  In this context, (\ref{eqn:harmonic-fourier}) says that the wavefunction of the oscillator at $t = \pi/4$ is equal to the Fourier transform of the wavefunction at $t = 0$, up to a (physically meaningless) phase shift.

When $V = \gog$ or $\tot$ with the invariant inner product $\langle \cdot, \cdot \rangle$, we have
$$H_\gog = - \omega(\partial) + |x|^2, \qquad H_\tot = - \bar \omega(\partial) + |x|^2,$$
where $\omega(\partial)$ and $\bar \omega(\partial)$ are the respective Laplacian operators associated to the inner product.  Using the formula (\ref{eqn:radial-laplacian}) for the radial part of the Laplacian, we find immediately that
\begin{equation*} 
M H_\gog \varphi (x) = H_\tot M \varphi (x)
\end{equation*}
for all $\varphi \in C^\infty(\gog)$ and $x \in \tot$ with $\Delta_\gog(x) \not = 0$, and therefore for all $x \in \tot$ by continuity.  Now taking $\varphi$ compactly supported and expanding in the basis of Hermite functions, we find
$$M e^{-i \pi H_\gog / 4} \varphi =  e^{-i \pi H_\tot / 4} M \varphi,$$
and comparing to (\ref{eqn:harmonic-fourier}), we obtain
\begin{align*}
M \mathcal{F}_\gog \varphi &= e^{i \pi \dim \gog / 4} M e^{-i \pi H_\gog / 4} \varphi = e^{i \pi \dim \gog / 4} e^{-i \pi H_\tot / 4} M \varphi \\& = e^{i \pi (\dim \gog - r)/ 4} \mathcal{F}_\tot M \varphi
\end{align*}
as desired.
\end{proof}

The Harish-Chandra formula (\ref{eqn:hc}) now follows quite directly by restricting Rossmann's formula (\ref{eqn:rossmann}) to the space of $\mathrm{Ad}$-invariant test functions.  The Lebesgue measure on $\gog$ has a well-known ``polar coordinates'' decomposition: for $f$ a continuous integrable function on $\gog$, we have
\begin{equation} \label{eqn:alg-lebesgue-polar}
 \int_\gog f(x) \, dx = c \int_\tot \Delta_\gog(y)^2 \int_G f(\mathrm{Ad}_g y) \, dg \, dy,
 \end{equation}
where $c$ is a constant; see e.g.~\cite[ch.~1, \textsection5]{SH}.  For an $\mathrm{Ad}$-invariant test function $\varphi \in C_c^\infty(\gog)$, we can then write the left-hand side of (\ref{eqn:rossmann}) as:
\begin{align}
\nonumber
M \mathcal{F}_\gog \varphi(x) &= \left( \frac{1}{2\pi} \right)^{\dim \gog / 2} \Delta_\gog(x) \int_G \int_\gog e^{-i \langle \mathrm{Ad}_g x, y \rangle} \varphi(y) \, dy \, dg \\
\nonumber &= \left( \frac{1}{2\pi} \right)^{\dim \gog / 2} \Delta_\gog(x) \int_\gog \mathcal{H}(-ix,y) \varphi(y) \, dy \\
 \label{eqn:ross-leftside}
&= c' \int_\tot \Delta_\gog(x) \Delta_\gog(y)^2 \, \mathcal{H}(-ix,y) \varphi(y) \, dy
\end{align}
for some constant $c'$.

Now we rewrite the right-hand side of (\ref{eqn:rossmann}).  Observe that the invariance of $\varphi$ implies $M \varphi = \Delta_\gog \varphi$, and $\varphi(w(y)) = \varphi(y)$ for $y \in \tot$ and $w \in W$.  Thus we can write the right-hand side of (\ref{eqn:rossmann}) as:
\begin{align} 
\nonumber
e^{i \pi (\dim \gog - r)/ 4} \mathcal{F}_\tot M \varphi(x) &= e^{i \pi (\dim \gog - r)/ 4}  \left( \frac{1}{2\pi} \right)^{r / 2} \int_\tot e^{-i \langle x, y \rangle} \Delta_\gog(y) \varphi(y) \, dy \\
\nonumber &= c'' \int_\tot e^{-i \langle x, y \rangle} \Delta_\gog(y) \sum_{w \in W} \varphi(w(y)) \, dy \\
\label{eqn:ross-rightside} &= c'' \int_\tot \Delta_\gog(y) \sum_{w \in W} \epsilon(w) \, e^{-i \langle x, w(y) \rangle}  \varphi(y) \, dy
\end{align}
for some constant $c''$, where in the last line we have used the invariance of $\varphi$ and the skewness of $\Delta_\gog$.

Observing that the integrands in (\ref{eqn:ross-leftside}) and (\ref{eqn:ross-rightside}) are $W$-invariant functions of $y \in \tot$ and that the restriction of $\varphi$ to $\tot$ is an arbitrary $W$-invariant test function in $C_c^\infty(\tot)$, we conclude that the integrands must be equal, up to normalization.  After some manipulations, we obtain
$$ \Delta_\gog(x) \Delta_\gog(y) \, \mathcal{H}(ix,y) = C \sum_{w \in W} \epsilon(w) \, e^{i \langle x, w(y) \rangle}$$
for $x, y \in \tot$ and some constant $C$.  The proof then concludes by analytically continuing to $x, y \in \tot_\C$ and determining the constant $C$ as in the preceding sections.

\subsection{Further proofs of the HCIZ formula}
\label{subsec:further-hciz-proofs}

Finally, we give two additional proofs of the HCIZ formula (\ref{eqn:hciz}), which make use of facts that are somewhat specific to the unitary case.  The first is due to Balantekin \cite{Balantekin} and uses a character expansion derived by Itzykson and Zuber \cite{IZ}.  The second is an inductive proof that relies on a reformulation of (\ref{eqn:hciz}) as an integral over a convex polytope.  It is a synthesis of two closely related derivations due to Shatashvili \cite[\textsection3]{Sha} and to Faraut \cite[\textsection2.2]{Faraut}.\footnote{Faraut in turn credits the main idea of his proof to much earlier work of Gel'fand and Naimark \cite[ch.~II, \textsection9.3]{GN}, where they used a similar argument to derive a determinantal formula for the spherical functions on the symmetric space $\GL(N, \C) / \U(N, \C)$.}

\subsubsection{Character expansion and Cauchy--Binet formula}

In the unitary case, we can obtain a character expansion for the HCIZ integral
\begin{equation} \label{eqn:hciz-charexp}
I(A, B) := \int_{\U(N)} e^{\mathrm{tr} (AUBU^\dagger)} dU
\end{equation}
via a more elementary method than the heat kernel analysis used in Section \ref{subsec:char-exp-proof}.  We start by recalling some basic results from representation theory, for which we refer the reader to \cite{EtingofRT, FH-RT}.

The irreducible polynomial representations of the general linear group $\GL(N, \C)$ are labelled by {\it Young diagrams}, which we think of as vectors $\lambda = (\lambda_1, \hdots, \lambda_N)$ with weakly decreasing nonnegative integer coordinates.  If $(x_1, \hdots, x_N)$ are the eigenvalues of $X \in \GL(N, \C)$, the character of the $\lambda$-representation is given by $\chi_\lambda(X) = s_\lambda(x_1, \hdots, x_N)$, where $s_\lambda$ is the {\it Schur polynomial} defined by
$$s_\lambda(x_1, \hdots, x_N) := \frac{ \det \big[ x_i^{\lambda_j + j -1} \big]_{i,j=1}^N }{\Delta(x)}.$$
Write $| \lambda | := \lambda_1 + \cdots + \lambda_N$ and $\lambda' := ( \lambda_N + N - 1,  \, \lambda_{N-1} + N - 2, \, \hdots \, , \, \lambda_1).$  The Frobenius formula, a fundamental identity in the representation theory of $\GL(N,\C)$, gives a character expansion for powers of the trace of $X$:
\begin{equation} \label{eqn:frobenius}
\mathrm{tr}(X)^n = \sum_{\substack{ \lambda \\ | \lambda | = n} } \frac{n ! \, \Delta( \lambda')}{\prod_{j=1}^N (\lambda_j + j -1)! } \, \chi_\lambda(X).
\end{equation}

We will use (\ref{eqn:frobenius}) to derive a character expansion for the HCIZ integral.  Expanding the exponential in (\ref{eqn:hciz-charexp}) as a power series and using bounded convergence to pull the summation out of the integral, we can write
$$ I(A,B) = \sum_{n=0}^\infty \frac{1}{n!} \int_{U(N)} \mathrm{tr} (AUBU^\dagger)^n \, dU.$$
Applying (\ref{eqn:frobenius}) to the integrand and again using bounded convergence, we obtain
$$ I(A,B) = \sum_{n=0}^\infty \sum_{|\lambda| = n} \frac{\Delta( \lambda')}{\prod_{j=1}^N (\lambda_j + j -1)! } \int_{U(N)} \chi_\lambda (AUBU^\dagger) \, dU,$$
and the Schur orthogonality relations for the characters $\chi_\lambda$ give
$$ \int_{U(N)} \chi_\lambda (AUBU^\dagger) \, dU = \int_{U(N)} \chi_\lambda (AU) \chi_\lambda (BU^\dagger) \, dU = \frac{\chi_\lambda (A) \, \chi_\lambda (B)}{d_\lambda},$$ where $d_\lambda := \Delta(\lambda') / \prod_{j = 1}^{N-1} j!$ is the dimension of the $\lambda$-representation of $\GL(N, \C)$.  Putting together all of the above, we find:
\begin{align}
\nonumber
I(A, B) &= \sum_{n=0}^\infty \sum_{|\lambda| = n} \left( \prod_{j=1}^N \frac{(j-1)!}{(\lambda_j + j -1)!} \right) \chi_\lambda (A) \, \chi_\lambda (B) \\
\nonumber
&= \sum_{n=0}^\infty \sum_{|\lambda| = n} \left(\prod_{j=1}^N \frac{(j-1)!}{(\lambda_j + j -1)! } \right) s_\lambda (a_1, \hdots, a_N) \, s_\lambda (b_1, \hdots, b_N) \\
\label{eqn:hciz-charseries}
&= \sum_{n=0}^\infty \sum_{|\lambda| = n} \left(\prod_{j=1}^N \frac{(j-1)!}{(\lambda_j + j -1)! } \right) \frac{\det \big[ a_i^{\lambda_j + j -1} \big]_{i,j=1}^n \det \big[ b_i^{\lambda_j + j -1} \big]_{i,j=1}^n}{\Delta(A) \Delta(B)},
\end{align}
where $a_1 > \hdots > a_N$, $b_1 > \hdots > b_N$ are the eigenvalues of $A$ and $B$.

The last step of the proof is to re-sum the series in (\ref{eqn:hciz-charseries}) using the following generalization of the Cauchy--Binet formula in linear algebra (see \cite[Appendix B]{Balantekin}).  For an analytic function $$f(z) = f_0 + f_1 z + f_2 z^2 + \cdots,$$ we have the identity:
\begin{equation} \label{eqn:gen-cauchy-binet}
\det \big [ f(x_i y_j) \big]_{i,j = 1}^N = \sum_{k_1 > \cdots > k_N} f_{k_1} f_{k_2} \cdots f_{k_N} \det \big [ x_i^{k_j} \big]_{i,j = 1}^N \det \big [ y_i^{k_j} \big]_{i,j = 1}^N.
\end{equation}
Applying (\ref{eqn:gen-cauchy-binet}) to (\ref{eqn:hciz-charseries}) with $f(z) = e^z$, we obtain the HCIZ formula (\ref{eqn:hciz}).

\subsubsection{Inductive proof via Gel'fand--Tsetlin polytopes}

This last proof begins with a reformulation of the HCIZ integral as an integral over a {\it Gel'fand--Tsetlin polytope}, a type of convex polytope that plays an important role in the representation theory of the unitary group.  The linear inequalities that define these polytopes have a recursive structure, which allows us to prove the integral formula (\ref{eqn:hciz}) by induction.  Although the proof exploits a very specific relationship between Gel'fand--Tsetlin polytopes and coadjoint orbits of $\U(N)$, there are analogous polytopes that play the same role for other compact groups \cite{Defosseux, Littelmann-CCP}.  In principle, therefore, one could give similar inductive proofs of the integral formulae for the other classical groups derived below in Section \ref{ch:specific-integrals}. However, it seems unlikely that such an approach could be used to prove the general Harish-Chandra formula (\ref{eqn:hc}), since for the exceptional groups there is no way to induct on the rank.

We start by quickly recalling some definitions and results related to Gel'fand--Tsetlin polytopes.  A {\it Rayleigh triangle} is a triangular array of real numbers $R = (R_{i,j})_{1 \le i \le j \le N}$ satisfying the {\it interlacing relations}
\begin{equation} \label{eqn:interlacing-relns}
R_{i-1,j} \ge R_{i,j-1} \ge R_{i,j}, \qquad 1 < i \le j \le N.
\end{equation}
The vector $R_{\bullet,j} := (R_{1,j},\, \hdots,\, R_{j,j}) \in \R^j$ is called the $j^\mathrm{th}$ row of $R$, and $R_{\bullet,N} \in \R^N$ is called the top row.  If we fix the top row of $R$ by setting $R_{\bullet,N} = \lambda$ for some $\lambda \in \R^N$ with $\lambda_1 \ge \cdots \ge \lambda_N$, we can regard the remaining numbers $R_{i,j}$, $j \le N-1$ as coordinates of a point in $\R^{N(N-1)/2}$.

The {\it Gel'fand--Tsetlin polytope} $GT(\lambda)$ is the space of all Rayleigh triangles with top row $\lambda$.  Concretely, it is the convex polytope in $\R^{N(N-1)/2}$ whose facets are cut out by the interlacing inequalities (\ref{eqn:interlacing-relns}) after fixing $R_{\bullet,N} = \lambda$.

A classical result in linear algebra, the Cauchy--Rayleigh Interlacing Theorem, describes a very natural mapping from Hermitian matrices to Rayleigh triangles.  Let $\mathrm{Her}(N)$ denote the space of $N$-by-$N$ Hermitian matrices.  Given $X \in \mathrm{Her}(N)$ and $k \le N$, write $X[k] \in \mathrm{Her}(k)$ for the $k^\mathrm{th}$ leading submatrix of $X$ (that is, the $k$-by-$k$ submatrix in the upper left corner of $X$).

\begin{theorem}[Cauchy--Rayleigh Interlacing Theorem] \label{thm:rayleigh}
Let $X \in \mathrm{Her}(N)$, and let $\lambda_{1,k} \ge \hdots \ge \lambda_{k,k}$ be the eigenvalues of the leading submatrix $X[k]$, for $1 \le k \le N$.  Then $(\lambda_{j,k})_{1 \le j \le k \le N}$ is a Rayleigh triangle.
\end{theorem}

Write $\mathcal{R} : X \mapsto (\lambda_{j,k})_{1 \le j \le k \le N}$ for the map that takes a Hermitian matrix to the ordered array of eigenvalues of its leading submatrices, and write $\mathcal{O}_\Lambda \subset \mathrm{Her}(N)$ for the unitary conjugation orbit of $\Lambda := \mathrm{diag}(\lambda)$.  Then Rayleigh's Theorem can be rephrased as saying that $\mathcal{R}$ maps $\mathcal{O}_\Lambda$ into $GT(\lambda)$.  A remarkable theorem of Baryshnikov \cite{Barysh} states that in fact this map is {\it measure preserving}.

\begin{theorem}[Baryshnikov] \label{thm:R-pushfwd}
Under $\mathcal{R}$, the invariant probability measure $\mu_\Lambda$ on $\mathcal{O}_\Lambda$ pushes forward to the uniform probability measure on $GT(\lambda)$.  That is,
\begin{equation} \label{eqn:mu-pushfwd}
\mathcal{R}_* \mu_\Lambda = \frac{1}{\mathrm{Vol}(GT(\lambda))} dR,
\end{equation}
where $dR$ is Lebesgue measure on $GT(\lambda)$,\footnote{Technically $dR$ indicates the restriction to $GT(\lambda)$ of the Lebesgue measure on the minimal affine subspace of $\R^{N(N-1)/2}$ containing $GT(\lambda)$, since if not all $\lambda_j$ are distinct then $GT(\lambda)$ has dimension less than $N(N-1)/2$.} and
\begin{equation} \label{eqn:GT-vol}
\mathrm{Vol}(GT(\lambda)) = \prod_{\substack{1 \le i < j \le N \\ \lambda_i \ne \lambda_j}} \frac{\lambda_i - \lambda_j}{j-i}
\end{equation}
is the Lebesgue volume of $GT(\lambda)$. 
\end{theorem}

See \cite[Proposition 4.6 and Lemma 1.12]{Barysh} as well as \cite{Neretin, Z2} for proofs and further discussion of Theorem \ref{thm:R-pushfwd}.  Observe that when all $\lambda_j$ are distinct, (\ref{eqn:GT-vol}) reads:
\begin{equation} \label{eqn:GT-vol-reg}
\mathrm{Vol}(GT(\lambda)) = \frac{\Delta(\lambda)} { \prod_{p=1}^{N-1} p! }.
\end{equation}
 
The {\it type} of a Rayleigh triangle $R$ is the vector
$$\mathrm{type}(R) := \bigg( R_{1,1},\, R_{1,2} + R_{2,2} - R_{1,1},\, \hdots,\, \sum_{i=1}^N R_{i,n} - \sum_{j=1}^{N-1} R_{j,N-1} \bigg) \in \R^N.$$
Observe that for $X \in \mathrm{Her}(N)$, the $k^\mathrm{th}$ coordinate of $\mathrm{type}(\mathcal{R}(X))$ is
$$\mathrm{tr}(X[k])-\mathrm{tr}(X[k-1])=X_{kk},$$
so that $\mathrm{type}(\mathcal{R}(X))$ is just the diagonal of $X$.

With the above definitions in hand, we can use Theorem \ref{thm:R-pushfwd} to rewrite the HCIZ integral as an integral over a Gel'fand--Tsetlin polytope.  Taking $A = \mathrm{diag}(a)$, $B = \mathrm{diag}(b)$ for $a, b \in \R^N$ with $a_1 > \hdots > a_N$, $b_1 > \hdots > b_N$, we have
\begin{align}
\nonumber
I(A,B) &= \int_{\U(N)} e^{\mathrm{tr(AUBU^\dagger)}} dU = \int_{\mathcal{O}_B} e^{\sum_{j = 1}^N a_j X_{jj}} d\mu_B(X) \\
\label{eqn:hciz-gt}
&= \frac{1}{\mathrm{Vol}(GT(b))} \int_{GT(b)} e^{a \cdot \mathrm{type}(R)} dR,
\end{align}
where $\mu_B$ is the invariant probability measure on the orbit $\mathcal{O}_B$, and $dR$ is Lebesgue measure on $GT(b)$.

We now prove the HCIZ formula (\ref{eqn:hciz}) by induction on $N$.  In the base case, $N = 1$, the matrices $A$, $B$ and $U$ are just complex numbers, and there is nothing to prove.  Assume now that (\ref{eqn:hciz}) holds for $\U(N-1)$; we will prove that it holds for $\U(N)$.

Write $y := R_{\bullet,N-1} \in \R^{N-1}$ for the $(N-1)^\mathrm{th}$ row of $R \in GT(b)$, and write $\Delta_{N-1}(y) := \prod_{1 \le i < j \le N-1} (y_i - y_j)$.  Let $a' := (a_1, \hdots, a_{N-1}) \in \R^{N-1}$ be the projection of $a$ onto the first $N-1$ coordinates, and set $A' := \mathrm{diag}(a')$, $Y := \mathrm{diag}(y)$.  Integrating individually over each entry $y_j$, we can rewrite (\ref{eqn:hciz-gt}) to get:
\begin{align*}
I(A,B) &= \frac{1}{\mathrm{Vol}(GT(b))} \int_{b_2}^{b_1} \cdots \int_{b_N}^{b_{N-1}} e^{a_N (\mathrm{tr}(B) - \sum y_j)} \int_{GT(y)} e^{a' \cdot \mathrm{type}(Q)} dQ \, dy_{N-1} \hdots dy_1 \\
&= \frac{1}{\mathrm{Vol}(GT(b))} \int_{b_2}^{b_1} \cdots \int_{b_N}^{b_{N-1}} e^{a_N (\mathrm{tr}(B) - \sum y_j)} \mathrm{Vol}(GT(y)) \int_{\U(N-1)} e^{\mathrm{tr}(A' U Y U^\dagger)} dU \, dy_{N-1} \hdots dy_1 \\
&= \frac{(N-1)!}{\Delta(b)} \int_{b_2}^{b_1} \cdots \int_{b_N}^{b_{N-1}} e^{a_N (\mathrm{tr}(B) - \sum y_j)} \Delta_{N-1}(y) \int_{\U(N-1)} e^{\mathrm{tr}(A' U Y U^\dagger)} dU \, dy_{N-1} \hdots dy_1,
\end{align*}
where $dQ$ is Lebesgue measure on $GT(y)$.  Using the inductive hypothesis, this last equation becomes
$$ I(A,B) =  \frac{\prod_{p=1}^{N-1} p!}{\Delta(b) \Delta_{N-1}(a)} \int_{b_2}^{b_1} \cdots \int_{b_N}^{b_{N-1}} e^{a_N (\mathrm{tr}(B) - \sum y_j)} \det \big[ e^{a_i y_j} \big]_{i,j=1}^{N-1} \, dy_{N-1} \hdots dy_1.$$
The remainder of the proof is a direct calculation: expanding the determinant as a sum over permutations and integrating term by term, we obtain the integral formula (\ref{eqn:hciz}).

\section{Integral formulae for specific groups}
\label{ch:specific-integrals}

This section provides detailed derivations of the specific forms of the Harish-Chandra integral formula (\ref{eqn:hc}) for all compact classical groups.  We also discuss how Harish-Chandra's original result, which assumes that the group in question is connected and semisimple, extends naturally to a formula for integrals over arbitrary compact Lie groups.  While some of the specific integral formulae derived in this section have appeared previously, for example in \cite{PEDZ, FILZ}, here we show the complete calculations in a pedagogical style, making an effort not to omit details.  We first study integrals over the groups $\U(N)$ and $\SU(N)$, showing how the HCIZ formula (\ref{eqn:hciz}) can be derived from (\ref{eqn:hc}).  Next we study integrals over $\SO(2N)$ and $\O(2N)$, $\SO(2N+1)$ and $\O(2N+1)$, $\mathrm{USp}(N)$, and finally arbitrary compact groups.

\subsection{The HCIZ integral: the cases $G=\U(N)$ and $G=\SU(N)$}
\label{sec:un_sun}

We start by showing how we can recover the HCIZ formula (\ref{eqn:hciz}) from (\ref{eqn:hc}).  The computation illustrates how to specialize (\ref{eqn:hc}) to the case of a particular group $G$, and our calculations of formulae for other classical groups in later sections will follow the same basic procedure.
 
The reader may have already remarked that Theorem \ref{thm:hc} assumes that $G$ is semisimple, while the group $\U(N)$ is not.  This turns out not to matter: Theorem \ref{thm:hc} actually allows us to compute analogous integrals over arbitrary compact $G$, which may be neither semisimple nor connected, by reducing to an integral over a connected semisimple group.  We discuss this in detail in Section \ref{sec:arbitrary_cpt} below, but for now we simply note that the formula (\ref{eqn:hc}) still works as written for $\U(N)$.

The derivation of (\ref{eqn:hciz}) from (\ref{eqn:hc}) goes as follows.  Let $G = \U(N)$.  Then $\mathfrak{g} = \mathfrak{u}(N)$ consists of $N$-by-$N$ skew-Hermitian matrices, and its complexification $\gog_\C = \mathfrak{gl}(N,\C)$ consists of all $N$-by-$N$ complex matrices.  Note that a skew-Hermitian matrix is just a Hermitian matrix multiplied by $i$.  For the Cartan subalgebra $\tot$, we can choose any maximal subspace of pairwise commuting matrices in $\gog$, so we take $\tot$ to be the real vector space of all $N$-by-$N$ diagonal matrices with purely imaginary entries, which we may identify with $\R^N$.  Then $\tot_\C \cong \C^N$ is the space of $N$-by-$N$ complex diagonal matrices.  We also take $\langle x, y \rangle = \mathrm{tr}(x^\dagger y)$, the Hilbert--Schmidt inner product, so that the matrices $e_j := i E_{jj}$, $1 \le j \le N$ are an orthonormal basis of $\tot$, where $E_{jj}$ is the matrix with a single 1 in the $j^{\textrm{th}}$ position on the diagonal and zeros everywhere else.

Our first step is to write down the discriminant $\Delta_\gog$ in a more explicit form, for which we must determine the roots of $\U(N)$.  If we identify the roots with elements of $\tot$ as in Section \ref{sec:lie-defs}, then the roots correspond to diagonal matrices $\alpha \in \tot$ satisfying
\begin{equation} \label{eqn:a-roots-relation}
hx - xh = i \, \mathrm{tr}(\alpha^\dagger h) \, x 
\end{equation}
for some $N$-by-$N$ complex matrix $x$ and all diagonal matrices $h$.  The relation (\ref{eqn:a-roots-relation}) holds exactly when $\alpha = e_j - e_k$ for $j \ne k$, and $x$ is a complex multiple of $E_{jk}$, the elementary matrix with a single $1$ in the $j^{\textrm{th}}$ row and $k^{\textrm{th}}$ column and zeros everywhere else.  The matrix $E_{jk}$ spans the one-dimensional root space $\gog_\alpha \subset \gog_\C$.

These matrices $\Phi = \{ e_j - e_k \ | \ j \ne k \}$ comprise the root system $A_{N-1}$ of $\U(N)$.  For the positive roots, we can choose the matrices $e_j - e_k$ with $j < k$.  Writing $A = \sum_j a_j e_j \in \tot_\C$,\footnote{Note that in (\ref{eqn:hciz}) we took $a_j = A_{jj}$, the $j^{\mathrm{th}}$ diagonal entry of $A$, whereas the definition we have just given differs by a factor of $i$: here $a_j = iA_{jj}$. This choice is convenient because it allows us to identify $\tot$ with the elements of $\tot_\C$ that have real, rather than imaginary, coordinates. One may check that the formula (\ref{eqn:hciz}) does not depend on which of these two choices we make, since the factors of $i$ cancel. In either case we define $\Delta(A) := \prod_{j < k} (a_j - a_k)$.} we then have $$\Delta_\gog(A) = \prod_{j < k} \mathrm{tr}[ (e_j - e_k)^\dagger A ] = \prod_{j < k}(a_j - a_k) = \Delta(A),$$ so that the discriminant is the Vandermonde.

Next we consider the Weyl group $W$, which is generated by reflections in the hyperplanes orthogonal to the roots.  The hyperplane orthogonal to $e_j - e_k$ is
$$H_{jk} = \bigg\{ \ \sum_{i=1}^N a_i e_i \in \tot_\C \ \bigg | \ a_j - a_k = 0 \ \bigg \},$$
and the reflection through $H_{jk}$ transposes $e_j$ and $e_k$ while leaving the other basis vectors fixed.  Thus we find that the Weyl group of $\U(N)$ is the symmetric group $S_N$, and it acts on $\tot_\C \cong \C^N$ by permuting the coordinates.  We identify $W = S_N$ with the group $\mathrm{Perm}(N) \subset \U(N)$ of $N$-by-$N$ permutation matrices, so that $W$ acts on $\tot_\C$ by matrix conjugation.  For $w \in S_N$, $\epsilon(w)$ is equal to the sign of the permutation, which is the determinant of the corresponding permutation matrix.  Taking $x = A^\dagger$ and $y = B$, and using the fact that $|S_N| = N!$, (\ref{eqn:hc}) thus becomes \begin{equation} \label{eqn:mid_un} \Delta(A) \Delta(B) \int_{\U(N)} e^{\mathrm{tr} (AUBU^\dagger)} dU = \frac{[\![ \Delta, \Delta ]\!]}{N!} \sum_{P \in \mathrm{Perm}(N)} (\det P) e^{\mathrm{tr}(APBP^\dagger)}.\end{equation}
We observe that the sum on the right-hand side is exactly the Leibniz formula for the determinant of the matrix $\big[ e^{a_i b_j} \big]_{i,j = 1}^N$, and assuming that $a_i \ne a_j$, $b_i \ne b_j$ for $i \ne j$, we can divide through on both sides by $\Delta(A) \Delta(B)$.

It now only remains to compute $\Delta(\partial)\Delta(x) |_{x=0}$.  The Vandermonde is a homogeneous polynomial of order $N(N-1)/2$, so we can write
$$\Delta(x) = \sum_{|\beta| = \frac{N(N-1)}{2}} \pi_\beta x^\beta$$
for some coefficients $\pi_\beta$, where $\beta$ runs over multi-indices whose components sum to $N(N-1)/2$.
There is a well-known determinantal formula for the Vandermonde: $\Delta(x) = \det \big [ x_i^{N-j} \big ]_{i,j=1}^N.$  Expanding this determinant as a sum over permutations, we find that there are $N!$ nonzero coefficients $\pi_\beta$, each with $\beta = \sigma(0,...,N-1)$ for some $\sigma \in S_N$, and \begin{equation} \label{eqn:pibeta} \pi_\beta = \mathrm{sgn}(\sigma) = \pm 1. \end{equation}  Plugging this into (\ref{eqn:normconst}), we have \begin{equation} \label{eqn:un_norm} [\![ \Delta, \Delta ]\!] = \sum_{\sigma \in S_N} \left( \mathrm{sgn}(\sigma)^2 \prod_{p = 0}^{N-1} \sigma(p)! \right) = N! \prod_{p=1}^{N-1} p! = \prod_{p=1}^{N} p!\end{equation} which gives the correct normalization for the HCIZ formula (\ref{eqn:hciz}), and we are done.

For the case $G=\SU(N)$, the formula turns out to be exactly the same: \begin{equation} \label{eqn:hciz_su} \int_{\SU(N)} e^{\mathrm{tr} (AUBU^\dagger)} dU = \left( \prod_{p=1}^{N-1}p! \right) \frac{\det \big[ e^{a_i b_j}\big]_{i,j = 1}^N}{\Delta(A) \Delta(B)}.\end{equation} Since $\SU(N)$ is semisimple unlike $\U(N)$, (\ref{eqn:hciz_su}) is indeed a direct special case of (\ref{eqn:hc}) if $A$ and $B$ are taken to be {\it traceless} diagonal matrices.\footnotemark \footnotetext{The complexified Lie algebra $\mathfrak{su}(N) \otimes \mathbb{C} = \mathfrak{sl}(N, \mathbb{C})$ is the space of $N$-by-$N$ traceless matrices, among which the diagonal matrices form the Cartan subalgebra $\tot_\C$.  However, it is easily checked that (\ref{eqn:hciz_su}) holds even if $A$ and $B$ are not traceless.}

The similarity of the formulae for $\U(N)$ and $\SU(N)$ is due to the fact that these two groups have the same root system $A_{N-1}$.  The Lie algebra of $\SU(N)$ is the space $\mathfrak{su}(N)$ of traceless skew-Hermitian matrices, and its complexification $\mathfrak{sl}(N, \mathbb{C})$ is the space of traceless complex matrices.  The group $\U(N)$ can be decomposed as a semidirect product \mbox{$\U(N) \cong \SU(N) \rtimes \U(1),$} which gives a corresponding decomposition of the Lie algebra $\mathfrak{u}(N) \cong \mathfrak{su}(N) \oplus \mathfrak{u}(1)$ and its complexification $\mathfrak{gl}(N, \mathbb{C}) \cong \mathfrak{sl}(N, \mathbb{C}) \oplus \mathbb{C}$.  The factor $\C$ corresponds to the trace and lies in the center of $\mathfrak{gl}(N, \mathbb{C})$, so that it adds one dimension to the Cartan subalgebra but contributes no roots.  Thus for $\SU(N)$, just as in the unitary case, we find that $\Delta_\gog = \Delta$ and $W = S_N$. The calculation then proceeds exactly as above for $\U(N)$.

\subsection{The cases $G = \SO(2N)$ and $\O(2N)$, and integrals over covering groups}
\label{sec:so2n}

In this section, we show how Theorem \ref{thm:hc} gives an HCIZ-like formula for the special orthogonal groups $\SO(2N)$, and then use this formula to obtain analogous identities for $\mathrm{Spin}(2N)$ and $\O(2N)$.  The Lie algebra of $\SO(2N)$ is $\mathfrak{so}(2N)$, the algebra of $2N$-by-$2N$ skew-symmetric real matrices.  Its complexification is $\mathfrak{g}_\C = \mathfrak{so}(2N, \mathbb{C})$, the algebra of $2N$-by-$2N$ skew-symmetric complex matrices.  The Cartan subalgebras $\tot$ and $\tot_\C$ consist respectively of skew-symmetric real and complex matrices that are block diagonal with 2-by-2 blocks.  An orthonormal basis of $\tot$ with respect to the Hilbert--Schmidt inner product $\langle \cdot, \cdot \rangle$ is given by \begin{equation} \label{eqn:so2n_cartan_basis} e_j := \frac{1}{\sqrt{2}}(E_{2j-1\textrm{ }2j} - E_{2j\textrm{ }2j-1}), \quad j = 1,\hdots, N, \end{equation} where $E_{jk}$ is the $2N$-by-$2N$ matrix with a 1 at the intersection of row $j$ and column $k$ and zeros everywhere else.

The root system of $\SO(2N)$ is $D_N$; as positive roots we may take the $N^2 - N$ matrices $e_j \pm e_k$ for $j < k$.  The Weyl group $W$ acts on the Cartan subalgebra by permuting the eigenvalues and changing an even number of their signs.  That is, $W \cong H_{N-1} \rtimes S_N,$ where $H_{N-1}$ is the normal subgroup of $(\mathbb{Z}/2\mathbb{Z})^N$ consisting of those elements with an even number of nonzero entries.  This means that each $w \in W$ can be written uniquely in the form $w = \eta \sigma$ with $\eta \in H_{N-1}$ and $\sigma \in S_N$, and $\epsilon(w) = \mathrm{sgn}(\sigma)$. Each $w \in W$ is represented in $\SO(2N)$ by a matrix of the form $H_\eta P_\sigma$, where $P_\sigma$ represents $\sigma \in S_N$ as a block permutation matrix with 2-by-2 blocks, and $H_\eta$ represents $\eta \in H_{N-1}$ as a block-diagonal matrix with blocks equal to $$I_2 = \begin{bmatrix} 1 & 0 \\ 0 & 1 \end{bmatrix} \quad \textrm{or} \quad Q_2 = \begin{bmatrix} 0 & 1 \\ 1 & 0 \end{bmatrix},$$ where the number of $Q_2$ blocks is even.  The order of the Weyl group is $$|W| = |H_{N-1}| \cdot |S_N| = 2^{N-1} N!$$

Now take $A, B \in \tot_\C$ and write $A = \sum_{j=1}^N a_j e_j$, $B = \sum_{j=1}^N b_j e_j$. Let $\Delta(A) = \prod_{j < k} (a_j - a_k)$ be the Vandermonde determinant, and define \begin{equation} \label{eqn:Xi_def} \Xi(A) := \prod_{j < k} (a_j + a_k). \end{equation}  Then $\Delta_\gog(A) = \Xi(A) \Delta(A)$.  Plugging all of this into (\ref{eqn:hc}), we have:
\begin{multline} \label{eqn:partial_hc_so2n}
\Xi(A) \Delta(A) \Xi(B) \Delta(B) \int_{\SO(2N)} e^{\mathrm{tr}(AOBO^T)} dO \\ = \frac{[\![\Xi \Delta,\Xi \Delta]\!]}{2^{N-1} N!} \sum_{\eta \in H_{N-1}} \sum_{\sigma \in S_N} \mathrm{sgn}(\sigma) e^{\mathrm{tr}(A H_\eta P_\sigma B P_\sigma^T H_\eta^T)}.
\end{multline}

There are two steps remaining: to simplify the sum over the Weyl group, and to determine the leading constant that gives the normalization.  We first turn to simplifying the sum.  Here too there is a determinantal structure under the surface, albeit a more complicated one than in the unitary case.  Since the trace of a product of matrices is invariant under cyclic permutations of the product, we have $$e^{\mathrm{tr}(A H_\eta P_\sigma B P_\sigma^T H_\eta^T)} = e^{\mathrm{tr}(P_\sigma^T H_\eta^T A H_\eta P_\sigma B)} = e^{2 \sum_{j=1}^N \eta^{-1}(j) a_{\sigma^{-1}(j)} b_j},$$ where we consider $\eta$ as a function $\eta : \{1, \hdots, N\} \to \{\pm 1\}$.  We can write $\mathrm{sgn}(\eta) = \prod_{j=1}^N \eta(j)$. Then $H_{N-1} = \{ \eta \in (\mathbb{Z}/2\mathbb{Z})^N \ | \ \mathrm{sgn}(\eta) = 1 \}$.

Consider the map $A \mapsto H_\eta A H_\eta^T.$  The nonzero entries of $H_\eta$ are contained in $N$ 2-by-2 blocks $L_1, \hdots, L_N$ along the diagonal, with $$L_j = \begin{cases} I_2, & \eta(j) = 1 \\ Q_2, & \eta(j) = -1. \end{cases}$$  Under conjugation by $H_\eta$, each block $L_j$ sends $a_j \mapsto \eta(j) a_j$, so we have
\begin{equation} \label{eqn:hn-1_sum} \sum_{\eta \in H_{N-1}} \sum_{\sigma \in S_N} \mathrm{sgn}(\sigma) e^{\mathrm{tr}(A H_\eta P_\sigma B P_\sigma^T H_\eta^T)} = \sum_{\sigma \in S_N} \mathrm{sgn}(\sigma) \sum_{\eta \in H_{N-1}} e^{2 \sum_{j = 1}^N \eta(j) a_{\sigma(j)} b_j}.\end{equation}

We can rearrange terms to make this expression more concise, by noting that $$\sum_{\eta \in H_{N-1}} e^{2 \sum_{j = 1}^N \eta(j) a_{\sigma(j)} b_j} = \frac{1}{2}\sum_{\eta \in (\mathbb{Z}/2\mathbb{Z})^N} \left[ e^{2 \sum_{j = 1}^N \eta(j) a_{\sigma(j)} b_j} + \mathrm{sgn}(\eta) e^{2 \sum_{j = 1}^N \eta(j) a_{\sigma(j)} b_j} \right],$$
so that the sum over $H_{N-1}$ can be written as an average of two sums over $(\mathbb{Z}/2\mathbb{Z})^N$, one signed and one unsigned.  Each of these sums can be factored into a product of hyperbolic sines or cosines:
\begin{equation} \label{eqn:eta_sum} \sum_{\eta \in (\mathbb{Z}/2\mathbb{Z})^N} e^{2 \sum_{j = 1}^N \eta(j) a_{\sigma(j)} b_j} = \prod_{j = 1}^N ( e^{2a_{\sigma(j)}b_j} + e^{2a_{\sigma(j)}b_j} ) = 2^N \prod_{j = 1}^N \mathrm{cosh}(2a_{\sigma(j)} b_j), \end{equation}
\begin{equation} \label{eqn:eta_signed_sum} \sum_{\eta \in (\mathbb{Z}/2\mathbb{Z})^N} \mathrm{sgn}(\eta) e^{2 \sum_{j = 1}^N \eta(j) a_{\sigma(j)} b_j} = \prod_{j = 1}^N ( e^{2a_{\sigma(j)}b_j} - e^{2a_{\sigma(j)}b_j} ) = 2^N \prod_{j = 1}^N \mathrm{sinh}(2a_{\sigma(j)} b_j). \end{equation}

These identities finally give: \begin{multline*} \sum_{\eta \in H_{N-1}} \sum_{\sigma \in S_N} \mathrm{sgn}(\sigma) e^{\mathrm{tr}(A H_\eta P_\sigma B P_\sigma^T H_\eta^T)} \\ = \sum_{\sigma \in S_N} \mathrm{sgn}(\sigma) \frac{1}{2} \left[2^N \prod_{j = 1}^N \mathrm{cosh}(2a_{\sigma(j)} b_j) + 2^N \prod_{j = 1}^N \mathrm{sinh}(2a_{\sigma(j)} b_j)\right] \\ = 2^{N-1} \big( \det \big[ \mathrm{cosh}(2a_j b_k) \big]_{j,k = 1}^N + \det \big[ \mathrm{sinh}(2a_j b_k) \big]_{j,k = 1}^N \big). \end{multline*}

Next we turn to evaluating $[\![\Xi \Delta,\Xi \Delta]\!]$.  We will find a way of rewriting the polynomial $\Xi \Delta$ that also allows us to simplify the expression $\Xi(A) \Delta(A) \Xi(B) \Delta(B)$ that appears in (\ref{eqn:partial_hc_so2n}).  To that end, we introduce a last piece of notation.  Given $A = \sum a_j e_j \in \tot_\C$, define $$A^{(2)} := \sum a_j^2 e_j \in \tot_\C.$$  Note that $A^{(2)}$ is in general not equal to $A^2$.  Similarly, given a polynomial $p(x) = \sum_\beta c_\beta x^\beta$, denote by $p(\partial^2)$ the differential operator $$p(\partial^2) := \sum_\beta c_\beta \frac{\partial^{2|\beta|}}{\partial x^{2\beta}}.$$

Now assume $a_i \ne \pm a_j$, $b_i \ne \pm b_j$ for $i \ne j$, so that $\Delta_\gog(A), \Delta_\gog(B) \not = 0$.  We find that we can write $$\Xi(A) = \prod_{i < j} (a_i + a_j) = \prod_{i < j} \frac{a_i^2 - a_j^2}{a_i - a_j} = \frac{\Delta(A^{(2)})}{\Delta(A)}.$$  That is, the polynomial $\Xi$ is in fact a ratio of two Vandermondes, and so we have $$\Xi(A) \Delta(A) \Xi(B) \Delta(B) = \Delta(A^{(2)}) \Delta(B^{(2)}).$$  Moreover, $$\Xi(\partial) \Delta(\partial) (\Xi \Delta)(x) \big |_{x=0} = \Delta(\partial^2) \Delta(x^{(2)}) \big |_{x = 0}$$ and this is an expression that we already know how to evaluate using Lemma \ref{lem:normconst} and our calculation (\ref{eqn:un_norm}) for the unitary case.  With $\pi_\beta$ as in (\ref{eqn:pibeta}) and $\beta_0 := (0,\hdots,N-1)$, we can write \begin{equation} \label{eqn:delta_sq} \Delta(x^{(2)}) = \sum_{\sigma \in S_N} \pi_\beta x^{2 \sigma(\beta_0)}.\end{equation}  Then Lemma \ref{lem:normconst} gives $$\Delta(\partial^2) \Delta(x^{(2)}) \big |_{x = 0} = N! \prod_{p=1}^{N-1} (2p)!$$
Plugging all of the above results into (\ref{eqn:partial_hc_so2n}) we arrive at the desired formula, which is an analogue of the HCIZ integral for $\SO(2N)$:
\begin{equation} \label{eqn:hc_so2n}
\int_{\SO(2N)} e^{\mathrm{tr}(AOBO^T)} dO = \left( \prod_{p = 1}^{N-1} (2p)! \right) \frac{\det \big[\mathrm{cosh}(2a_j b_k) \big]_{j,k = 1}^N + \det \big[\mathrm{sinh}(2a_j b_k) \big]_{j,k = 1}^N} {\Delta(A^{(2)}) \Delta(B^{(2)})} .
\end{equation}

We now move on to the calculations for $\mathrm{Spin}(2N)$ and $\O(2N)$.  The group $\mathrm{Spin}(2N)$ is a double cover of $\SO(2N)$, so that it has the same Lie algebra $\mathfrak{so}(2N)$. Since the right-hand side of (\ref{eqn:hc}) depends only on the algebra, we can immediately conclude that the formula for $\SO(2N)$ applies exactly as written: \begin{equation} \label{eqn:hc_spin2n} \int_{\mathrm{Spin}(2N)} e^{\langle \mathrm{Ad}_g A, B \rangle} dg = \int_{\SO(2N)} e^{\mathrm{tr}(AOBO^T)} dO. \end{equation}  

In general, if $\tilde G$ is a connected compact covering group of $G$, then the integral formula for $G$ will hold as written for $\tilde G$ as well.  This equality is a consequence of the following proposition, which we will use again below when we consider integrals over arbitrary compact groups in Theorem \ref{thm:hc_generalization}.  Note that if $\tilde G$ is a covering group of $G$ then their Lie algebras are isomorphic.  The statement below follows from the observation that, moreover, $\tilde G$ and $G$ have the same adjoint orbits.

\begin{proposition} \label{prop:cover_integrals} Let $G$ be a compact, connected Lie group, not necessarily semisimple, with Lie algebra $\mathfrak{g}$.  Let $\tilde G$ be a compact, connected covering group of $G$, and let $f: \mathfrak{g}_\C \to \C$.  Then for $x \in \mathfrak{g}_\C$, \begin{equation} \label{eqn:cover_int} \int_{\tilde G} f(\mathrm{Ad}_{\tilde g} x) \ d\tilde g = \int_G f(\mathrm{Ad}_g x) \ dg, \end{equation} where $d\tilde g$ and $dg$ are the normalized Haar measures. \end{proposition}

\begin{proof}
Let $\pi: \tilde G \to G$ be the covering homomorphism. The exponential maps $\mathrm{exp}_{\tilde G}$ and $\mathrm{exp}_G$ are surjective since both groups are compact and connected, so we can write $\tilde g \in \tilde G$ as $\tilde g = \mathrm{exp}_{\tilde G}(y)$ for some $y \in \mathfrak{g}$.  We then have $\pi(\tilde g) = \mathrm{exp}_{G}(y)$, so that $$\mathrm{Ad}_{\tilde g} = \sum_{k = 0}^\infty \frac{1}{k!} \mathrm{ad}_y^k = \mathrm{Ad}_{\pi(\tilde g)}.$$ Thus $x \in \mathfrak{g}_\C$ has the same orbit $\mathcal{O}_x$ under both adjoint actions.  The pushforwards of $d\tilde g$ and $dg$ to $\mathcal{O}_x$ are normalized measures that are invariant under either adjoint action, so they must be equal, and writing both sides of (\ref{eqn:cover_int}) as an integral over $\mathcal{O}_x$ gives the desired equality.
\end{proof}

In the case $G = \O(2N)$, there is an additional subtlety: even though $\O(2N)$ is compact and has the same Lie algebra as $\SO(2N)$, it is not connected, so that we cannot use (\ref{eqn:hc}) directly.  Instead we must use a trick that allows us to reduce an integral over a non-connected group to an integral over its identity component, which we will revisit in the proof of Theorem \ref{thm:hc_generalization} below.  We have $$\O(2N) = \SO(2N) \sqcup \O^-(2N),$$ where $\O^-(2N)$ is the orientation-reversing component consisting of $2N$-by-$2N$ orthogonal matrices with determinant $-1$.  Notice that every matrix in $\O^-(2N)$ can be written as $\tilde I O$, where $O \in \SO(2N)$ and $\tilde I$ is the matrix with $\tilde I_{11} = -1$, $\tilde I_{ii} = 1$ for $i > 1$, and $\tilde I_{ij} = 0$ for $i \not = j$.   Thus \begin{multline*} \int_{\O(2N)} e^{\mathrm{tr}(AOBO^T)} dO \\ = \int_{\SO(2N)} e^{\mathrm{tr}(AOBO^T)} dO + \int_{O^-(2N)} e^{\mathrm{tr}(AOBO^T)} dO \\ = \int_{\SO(2N)} \left[ e^{\mathrm{tr}(AOBO^T)} + e^{\mathrm{tr}(A\tilde I OBO^T \tilde I^T)} \right ] dO.\end{multline*} Here we take the Haar measure $dO$ to be normalized over all of $\O(2N)$, so that in this case $\int_{\SO(2N)} dO = 1/2$ and the integration measure on $\SO(2N)$ differs from that used in (\ref{eqn:hc_so2n}) by a factor of 2.

Next, notice that by the invariance of the trace under cyclic permutations and the fact that $\tilde I^T = \tilde I$, we have $$e^{\mathrm{tr}(A\tilde I OBO^T \tilde I^T)} = e^{\mathrm{tr}(\tilde I A \tilde I OBO^T)}.$$  Conjugation of $A$ by $\tilde I$ has the effect of sending $a_1 \mapsto -a_1$.  This has the consequence that for the orientation-reversing component, the sum over the Weyl group in (\ref{eqn:hn-1_sum}) becomes a sum over $-\eta$ rather than $\eta$.  Therefore in order to integrate over both components of $\O(2N)$, we should simply take $\eta$ in (\ref{eqn:hn-1_sum}) to run over all of $(\mathbb{Z}/2\mathbb{Z})^N$ rather than merely over $H_{N-1}$, and then divide by 2 to account for the different normalization.  Another application of (\ref{eqn:eta_sum}) then gives the desired integral formula for $\O(2N)$:
\begin{equation} \label{eqn:hc_o2n}
\int_{\O(2N)} e^{\mathrm{tr}(AOBO^T)} dO = \left( \prod_{p = 1}^{N-1} (2p)! \right) \frac{\det \big[\mathrm{cosh}(2a_j b_k) \big]_{j,k = 1}^N} {\Delta(A^{(2)}) \Delta(B^{(2)})}.
\end{equation}

\subsection{The cases $G = \SO(2N+1)$ and $\O(2N+1)$}
\label{sec:so2n1}

The orthogonal groups $\SO(2N+1)$ and $\O(2N+1)$ need to be treated separately from $\SO(2N)$ and $\O(2N)$, as their root system $B_N$ differs from the system $D_N$ of $\SO(2N)$.

We start with $G = \SO(2N+1)$.  The Lie algebra of $\SO(2N+1)$ is $\mathfrak{g} = \mathfrak{so}(2N+1)$, the algebra of $(2N+1)$-by-$(2N+1)$ skew-symmetric real matrices, with complexification $\mathfrak{g}_\C = \mathfrak{so}(2N+1, \mathbb{C})$, the algebra of $(2N+1)$-by-\mbox{$(2N+1)$} skew-symmetric complex matrices.  We take the Cartan subalgebra $\tot_\C$ to be the skew-symmetric block-diagonal complex matrices whose diagonal consists of $N$ 2-by-2 blocks followed by a single zero in the bottom right corner; then $\tot$ consists of those matrices in $\tot_\C$ that have strictly real entries.  An orthonormal basis of $\tot$ with respect to the Hilbert--Schmidt inner product is given by $$e_j := \frac{1}{\sqrt{2}}(E_{2j-1\textrm{ }2j} - E_{2j\textrm{ }2j-1}), \quad j = 1,\hdots, N,$$ which is nearly identical to (\ref{eqn:so2n_cartan_basis}), except that now $E_{jk}$ is the $(2N+1)$-by-$(2N+1)$ matrix with a 1 at the intersection of row $j$ and column $k$ and zeros everywhere else.

The difference from the case $G=\SO(2N)$ arises because $B_N$ contains additional roots of a different length: we must consider both the $N^2 -N$ {\it long} positive roots $e_j \pm e_k$ for $j < k$, corresponding to the roots of $D_N$, and also the $N$ {\it short} positive roots $e_j$, $1 \le j \le N$.

Because of these additional roots, the $B_N$ Weyl group is bigger than that of $D_N$.  We have $$W \cong (\mathbb{Z}/2\mathbb{Z})^N \rtimes S_N,$$ so that $|W| = |(\mathbb{Z}/2\mathbb{Z})^N| \cdot |S_N| = 2^N N!$.  Like before, this means that each $w \in W$ can be written uniquely in the form $w = \eta \sigma$, with $\sigma \in S_N$ and $\eta \in (\mathbb{Z}/2\mathbb{Z})^N$.  Since $\eta$ can now have a negative sign, we have $\epsilon(w) = \mathrm{sgn}(\eta)\mathrm{sgn}(\sigma)$. Each element of $W$ is then represented in $\SO(2N+1)$ by a matrix of the form $H_\eta P_\sigma$, where $P_\sigma$ represents $\sigma \in S_N$ as a block permutation matrix with 2-by-2 blocks followed by a final $1$ in the bottom right corner, and $H_\eta$ represents $\eta \in (\mathbb{Z}/2\mathbb{Z})^N$ as a block-diagonal matrix with $N$ 2-by-2 blocks equal to $$I_2 = \begin{bmatrix} 1 & 0 \\ 0 & 1 \end{bmatrix} \quad \textrm{or} \quad Q_2 = \begin{bmatrix} 0 & 1 \\ 1 & 0 \end{bmatrix},$$ followed by a final $\pm 1$ in the bottom right corner chosen so that $\det H_\eta = 1$.  Because we can choose this final entry to ensure a unit determinant, unlike in the case of $\SO(2N)$ we do not need to require the number of $Q_2$ blocks to be even.

Let $A = \sum_{j=1}^N a_j e_j$, $B = \sum_{j=1}^N b_j e_j \in \tot_\C$. Let $\Delta(A) = \prod_{j < k} (a_j - a_k)$ be the Vandermonde, and define $\Xi$ as in (\ref{eqn:Xi_def}). Additionally, define $\Phi(A) := \prod_{j=1}^N a_j$.  Then $\Delta_{\gog} = \Xi \Delta \Phi$. Plugging into (\ref{eqn:hc}), we have:
\begin{multline} \label{eqn:partial_hc_so2n+1}
\Xi(A) \Delta(A) \Phi(A) \Xi(B) \Delta(B) \Phi(B) \int_{\SO(2N+1)} e^{\mathrm{tr}(AOBO^T)} dO \\ = \frac{[\![\Xi \Delta \Phi,\Xi \Delta \Phi]\!]}{2^N N!} \sum_{\eta \in (\mathbb{Z}/2\mathbb{Z})^N} \sum_{\sigma \in S_N} \mathrm{sgn}(\eta)\mathrm{sgn}(\sigma) e^{\mathrm{tr}(A H_\eta P_\sigma B P_\sigma^T H_\eta^T)}.
\end{multline}

Proceeding as before, we next turn to simplifying the sum over the Weyl group.  Again we can write $$e^{\mathrm{tr}(A H_\eta P_\sigma B P_\sigma^T H_\eta^T)} = e^{\mathrm{tr}(P_\sigma^T H_\eta A H_\eta^T P_\sigma B)},$$ and the map $A \mapsto H_\eta A H_\eta^T$ acts on $\tot_\C$ in a similar way to the previous case.  The nonzero entries of $H_\eta$ are contained in $N$ diagonal blocks $L_1, \hdots, L_N$, each of which is equal to one of the 2-by-2 matrices $I_2$ or $Q_2$, plus the final $\pm 1$ on the diagonal, which we ignore since the final row and column of any matrix in $\tot_\C$ are all zeros.  Just like before, under conjugation by $H_\eta$, the block $L_j$ sends $a_j \mapsto \eta(j) a_j$.  We apply (\ref{eqn:eta_signed_sum}) to conclude that
\begin{align*}
\sum_{\eta \in (\mathbb{Z}/2\mathbb{Z})^N} \sum_{\sigma \in S_N} \mathrm{sgn}(\eta) \mathrm{sgn}(\sigma) e^{\mathrm{tr}(A H_\eta P_\sigma B P_\sigma^T H_\eta^T)} &= \sum_{\sigma \in S_N} \mathrm{sgn}(\eta) 2^N \prod_{j=1}^N \mathrm{sinh}(2 a_{\sigma(j)} b_j) \\
& = 2^N \det\big[\mathrm{sinh}(2 a_j b_k)\big]_{j,k=1}^N.
\end{align*}

It now remains to compute the constant $[\![\Xi \Delta \Phi,\Xi \Delta \Phi]\!]$. With reference to (\ref{eqn:delta_sq}), we can write $$\Xi(x) \Delta(x) \Phi(x) = \Delta(x^{(2)}) \Phi(x) = \sum_{\sigma \in S_N} \mathrm{sgn}(\sigma) x^{2 \sigma(\beta_0) + \vec 1}$$ where $\vec 1$ is the $N$-component multi-index $(1, 1, \hdots, 1)$.  Lemma \ref{lem:normconst} then gives $$\Xi(\partial) \Delta(\partial) \Phi(\partial) \big [ \Xi(x) \Delta(x) \Phi(x) \big] \Big|_{x=0} = \sum_{\sigma \in S_N} \left( \mathrm{sgn}(\sigma)^2 \prod_{p=0}^{N-1} (2 \sigma(p) + 1)! \right)$$$$ = N! \prod_{p = 1}^{N-1} (2p + 1)!$$

We substitute this into (\ref{eqn:partial_hc_so2n+1}) under the assumptions that $a_i \ne \pm a_j$, $b_i \ne \pm b_j$ for $i \ne j$ and that all $a_i$, $b_i$ are nonzero, so that we can divide by $\Delta_\gog(A) \Delta_\gog(B)$.  This yields the integral formula for $\SO(2N+1)$:
\begin{multline} \label{eqn:hc_so2n+1}
\int_{\SO(2N+1)} e^{\mathrm{tr}(AOBO^T)} dO = \left( \prod_{p = 1}^{N-1} (2p+1)! \right) \frac{\det\big[\mathrm{sinh}(2 a_j b_k)\big]_{j,k=1}^N}{\Delta(A^{(2)}) \Delta(B^{(2)}) \prod_{i=1}^N a_i b_i}.
\end{multline}
Proposition \ref{prop:cover_integrals} yields a formula for $\mathrm{Spin}(2N+1)$ that is analogous to (\ref{eqn:hc_spin2n}).

For the case $G = \O(2N+1)$, we must compute the integral over the orientation-reversing component.  Applying the same logic as in the case of $\O(2N)$, we find $$\int_{\O(2N+1)} e^{\mathrm{tr}(AOBO^T)} dO = \int_{\SO(2N+1)} e^{\mathrm{tr}(AOBO^T)} dO + \int_{\O^-(2N+1)} e^{\mathrm{tr}(AOBO^T)} dO$$$$ = \int_{\SO(2N+1)} \left[ e^{\mathrm{tr}(AOBO^T)} +  e^{\mathrm{tr}(\tilde I A \tilde I OBO^T)} \right] dO,$$ where the measure $dO$ is now taken to be normalized over the entire group $\O(2N+1)$.  The same argument used in the calculation for $\O(2N)$ shows that integrating over the orientation-reversing component amounts to sending $\eta \mapsto -\eta$ in the sum over the Weyl group on the right-hand side of (\ref{eqn:hc}).  But since in this case we are summing over all of $(\mathbb{Z}/2\mathbb{Z})^N$, this does not actually change the sum, so that the integrals over both connected components of $\O(2N+1)$ are equal.  Since $dO$ is normalized over the entire group, each of these two integrals is equal to half the integral in (\ref{eqn:hc_so2n+1}), and the final formula for $\O(2N+1)$ is identical to the formula for $\SO(2N+1)$:
\begin{multline} \label{eqn:hc_o2n+1}
\int_{\O(2N+1)} e^{\mathrm{tr}(AOBO^T)} dO = \left( \prod_{p = 1}^{N-1} (2p+1)! \right) \frac{\det\big[\mathrm{sinh}(2 a_j b_k)\big]_{j,k=1}^N}{\Delta(A^{(2)}) \Delta(B^{(2)}) \prod_{i=1}^N a_i b_i}.
\end{multline}

\subsection{The case $G = \mathrm{USp}(N)$}
\label{sec:uspn}

The unitary symplectic, or compact symplectic, group $\mathrm{USp}(N)$ is the compact real form of $\mathrm{Sp}(2N, \mathbb{C})$, the symplectic group over the complex numbers.  The group $\mathrm{Sp}(2N, \mathbb{C})$ consists of $2N$-by-$2N$ complex matrices $S$ satisfying $$SJS^T = J,$$ where $$J = \begin{bmatrix} 0 & I_N \\ -I_N & 0 \end{bmatrix}$$ is the standard symplectic matrix. (Here $I_N$ indicates the $N$-by-$N$ identity matrix.) The name ``unitary symplectic'' comes from the fact that
\begin{equation} \label{eqn:unitary-symplectic}
\mathrm{USp}(N) = \U(2N) \cap \mathrm{Sp}(2N, \mathbb{C}) = \left\{ \begin{bmatrix} A & - \bar B \\ B & \bar A \end{bmatrix} \in \U(2N) \right\}.
\end{equation}
In fact, $\mathrm{USp}(N)$ is isomorphic to the quaternionic unitary (or ``hyperunitary'') group $\U(N, \mathbb{H})$ of linear operators on $\mathbb{H}^N$ that preserve the standard Hermitian form $$( v, w ) := \sum_{i=1}^N \bar v_i w_i, \qquad v, w \in \mathbb{H}^N.$$  Thus we could equivalently consider $\mathrm{USp}(N)$ as a subgroup of $\GL(N, \mathbb{H})$, though here we regard it as a subgroup of $\GL(2N, \mathbb{C})$.

The complexified Lie algebra of $\mathrm{USp}(N)$ is $$\mathfrak{g}_\C = \mathfrak{sp}(2N, \mathbb{C}) = \left\{ \begin{bmatrix} A & B \\ C & -A^T \end{bmatrix} \in \mathfrak{gl}(2n, \mathbb{C}) \textrm{  } \bigg | \textrm{  } B^T = B, \textrm{  } C^T = C \right\}.$$ From (\ref{eqn:unitary-symplectic}) we then have $\gog = \gog_\C \cap \mathfrak{u}(2N),$ so that $\gog$ consists of the matrices in $\mathfrak{sp}(2N, \mathbb{C})$ that are also skew-Hermitian. We take $\tot$ to be the Cartan subalgebra of diagonal matrices in $\gog$.  It is spanned by the orthonormal basis $$e_j := \frac{i}{\sqrt{2}}( E_{jj} - E_{j+N \textrm{ } j+N}), \quad 1 \le j \le N,$$ so that $\mathrm{USp}(N)$ has rank $N$.  The root system in this case is $C_N$, which is similar to the root system $B_N$, but with the short roots multiplied by 2 (so that for $C_N$ these are the long roots).  That is, in the basis $\{ e_j \}_{j=1}^N$ of $\tot$, the $N$ long positive roots are given by $2 e_j$ for $1 \le j \le N$ and the $N^2 - N$ short positive roots are given by $e_j \pm e_k$ for $j < k$.  Thus for the $C_N$ root system, \begin{equation} \label{eqn:uspn_poly} \Delta_\gog(x) = 2^N \Xi(x) \Delta(x) \Phi(x), \end{equation} which differs from the $B_N$ case only by a factor of $2^N$.  Our previous calculation for the $B_N$ system therefore immediately gives $$\Delta_\gog(\partial)\Delta_\gog(x) \big |_{x=0} = 2^{2N} N! \prod_{p=1}^{N-1} (2p + 1)!$$

The Weyl group for $C_N$ is isomorphic to the group for $B_N$: $$W \cong (\mathbb{Z}/2\mathbb{Z})^N \rtimes S_N,$$ with $|W| = 2^N N!$, and it acts on $\tot$ in the same fashion.  Thus if we let $A = \sum_{j=1}^N a_j e_j$, $B = \sum_{j=1}^N b_j e_j$ and take $x = A^\dagger$, $y = B$ in (\ref{eqn:hc}), then the sum over the Weyl group on the right-hand side is equal to $$\sum_{\eta \in (\mathbb{Z}/2\mathbb{Z})^N} \sum_{\sigma \in S_N} \mathrm{sgn}(\eta) \mathrm{sgn}(\sigma) e^{\mathrm{tr}(A H_\eta P_\sigma B P_\sigma^T H_\eta^T)} = 2^N \det \big[ \mathrm{sinh}(2 a_j b_k) \big]_{j,k=1}^N$$ just as in the case of $\SO(2N+1)$.  Putting this all together in (\ref{eqn:hc}) and assuming as before that $a_i \ne \pm a_j$, $b_i \ne \pm b_j$ for $i \ne j$ and that all $a_i$, $b_i$ are nonzero, we conclude:
\begin{equation} \label{eqn:hc_uspn}
\int_{\mathrm{USp}(N)} e^{\mathrm{tr}(ASBS^\dagger)} dS \ = \ \left( \prod_{p = 1}^{N-1} (2p+1)! \right) \frac{\det \big[ \mathrm{sinh}(2 a_j b_k) \big]_{j,k=1}^N}{\Delta(A^{(2)}) \Delta(B^{(2)}) \prod_{i=1}^N a_i b_i}.
\end{equation}

Thus for the unitary symplectic integral, we end up with exactly the same expression as for the odd special orthogonal groups: the sum over the Weyl group in each case gives the same result, and although $\Delta_\gog$ differs from the $B_N$ case by a factor of $2^{2N}$, these factors cancel from either side of (\ref{eqn:hc_uspn}), so that the right-hand side matches that of (\ref{eqn:hc_so2n+1}).

\subsection{Integrals over arbitrary compact groups}
\label{sec:arbitrary_cpt}

In this section, we show how Theorem \ref{thm:hc} can be used to compute integrals over arbitrary compact Lie groups that may be neither semisimple nor connected, such as $\U(N)$ and $\O(N)$.  As a special case, this calculation justifies the derivation of the HCIZ integral (\ref{eqn:hciz}) from (\ref{eqn:hc}).  The proof makes use of the classification of compact Lie groups, as stated in the following theorem.\footnotemark \footnotetext{This theorem is a slightly weakened version of \cite[ch.~10, \textsection7.2, Theorem 4]{ProcesiLieGroups}.}

\begin{theorem} \label{thm:cp_classification}
Every compact, connected real Lie group is of the form \mbox{$(K \times H)/Z$}, where $K$ is connected, compact, and semisimple, $H \cong (S^1)^d$ is a torus, and $Z$ is a finite subgroup of the center of $K \times H$ satisfying $Z \cap H = \{\mathrm{id}_{K\times H}\}$.
\end{theorem}

The idea is that we will express an integral over an arbitrary compact group $G$ as a sum of integrals over its identity component $G_1$, which we decompose as $G_1 \cong (K \times H)/Z$ following Theorem \ref{thm:cp_classification}. We then rewrite each of these integrals over $G_1$ as an integral over the semisimple group $K$ and evaluate using (\ref{eqn:hc}).  Note that $G$, $G_1$, $K \times H$ and $K$ all have the same root system $\Phi$, so that the discriminant $\Delta_\gog$ is the same for all of these groups.  We write $W_\Phi$ for the Weyl group generated by reflections through the root hyperplanes.  Since $G$, $G_1$ and $K \times H$ have isomorphic Lie algebras, we will write $\mathfrak{g}_\C$ and $\tot_\C$ for the complexified Lie algebra and Cartan subalgebra of all three.  Define $\mathfrak{k} := \mathrm{Lie}(K) \otimes \C$ and $\mathfrak{h} := \mathrm{Lie}(H) \otimes \C$, so that $\gog_\C = \mathfrak{k} \oplus \mathfrak{h}$, and let $\tot_K$ be the Cartan subalgebra of $\mathfrak{k}$ obtained by taking the orthogonal complement of $\mathfrak{h}$ in $\tot_\C$.

\begin{remark} \label{rem:disco-weyl}
There is no standardized definition for the Weyl group of a disconnected Lie group, since the definition for connected groups can be generalized in multiple inequivalent ways.  When $G$ is compact and connected, one usually defines $W := N_G(T)/T$, where $N_G(T)$ is the normalizer of the maximal torus $T$ in $G$.  It is then a nontrivial theorem that $W$ is isomorphic to the group $W_\Phi$ generated by reflections through the hyperplanes $\{ x \in \tot \ | \ \langle \alpha, x \rangle = 0 \}$ for $\alpha$ a simple root.  However, when $G$ is disconnected, $N_G(T)/T$ is {\it not} isomorphic to $W_\Phi$ but rather contains it as a proper subgroup.  To avoid confusion, we do not define the ``Weyl group of $G$'' in the disconnected case, choosing instead to work only with the Weyl group $W_\Phi$ associated to the root system $\Phi$ (or equivalently, to the Lie \mbox{algebra $\gog$}).  See Remark \ref{rem:disco-weyl2} below for further discussion of the group $N_G(T)/T$ when $G$ is disconnected.
\remdone
\end{remark}

First we show that the torus factor $H$ does not change the right-hand side of the formula (\ref{eqn:hc}).

\begin{lemma} \label{lem:torus_factor}
For all $x, y \in \tot_\C$, $$\Delta_\gog(x) \Delta_\gog(y) \int_{K \times H} e^{\langle \mathrm{Ad}_{(k,h)} x, y \rangle} dk \, dh = \frac{ [ \! [ \Delta_\gog, \Delta_\gog ] \!] }{|W_\Phi|} \sum_{w \in W_\Phi} \epsilon(w) e^{\langle w(x),y \rangle},$$ where $dk$ and $dh$ are the normalized Haar measures on $K$ and $H$ respectively.
\end{lemma}

\begin{proof}
For $x \in \tot_\C$, write $x = x^K + x^H$, with $x^K \in \mathfrak{t}_K,$ $x^H \in \mathfrak{h}$.  All of the roots vanish on $\mathfrak{h}$, so that $\Delta_\gog(x) = \Delta_\gog(x^K)$.  Since $\{\mathrm{id}_K\} \times H$ lies in the center of $K \times H$, we have $\mathrm{Ad}_{(k,h)} x = x^H + \mathrm{Ad}_k x^K$.  Then since $\mathfrak{h}$ is orthogonal to $\mathfrak{k}$, we find
\begin{equation*} \langle \mathrm{Ad}_g x, y \rangle = \langle \mathrm{Ad}_g (x^H + x^K), y^H + y^K \rangle = \langle x^H, y^H \rangle + \langle \mathrm{Ad}_g x^K, y^K \rangle. \end{equation*}
Accordingly, we have:
\begin{align*} \Delta_\gog(x) \Delta_\gog(y) &\int_{K \times H} e^{\langle \mathrm{Ad}_{(k,h)} x, y \rangle} dk \, dh \\ &= \Delta_\gog(x^K) \Delta_\gog(y^K) \, e^{\langle x^H, \, y^H \rangle} \int_{K} \int_H e^{\langle \mathrm{Ad}_k x^K, \, y^K \rangle} dh\, dk \\ &= \Delta_\gog(x^K) \Delta_\gog(y^K) \, e^{\langle x^H, \, y^H \rangle} \int_{K} e^{\langle \mathrm{Ad}_k x^K, \, y^K \rangle} dk \\ &= \frac{ [ \! [ \Delta_\gog, \Delta_\gog ] \!] }{|W_\Phi|} \sum_{w \in W_\Phi} \epsilon(w) e^{\langle x^H, \, y^H \rangle} e^{ \langle w(x^K), \, y^K \rangle} \\ &= \frac{ [ \! [ \Delta_\gog, \Delta_\gog ] \!] }{|W_\Phi|} \sum_{w \in W_\Phi} \epsilon(w) e^{\langle w(x),y \rangle}\end{align*}
as desired.
\end{proof}

Since $K \times H$ is a compact, connected, $|Z|$-fold covering group of \mbox{$(K \times H)/Z$}, by Proposition \ref{prop:cover_integrals} we have
\begin{equation} \label{eqn:torus_factor_cover_int}
\int_{K \times H} e^{\langle \mathrm{Ad}_{(k,h)} x, y \rangle} dk \, dh = \int_{(K \times H)/Z} e^{\langle \mathrm{Ad}_\eta x, y \rangle} \, d\eta,
\end{equation}
where $d\eta$ is the normalized Haar measure on $(K \times H)/Z$.
We now have all the tools to extend Theorem \ref{thm:hc} to arbitrary compact $G$.

\begin{theorem}
\label{thm:hc_generalization}
Let $G$ be a compact real Lie group with root system $\Phi$, normalized Haar measure $dg$, and $m$ connected components $G_1, \hdots, G_m$.  For $j = 1, \hdots, m$, let $g_j \in G_j$.  Then for all $x, y \in \tot_\C$,
\begin{equation} \label{eqn:hc_generalization}
\Delta_\gog(x) \Delta_\gog(y) \int_G e^{\langle \mathrm{Ad}_g x, y \rangle} dg = \frac{1}{m} \frac{ [ \! [ \Delta_\gog, \Delta_\gog ] \!] }{|W_\Phi|} \sum_{j=1}^m \sum_{w \in W_\Phi} \epsilon(w) e^{\langle w(\mathrm{Ad}_{g_j} x), \, y \rangle}.
\end{equation}
\end{theorem}

\begin{proof}
Take $G_1$ to be the identity component of $G$. First we note that \begin{align} \nonumber \int_G e^{\langle \mathrm{Ad}_g x, y \rangle} dg &= \sum_{j=1}^m \int_{G_j} e^{\langle \mathrm{Ad}_{g} x, y \rangle} dg \\ \label{eqn:component_decomp} &= \sum_{j=1}^m \int_{G_1} e^{\langle \mathrm{Ad}_{gg_j} x, y \rangle} dg = \frac{1}{m} \sum_{j=1}^m \int_{G_1} e^{\langle \mathrm{Ad}_{g}(\mathrm{Ad}_{g_j} x), y \rangle} dg_1, \end{align} where $dg_1$ is the Haar measure normalized so that the identity component $G_1$, rather than the entire group $G$, has unit volume.  By Theorem \ref{thm:cp_classification} we can write $G_1 = (K \times H)/Z$ as above.  Applying (\ref{eqn:torus_factor_cover_int}) and then Lemma \ref{lem:torus_factor} to the final expression in (\ref{eqn:component_decomp}) completes the proof of the theorem.
\end{proof}

\begin{remark} \label{rem:disco-weyl2}
Suppose we were to define $W := N_G(T)/T$, irrespective of whether or not $G$ is connected.  As in Theorem \ref{thm:hc_generalization}, let $G_1, \hdots, G_m$ be the connected components of $G$. By a simple argument using Cartan's theorem on the conjugacy of maximal tori, one can show that $N_G(T)$ intersects each component $G_j$, and that $|W| = m|W_\Phi|$.  Since $T$ acts trivially on $\tot$, the adjoint action of $N_G(T)$ on $\tot$ descends to an action of $W$.  If we choose each $g_j$ in Theorem \ref{thm:hc_generalization} to lie in $G_j \cap N_G(T)$, then each element of $W$ acts on $\tot$ as $w \circ \mathrm{Ad}_{g_j}$ for some unique $g_j$ and $w \in W_\Phi$.  If we then extend the sign function $\epsilon$ from $W_\Phi$ to $W$ by defining $\epsilon(w \circ \mathrm{Ad}_{g_j} ) = \epsilon(w)$, the double sum in Theorem \ref{thm:hc_generalization} can be rewritten as a single sum over $W$, so that the right-hand side of (\ref{eqn:hc_generalization}) looks identical to the right-hand side of (\ref{eqn:hc}).  Thus there is a certain sense in which (\ref{eqn:hc}) holds ``as written'' for any compact Lie group, though for disconnected groups the formulation (\ref{eqn:hc_generalization}) is arguably more transparent.
\remdone
\end{remark}

\section*{Acknowledgements}
This work was partially supported by JST CREST program JPMJCR18T6, as well as by the National Science Foundation under grants DMS 1411278 and DMS 1714187.  The author would like to thank Govind Menon and Jean-Bernard Zuber for their mentorship, as well as an anonymous referee for helpful comments.

\bibliography{hc-refs}
\bibliographystyle{plain}

\end{document}